\documentclass[11pt,english]{article}
\usepackage{amsmath,amssymb,amsthm}
\usepackage{multicol}
\usepackage[margin=1in]{geometry}
\usepackage{graphicx,color}
\usepackage{enumitem}
\usepackage{fullpage}
\usepackage[noblocks]{authblk}
\usepackage{tcolorbox}
\usepackage{babel}
\usepackage{wrapfig}
\usepackage{MnSymbol}
\usepackage{mdframed}
\usepackage{mathtools}
\usepackage{floatrow}
\usepackage{multirow}
\usepackage[leftcaption]{sidecap}

\usepackage[ruled,linesnumbered,vlined]{algorithm2e}
\usepackage{tablefootnote}
\usepackage{thm-restate}
\usepackage{bbding}
\usepackage{pifont}
\usepackage{nicefrac}

\definecolor{Darkblue}{rgb}{0,0,0.4}
\definecolor{Brown}{cmyk}{0,0.61,1.,0.60}
\definecolor{Purple}{cmyk}{0.45,0.86,0,0}
\definecolor{Darkgreen}{rgb}{0.133,0.543,0.133}

\usepackage[colorlinks,linkcolor=Darkblue,filecolor=blue,citecolor=blue,urlcolor=Darkblue,pagebackref]{hyperref}
\usepackage[nameinlink]{cleveref}

\usepackage[colorinlistoftodos,prependcaption,textsize=tiny]{todonotes}

\newcommand{\red}[1]{{\color{red}#1}}
\newcommand{\redno}[1]{{#1}}

\newcommand{\hide}[1]{}

\newif\ifdraft
\draftfalse

\newcommand{\namedref}[2]{\hyperref[#2]{#1~\ref*{#2}}}
\newcommand{\propref}[1]{\hyperref[#1]{property~(\ref*{#1})}}
\newcommand{\Propref}[1]{\hyperref[#1]{Property~(\ref*{#1})}}

\newcommand{\questionref}[1]{\namedref{Question~}{#1}}

\newtheorem{theorem}{Theorem}
\newtheorem{lemma}{Lemma}
\newtheorem{definition}{Definition}
\newtheorem{claim}{Claim}
\newtheorem{observation}{Observation}
\newtheorem{corollary}{Corollary}

\newtheorem{question}{Question}

\newtheorem{remark}{Remark}

\newtheorem{innercustomthm}{Theorem}
\newenvironment{customthm}[1]
{\renewcommand\theinnercustomthm{#1}\innercustomthm}
{\endinnercustomthm}

\newcommand{\poly}{\mathrm{poly}}
\newcommand{\polylog}{\mathrm{polylog}}

\newcommand{\R}{\mathbb{R}}
\newcommand{\N}{\mathbb{N}}

\newcommand{\A}{\mathcal{A}}

\newcommand{\opt}{\mathrm{OPT}}

\newcommand{\supp}{\mathrm{supp}}
\newcommand{\home}{\mbox{\bf home}}
\newcommand{\diam}{\mathrm{diam}}
\newcommand{\dm}{\mathrm{diam}}

\newcommand{\hop}{\mathrm{hop}}
\newcommand{\rhop}[1]{^{{\color{red}(#1)}}}
\newcommand{\DO}{\mathsf{DO}}

\newcommand{\QED}{\\{\color{white}.}\hfill\qedsymbol}

\newcommand{\cA}{\mathcal{A}}

\newcommand{\cD}{\mathcal{D}}

\newcommand{\cK}{\mathcal{K}}

\newcommand{\cM}{\mathcal{M}}

\newcommand{\Geo}{\mathsf{Geo}}
\newcommand{\inter}{\mathsf{int}}

\newcommand{\logdiam}{\phi}

\definecolor{forestgreen}{rgb}{0.13, 0.55, 0.13}

\SetCommentSty{mycommfont}

\def\eps{\epsilon}

\DeclareMathAlphabet{\mathpzc}{OT1}{pzc}{m}{it}
\newcommand{\etal}{{\em et al. \xspace}}

\newlength{\dhatheight}

\newcommand {\ignore} [1] {}

\newcommand{\initOneLiners}{%
	\setlength{\itemsep}{0pt}
	\setlength{\parsep }{0pt}
	\setlength{\topsep }{0pt}
}

\title{Hop-Constrained Metric Embeddings and their Applications\thanks{An extended abstract of this paper appeared in the proceedings of FOCS 21 \cite{Fil21}. This short preliminary includes only the introduction, and statements of the results.}}
\author{Arnold Filtser
	\\Bar-Ilan University
	\\ 	Email: \texttt{arnold273@gmail.com} }

\date{}
\begin{document}
\maketitle
\begin{abstract}
In network design problems, such as compact routing, the goal is to route
packets between nodes using the (approximated) shortest paths.
A desirable property of these routes is a small number of hops, which makes
them more reliable, and reduces the transmission costs.

Following the overwhelming success of stochastic tree embeddings for
algorithmic design, Haeupler, Hershkowitz, and Zuzic (STOC'21) studied
hop-constrained Ramsey-type metric embeddings into trees. Specifically,
embedding $f:G(V,E)\rightarrow T$ has Ramsey hop-distortion $(t,M,\beta,h)$
(here $t,\beta,h\ge1$ and $M\subseteq V$) if $\forall u\in M,v\in V$,
$d_G^{(\beta\cdot h)}(u,v)\le d_T(u,v)\le t\cdot d_G^{(h)}(u,v)$. $t$ is
called the distortion, $\beta$ is called the hop-stretch, and $d_G^{(h)}(u,v)$
denotes the minimum weight of a $u-v$ path with at most $h$ hops.
Haeupler {\em et al.} constructed embedding where $M$ contains $1-\epsilon$
fraction of the vertices and $\beta=t=O(\frac{\log^2 n}{\epsilon})$. They used
their embedding to obtain multiple bicriteria approximation algorithms for
hop-constrained network design problems.

In this paper, we first improve the Ramsey-type embedding to obtain parameters
$t=\beta=\frac{\tilde{O}(\log n)}{\epsilon}$, and generalize it to arbitrary
distortion parameter $t$ (in the cost of reducing the size of $M$).
This embedding immediately implies polynomial improvements for all the
approximation algorithms from Haeupler {\em et al.}.
Further, we construct hop-constrained clan embeddings (where each vertex has
multiple copies), and use them to construct bicriteria approximation
algorithms for the group Steiner tree problem, matching the state of the art
of the non-constrained version.
Finally, we use our embedding results to construct hop-constrained distance
oracles, distance labeling, and most prominently, the first hop-constrained
compact routing scheme with provable guarantees. All our metric data
structures almost match the state of the art parameters of the non
hop-constrained versions.
\end{abstract}

\newpage

\addtocontents{toc}{\protect\setcounter{tocdepth}{3}}
\tableofcontents
	\pagenumbering{gobble}
    \newpage
    \pagenumbering{arabic}

    \newpage

\section{Introduction}
Low distortion metric embeddings provide a powerful algorithmic toolkit, with
applications ranging from approximation/sublinear/online/distributed
algorithms \cite{LLR95,AMS99Sketch,BCLLM18,KKMPT12} to machine learning
\cite{GKK17}, biology \cite{HBKKW03}, and vision \cite{AS03}.
A basic approach for solving a problem in a ``hard'' metric space $(X,d_X)$
is to embed the points $X$ into a ``simple'' metric space $(Y,d_Y)$ while
preserving all pairwise distances up to small multiplicative factor.
Then one solves the problem in $Y$, and pulls back the solution into $X$.

A highly desirable target space is trees, as many hard problems become easy
once the host space is a tree metric.  Fakcharoenphol, Rao, and Talwar
\cite{FRT04} (improving over \cite{AKPW95,Bar96,Bartal98}, see also
\cite{Bartal04}) showed that every $n$ point metric space could be embedded
into a distribution $\mathcal{D}$ over dominating trees with expected distortion
$O(\log n)$. Formally, $\forall x,y\in X$, $T\in\supp(\cD)$,  $d_X(x,y)\le
d_T(x,y)$, and $\mathbb{E}_{T\sim\cD}[d_T(x,y)]\le O(\log n)\cdot d_X(x,y)$.
This stochastic embedding enjoyed tremendous success and has numerous
applications. However, the distortion guarantee is only in expectation.
A different solution is Ramsey type embeddings which have a worst case
guarantee, however only w.r.t. a subset of the points. Specifically, Mendel
and Naor \cite{MN07} (see also
\cite{BFM86,BBM06,BLMN05b,NT12,BGS16,ACEFN20,Bar21}) showed that for every
parameter $k$, every $n$-point metric space contains a subset $M$ of at least
$n^{1-\frac1k}$ points, where $X$ could be embedded into a tree such that all
the distances in $M\times X$ are preserved up to an $O(k)$ multiplicative
factor. Formally, $\forall x,y\in X$,  $d_X(x,y)\le d_T(x,y)$, and $\forall
x\in M,y\in X$,  $d_T(x,y)\le O(k)\cdot d_X(x,y)$.
Finally, in order to obtain worst case guarantee w.r.t. all point pairs, the
author and Hung \cite{FL21} introduced clan embeddings into trees. Here each
point $x$ is mapped into a subset $f(x)\subseteq Y$ which is referred to as
the copies of $x$, with a special chief copy $\chi(x)\in f(x)$. For every
parameter $k$, \cite{FL21} constructed a distribution over dominating clan
embeddings such that for every pair of vertices $x,y$, some copy of $x$ is
close to the chief of $y$: $d_X(x,y)\le \min_{x'\in f(x),y'\in
f(y)}d_T(x',y')\le \min_{x'\in f(x)}d_T(x',\chi(y))\le O(k)\cdot d_X(x,y)$,
and the expected number of copies of every vertex is bounded:
$\mathbb{E}[|f(x)|]=O(n^{\frac1k})$.

In many applications, the metric space is the shortest path metric $d_G$ of a
weighted graph $G=(V,E,w)$.
In some scenarios, in addition to metric distances, there are hop constraints.
For instance, one may wish to route a packet between two nodes, using a path
with at most $\red{h}$ hops,\footnote{To facilitate the reading of the paper,
all numbers referring to hops are colored in red.} i.e. the number of edges
in the path.
Such hop constraints are desirable as each transmission causes delays, which
are non-negligible when the number of transmissions is large \cite{AT11,BF18}.
Another advantage is that low-hop routes are more reliable: if each
transmission is prone to failure with a certain probability, then low-hop
routes are much more likely to reach their destination \cite{BF18,WA88,RAJ12}.
Electricity and telecommunications distribution network configuration problems
include hop constraints that limit the maximum number of edges between a
customer and its feeder \cite{BF18}, and there are many other (practical)
network design problems with hop constraints
\cite{BCM99,BA92,GPSV03,GM03,PS03}.
Hop-constrained network approximation is often used in parallel computing
\cite{Cohen00,ASZ20}, as the number of hops governs the number of required
parallel rounds (e.g. in Dijkstra).
Finally, there is an extensive work on approximation algorithms for
connectivity problems like spanning tree and Steiner tree with hop constraints
\cite{Rav94,KP97,MRSRRH98,AFHKRS05,KLS05,KP09,HKS09,KS16}, and for some
generalizations in \cite{HHZ21}.

Given a weighted graph $G=(V,E,w)$, $d_G\rhop{h}(u,v)$ denotes the minimum
weight of a $u$-$v$ path containing at most $\red{h}$ hops.
Following the tremendous success of tree embeddings, Haeupler, Hershkowitz and
Zuzic \cite{HHZ21} suggested to study hop-constrained tree embeddings. That
is, embedding the vertex set $V$ into a tree $T$, such that $d_T(u,v)$ will
approximate the $\red{h}$-hop-constrained distance $d_G\rhop{h}(u,v)$.
Unfortunately, \cite{HHZ21} showed that $d_G\rhop{h}$ is very far from being a
metric space, and the distortion of every embedding of $d_G\rhop{h}$
is unbounded.
To overcome this issue, \cite{HHZ21} allowed a slack not only in the distance
approximation (distortion) but also in the allowed number of hops
(hop-stretch). Specifically, allowing bi-criteria approximation in the
following sense: $d_G^{\red{(\beta\cdot h)}}(u,v)\le d_T(u,v)\le t\cdot
d\rhop{h}_G(u,v)$, for some parameters $\red{\beta}$, and $t$.
However, even this relaxation is not enough as long as
$\red{\beta\cdot h}<n-1$.
To see this,  consider the unweighted $n$-path $P_n=\{v_0,\dots,v_{n-1}\}$.
Then $d_G^{\red{(n-2)}}(v_0,v_{n-1})=\infty$, implying that for every metric
$X$ over $V$,
$\max_{i}\{d_{X}(v_{i},v_{i+1})\}\ge\frac{1}{n-1}\sum_{i}d_{X}(v_{i},v_{i+1})
\ge\frac{d_{X}(v_{0},v_{n-1})}{n-1}\ge\frac{d_{G}^{\red{(n-2)}}(v_{0},v_{n-1}
)}{n-1}=\infty$.
As a result, some pairs have distortion $\infty$, and hence no bi-criteria
metric embedding is possible. In particular, there is no hop-constrained
counterpart for \cite{FRT04}.

To overcome this barrier, \cite{HHZ21} studied hop-constrained Ramsey-type
tree embeddings.\footnote{\cite{HHZ21} called this embedding \emph{partial
metric distribution}, rather than Ramsey-type.} Their main result is that for
every $n$-point weighted graph $G=(V,E,w)$ with polynomial aspect
ratio,\footnote{The \emph{aspect ratio} (sometimes referred to as spread), is
the ratio between the maximum and minimum weight in $G$ $\frac{\max_{e\in
E}w(e)}{\min_{e\in E}w(e)}$. Often in the literature the aspect ratio is
defined as $\frac{\max_{u,v\in V}d_G(u,v)}{\min_{u\ne v\in V}d_G(u,v)}$:
the ratio between the maximum and minimum distance in $G$, which is equivalent
up to a factor $n$ to the definition used here.\label{foot:aspectRatio}
}
and parameter $\eps\in(0,1)$, there is a distribution over pairs $(T,M)$,
where $M\subseteq V$, and $T$ is a tree with $M$ as its vertex set such that
\[
\forall u,v\in M,\qquad d_G^{\red{(O(\frac{\log^2n}{\eps})\cdot h)}}(u,v)\le
d_T(u,v)\le O(\frac{\log^2n}{\eps})\cdot d_G\rhop{h}(u,v)~,
\]
and for every $v\in V$, $\Pr[v\in M]\ge 1-\eps$ (referred to as inclusion
probability). Improving the distortion, and hop-stretch, was left as the
``main open question'' in \cite{HHZ21}.
\begin{question}[\cite{HHZ21}]\label{question:ImproveEmbedding}
	Is it possible to construct better hop-constrained Ramsey-type tree
	embeddings? What are the optimal parameters?
\end{question}

Adding hop constraints to a problem often makes it considerably harder. For
example, obtaining an $o(\log n)$ approximation for the hop-constrained MST
problem (where the hop-diameter must be bounded by $\red{h}$) is NP-hard
\cite{BKP01}, while MST is a simple problem without it. Similarly, for every
${\eps>0}$, obtaining an $o(2^{\log^{1-\eps}}n)$-approximation for the
hop-constrained Steiner forest problem is NP-hard \cite{DKR16} (while constant
approximation is possible without hop constraints \cite{AKR95}).

\cite{HHZ21} used their hop-constrained Ramsey-type tree embeddings to
construct many approximation algorithms.
Roughly, given that some problem has (non constrained) approximation ratio
$\alpha $ for trees, \cite{HHZ21} obtain a bicriteria approximation for the
hop-constrained problem with cost $O(\alpha\cdot \log^3n)$, and hop-stretch
$O(\log^3n)$.
For example, in the group Steiner tree problem we are given $k$ sets
$g_1,\dots,g_k\subseteq V$, and a root vertex $r\in V$. The goal is to
construct a minimum-weight tree spanning the vertex $r$ and at least one
vertex from each group (set) $g_i$. In the hop-constrained version of the
problem, we are given in addition a hop-bound $\red{h}$, and the requirement
is to find a minimum weight subgraph $H$, such that for every $g_i$,
$\hop_H(r,g_i)\le \red{h}$.
Denote by $\opt$ the weight of the optimal solution.
In the non-constrained version, where the input metric is a tree, \cite{GKR00}
provided a  $O(\log n\cdot\log k)$ approximation. Using the scheme of
\cite{HHZ21}, one can find a solution $H$ of weight $O(\log^4 n\cdot\log
k)\cdot \opt$, such that for every $g_i$, $\hop_H(r,g_i)\le
\red{O(\log^3n)\cdot h}$.
The best approximation for the (non-constrained) group Steiner tree problem is
$O(\log ^2n\cdot\log k)$ \cite{GKR00}.
Clearly, every improvement over the hop-constrained Ramsey-type embedding will
imply better bicriteria approximation algorithms. However, even if one used
the \cite{HHZ21} framework on optimal Ramsey-type embeddings (ignoring hop
constraints), the cost factors will be inferior compared to the
non-constrained versions (specifically $O(\log^3n\cdot\log k)$ for GST).
It is interesting to understand whether there is a separation between the cost
of a bicriteria approximation for hop-constrained problems and that of their
non-constrained counterparts. A more specific question is the following:
\begin{question}\label{question:Approximation}
	Could one match the state of the art cost approximation of the (non
	hop-constrained) group Steiner tree problem while having $\polylog(n)$
	hop-stretch?
\end{question}

Metric data structures such as distance oracles, distance labeling, and
compact routing schemes are extensively studied, and widely used throughout
algorithmic design.
A natural question is whether one can construct similar data structures w.r.t.
hop-constrained distances.
Most prominent is the question of compact routing schemes \cite{TZ01b}, where
one assigns each node in a network a short label and small local routing
table, such that packets could be routed throughout the network with routes
not much longer than the shortest paths.
Hop-constrained routing is a natural question due to its clear advantages
(bounding transmission delays, and increasing reliability). Indeed, numerous
papers developed different heuristics for different models of hop-constrained
routing.\footnote{In fact, at the time of writing, the search
"hop-constrained" + routing on Google Scholar returns 1,020 papers.}
However, to the best of the author's knowledge, no prior provable guarantees
were provided for hop-constrained compact routing schemes.
\begin{question}\label{question:Routing}
	What hop-constrained compact routing schemes are possible?
\end{question}

\subsection{Our contribution}
\subsubsection{Metric embeddings}
\begin{figure}[t]
	\centering{\includegraphics[scale=.85,page=1]{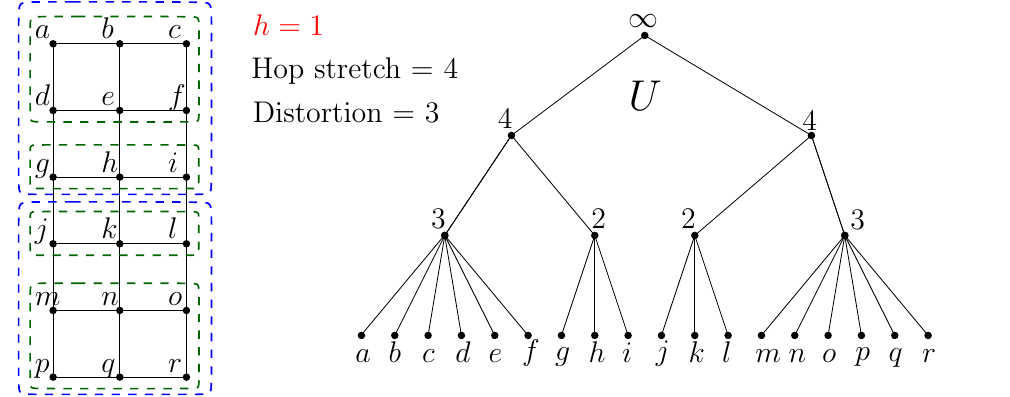}}
	\caption{\label{fig:RamseyDef}\small An illustration of an embedding of
	the unweighted $3\times6$ grid into an ultrametric with Ramsey
	hop-distortion $(t=3,M=\{a,b,c,p,q,r\},\red{\beta=4},\red{h=1})$. In the
	ultrametric, the distance between a pair of vertices equals the label
	of the least common ancestor.
	The ultrametric defines a hierarchical partition, which corresponds to
	the blue and green dashed lines on the left.
	The embedding has hop stretch $\red{\beta=4}$ as $\forall u,v\in V$,
	$d_{G}^{\red{(4\cdot 1)}}(u,v)\le d_U(u,v)$. Note that $4$ is tight, as
	$d_{G}^{\red{(4)}}(a,i)=4=d_U(a,i)<d_{G}^{\red{(3)}}(a,i)=\infty$. The
	embedding has distortion $t=3$ as $\forall u\in M,v\in V$, $d_U(u,v)\le
	t\cdot d_{G}\rhop{1}(u,v)$. Note that $3$ is tight, as $d_U(a,b)=3=3\cdot
	d_G\rhop{1}(a,b)$.
	Interestingly, the exact same embedding also has Ramsey hop-distortion
	$(t=4,M=\{a,b,c,d,e,f,m,n,o,p,q,r\},\red{\beta=4},\red{h=1})$.}
\end{figure}
Here we present our metric embedding results. Our embeddings will be into
ultrametrics (a.k.a. HST), which are a structured type of tree having the
strong triangle inequality (see \Cref{def:ultra}). We denote ultrametrics
usually with $U$.
Our first contribution is an improved Ramsey-type hop-constrained embedding
into trees,  partially solving \Cref{question:ImproveEmbedding}.
Embedding $f:V(G)\rightarrow U$ is said  to have Ramsey hop-distortion
$(t,M,\red{\beta},\red{h})$ (here $\red{\beta}>1$) if $\forall u,v\in V$,
$d_{G}^{\red{(\beta\cdot h)}}(u,v)\le d_U(u,v)$, and $\forall u\in M,v\in V$,
$d_U(u,v)\le t\cdot d_{G}\rhop{h}(u,v)$. See \Cref{def:Ramsey} for a formal
definition and \Cref{fig:RamseyDef} for an illustration.
In \Cref{thm:UltrametricRamsey} below we construct a distribution over
Ramsey-type embeddings with inclusion probability $\Omega(n^{-\frac1k})$ (that
is every vertex $v$ belongs to the set $M$ with probability
$\Omega(n^{-\frac1k})$), and every embedding in the support has Ramsey
hop-distortion $(O(k),M,\red{O(k\cdot\log n)},\red{h})$. That is the distance
between every pair $u,v$ (where at least one of them belongs to $M$) has
distortion $O(k)$ and hop stretch $\red{O(k\cdot\log n)}$. We also allow
taking the inclusion probability to the extreme, where each vertex belongs to
$M$ with probability $1-\eps$. Then the distortion is  $O(\frac{\log
n}{\eps})$, while the hop-stretch is $\red{O(\frac{\log^2 n}{\eps})}$.
\begin{restatable}[Hop-Constrained Ramsey Embedding]{theorem}{RamseyType}
	\label{thm:UltrametricRamsey}
	Consider an $n$-vertex graph $G=(V,E,w)$ with polynomial aspect ratio, and
	parameters $k,\red{h}\in [n]$.
	Then there is a distribution $\mathcal{D}$ over dominating ultrametrics,
	such that:
	\begin{enumerate}
		\item Every $U\in\supp(\mathcal{D})$, has Ramsey hop-distortion
		$(O(k),M,\red{O(k\cdot\log n)},\red{h})$, where $M\subseteq V$ is a
		random variable.
		\item For every $v\in V$, $\Pr_{U\sim\mathcal{D}}[v\in M]\ge
		\Omega(n^{-\frac1k})$.
	\end{enumerate}
	In addition, for every $\eps\in (0,1)$, there is a distribution
	$\mathcal{D}$ as above such that every $U\in\supp(\mathcal{D})$ has Ramsey
	hop-distortion $(O(\frac{\log n}{\eps}),M,\red{O(\frac{\log^2
	n}{\eps})},\red{h})$, and for every $v\in V$, $\Pr[v\in M]\ge 1-\eps$.
\end{restatable}
 \Cref{thm:UltrametricRamsey}  is optimized for the distortion parameter
 ($O(k)/O(\frac{\log n}{\eps})$).
Indeed, ignoring the hop-distortion (e.g. setting $\red{h=n}$), the tradeoff
in \Cref{thm:UltrametricRamsey} between distortion and inclusion
probability is asymptotically tight (see \cite{BBM06,FL21}).
However, it is yet unclear what the best possible hop-stretch obtainable with
asymptotically optimal distortion-inclusion probability trade-off is.
Interestingly, in \Cref{thm:UltrametricRamseyAlt} we show that by increasing
the distortion by a $\log\log n$ factor, we can obtain a sub-logarithmic
hop-stretch (though not with a constant inclusion probability). Further, in
this version we can also drop the assumption that the aspect ratio
is polynomial.

\begin{restatable}[Hop-Constrained Ramsey Embedding with small hop-stretch and
arbitrary aspect-ratio]{theorem}{RamseyTypeAlt}
	\label{thm:UltrametricRamseyAlt}
		Consider an $n$-vertex graph $G=(V,E,w)$, and parameters
		$k,\red{h}\in [n]$.
	Then there is a distribution $\mathcal{D}$ over dominating ultrametrics
	with $V$ as leaves, such that every~$U\in\supp(\mathcal{D})$ has Ramsey
	hop-distortion $(O(k\cdot\log\log n),M,\red{O(k\cdot\log\log
	n)},\red{h})$, where $M\subseteq V$ is a random variable, and $\forall
	v\in V$, $\Pr_{U\sim\mathcal{D}}[v\in M]\ge \Omega(n^{-\frac1k})$
\\In addition, for every $\eps\in (0,1)$, there is a distribution $\mathcal{D}$
as above such that every $U\in\supp(\mathcal{D})$ has Ramsey hop-distortion
$(O(\frac{\log n\cdot\log\log n}{\eps}),M,\red{O(\frac{\log n\cdot\log\log
n}{\eps})},\red{h})$, and $\forall v\in V$, $\Pr[v\in M]\ge 1-\eps$.
\end{restatable}
Compared to \cite{HHZ21}, \cite{HHZ21} provides results for inclusion
probability in the $1-\eps$ parameter regime, while our
\Cref{thm:UltrametricRamsey,thm:UltrametricRamseyAlt} hold in a much wider
spectrum of parameters.
Fixing inclusion probability $1-\eps$, in \Cref{thm:UltrametricRamsey} we
improved the distortion by a quadratic factor in the dependence on $n$, while
keeping the hop-stretch intact.
In \Cref{thm:UltrametricRamseyAlt} we improve both distortion and hop-stretch
by almost a quadratic factor (in the dependence on $n$).

\begin{figure}[t]
	\centering{\includegraphics[scale=.85,page=2]{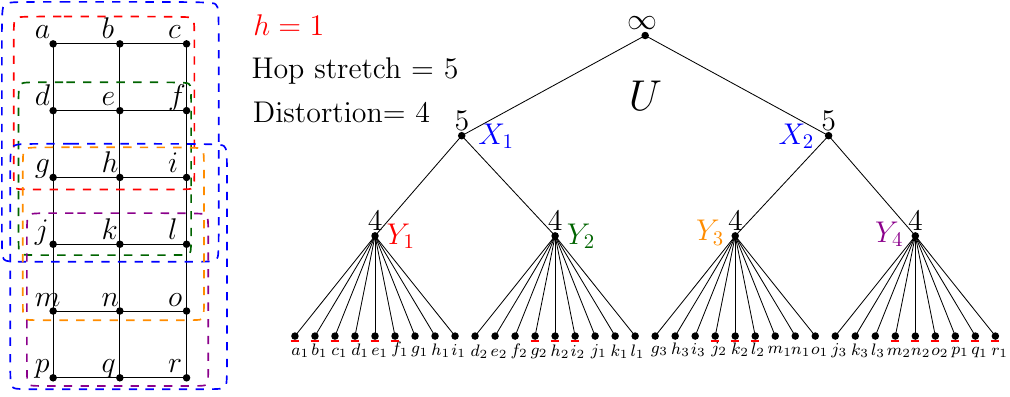}}
	\caption{\label{fig:ClanDef}\small An illustration of a clan embedding of
	the unweighted $3\times6$ grid into an ultrametric with hop distortion
	$(t=4,\red{\beta=5},\red{h=1})$.
	The copies of each vertex $v$ are denoted by a subset of $v_1,v_2,v_3$.
	The vertex representing the copy chosen to be the chief has a
	red underline.
	The clan embedding has hop stretch $\beta=5$, indeed
	$\forall u,v\in V,~~d^{\red{(5\cdot 1)}}_G(u,v)\le \min_{u'\in f(u),v'\in
	f(v)}d_U(u',v')$.
	Note that $5$ is tight, as
	$d^{\red{(5)}}_G(a,l)=5=d_U(a_1,l_1)<d^{\red{(4)}}_G(a,l)=\infty$.
	The clan embedding has distortion $t=4$, indeed
	$\forall u,v\in V,~~\min_{u'\in f(u)}d_U(u',\chi(v))\le 4\cdot
	d\rhop{1}_G(u,v)$. Note that $4$ is tight, as
	$\min_{j'\in f(j)}d_U(j',\chi(g))=d_U(j_1,g_2)=4=4\cdot
	d\rhop{1}_G(j,g)$.\\
	The clan embedding has hop-path-distortion $(4,\red{h=2})$. Consider the
	$2$-respecting path $(e,h,k)$ in $G$ of weight $2$. Then the copies
	$e_2,h_2,k_1$ in $U$ correspond to a ``path'' of weight $8$. On the other
	hand the path $(e,h,k,n)$ is not $2$-respecting, and indeed there is no
	corresponding path in $U$ of finite weight.}
\end{figure}

\sloppy Next, we construct a hop-constrained clan embedding $(f,\chi)$.
Recall that clan
embedding is a one-to-many embedding $f:V(G)\rightarrow 2^U$ (where the
vertices in $f(v)$ are called copies of $v$), and a choice of chiefs
$\chi:V(G)\rightarrow U$ (where $\forall v,~\chi(v)\in f(v)$).
A clan embedding $(f,\chi)$  is said to have hop-distortion
$(t,\red{\beta},\red{h})$ if it is dominating, that is, the distance in $U$ between
every two copies is at least the $\red{\beta\cdot h}$ hop-distance, and the
distance from some copy to the chief is distorted by at most a $t$ factor
(w.r.t. the $\red{h}$ hop distance). Formally:
$$\forall u,v\in V,~~d^{\red{(\beta\cdot h)}}_G(u,v)\le \min_{u'\in f(u),v'\in
f(v)}d_U(u',v')\le \min_{u'\in f(u)}d_U(u',\chi(v))\le t\cdot
d\rhop{h}_G(u,v)~.$$
See \Cref{def:clan} for a formal definition, and \Cref{fig:ClanDef} for an
illustration.
When constructing a clan embedding, the goal is to minimize the trade-off
between the hop-distortion and the number of copies each vertex has.
In \Cref{thm:ClanHopUltrametric}, assuming polynomial aspect ratio, for
$k\ge 1$ we construct a distribution over clan embeddings with hop distortion
$\left(O(k),\red{O(k\cdot\log n)},\red{h}\right)$, while every vertex has only
$ O(n^{\frac1k})$ copies in expectation.
Alternatively, it is also possible to get only $1+\eps$ copies in expectation
(that is, most vertices will have a single copy), while all the clan
embeddings in the support will have hop-distortion $\left(O(\frac{\log
n}{\epsilon}),\red{O(\frac{\log^2 n}{\epsilon})},\red{h}\right)$.

In addition to the hop-distortion guarantee, in
\Cref{thm:ClanHopUltrametric,thm:ClanHopUltrametricAlt} we also provide
hop-path-distortion guarantee. Roughly speaking,
a path $P=(v_0,v_1,\dots,v_m)$ is $\red{h}$-respecting if the induced graph
$G[P]$ (note that $P$ might not be simple) does not contain low cost paths
with more than $\red{h}$-hops (see \Cref{def:hRespecting}).
A one-to-many embedding $f$ has hop-path-distortion $(t,\red{h})$ if for every
$\red{h}$-respecting path $P=(v_0,v_1,\dots,v_m)$, one can choose copies
$v_i'\in f(v_i)$ such that the path $P'=(v_0',v_1',\dots,v_m')$ in $U$ has
total weight at most $t\cdot w_G(P)$.
Path-distortion was originally introduced by Bartal and Mendel
\cite{BM04multi} in the context of multi-embeddings of general graphs (without
$\red{h}$-respecting requirement), and was shown to be useful for
approximation algorithms (e.g. group Steiner tree).
See \Cref{def:HopPathDist} for a formal definition.
\begin{restatable}[Clan embedding into
ultrametric]{theorem}{ClanHopUltrametric}
	\label{thm:ClanHopUltrametric}
	Consider an $n$-vertex graph $G=(V,E,w)$ with polynomial aspect ratio, and
	parameters $k,\red{h}\in [n]$.
	Then there is a distribution $\mathcal{D}$ over clan embeddings $(f,\chi)$
	into ultrametrics with hop-distortion $\left(O(k),\red{O(k\cdot\log
	n)},\red{h}\right)$, hop-path-distortion $\left(\redno{O(k\cdot\log
	n)},\red{h}\right)$, and such that for every vertex $v\in V$,
	$\mathbb{E}_{f\sim\mathcal{D}}[|f(v)|]\le O(n^{\frac1k})$.

	In addition, for every $\eps\in (0,1)$, there is a distribution
	$\mathcal{D}$ as above such that every $U\in\supp(\mathcal{D})$ has
	hop-distortion $\left(O(\frac{\log n}{\epsilon}),\red{O(\frac{\log^2
	n}{\epsilon})},\red{h}\right)$, hop-path-distortion
	$\left(\redno{O(\frac{\log^2 n}{\epsilon})},\red{h}\right)$, and such that
	for every vertex $v\in V$,
	$\mathbb{E}_{f\sim\mathcal{D}}[|f(v)|]\le1+\eps$.
\end{restatable}
In \Cref{thm:ClanHopUltrametric} we optimize the trade-off between the
distortion $t$ and the expected clan size
$\mathbb{E}_{f\sim\mathcal{D}}[|f(v)|]$ (in similar fashion to
\Cref{thm:UltrametricRamsey}). Indeed, this trade-off in
\Cref{thm:ClanHopUltrametric} is asymptotically optimal \cite{FL21}.
It is yet unclear what the best possible hop-stretch obtainable with
asymptotically optimal distortion-expected clan size trade-off is.
In a similar fashion to \Cref{thm:UltrametricRamseyAlt}, in
\Cref{thm:ClanHopUltrametricAlt}, we show that by increasing the distortion by
a $\log\log n$ factor, we can obtain a sub-logarithmic hop-stretch (while
keeping the expected clan sizes the same). Further, in this version we can
also drop the assumption that the aspect ratio is polynomial.
\begin{restatable}[Clan embedding into ultrametric with small hop-stretch and
arbitrary aspect-ratio]{theorem}{ClanHopUltrametricAlt}
	\label{thm:ClanHopUltrametricAlt}
	Consider an $n$-vertex graph $G=(V,E,w)$, and parameters
	$k,\red{h}\in [n]$.
	Then there is a distribution $\mathcal{D}$ over clan embeddings $(f,\chi)$
	into ultrametrics with hop-distortion $\left(O(k\cdot\log\log
	n),\red{O(k\cdot\log\log n)},\red{h}\right)$, hop-path-distortion
	$\left(O(k\cdot\log n\cdot\log\log n),\red{h}\right)$, and such that for
	every point $v\in V$, $\mathbb{E}_{f\sim\mathcal{D}}[|f(v)|]\le
	O(n^{\frac1k})$.

	In addition, for every $\eps\in (0,1)$, there is a distribution
	$\mathcal{D}$ as above such that every $U\in\supp(\mathcal{D})$ has
	hop-distortion $\left(O(\frac{\log n\cdot\log\log
	n}{\epsilon}),\red{O(\frac{\log n\cdot\log\log
	n}{\epsilon})},\red{h}\right)$, hop-path-distortion
	$\left(\redno{O(\frac{\log^2 n\cdot\log \log
	n}{\epsilon})},\red{h}\right)$, and such that for every vertex $v\in V$,
	$\mathbb{E}_{f\sim\mathcal{D}}[|f(v)|]\le1+\eps$.
\end{restatable}

Our final embedding result is a one-to-many ``subgraph preserving'' embedding.
This is a strengthening of the path-distortion guarantee from
\Cref{thm:ClanHopUltrametric,thm:ClanHopUltrametricAlt}. The path-distortion
guarantee provides a bound only for $\red{h}$-respecting paths, while in many
applications we will not have such a promise.
A \emph{path-tree one-to-many embedding} is a map $f:V\rightarrow 2^T$, where
each vertex is sent to a subset of vertices of a tree $T$, and all the tree
vertices are copies of some graph vertices $f(V)=T$ (here for $A\subseteq V$,
$f(A)=\cup_{v\in A}f(v)$).
In addition, each edge $e=\{u',v'\}$, where $u'\in f(u),v'\in f(v)$, is
associated with a $u-v$ path $P^T_e$ in $G$ of weight at most $w_T(e)$. Given
a path $P$ in $T$, we can obtain an associated path $P^T$ in $G$ by simply
concatenating the associated paths of all the edges in $P$.
We say that $f$ has hop-bound $\red{\beta}$ if the length (in hops) of every
such concatenation of associated paths $P^T$ has at most $\red{\beta}$ hops.
Note that $T$ might contain an edge of weight $\infty$.
Compared with the definition of clan embedding, here we embed into a tree (and
not an ultrametric), there are no chiefs, and most importantly, each edge
in the tree is associated with a path in the graph. See
\Cref{def:OTM-PathTreeEmbedding} and \Cref{fig:PathTree} for a formal
definition and an illustration.

In \Cref{thm:Subgraph} below, we obtain a path-tree one-to-many embedding $f$
with hop-bound $\red{O(\log^3n)\cdot h}$ (for a parameter $\red{h}$).
The important guarantee is the following: given a subgraph $H$ of $G$, there
is a subgraph $H'$ of $T$ of weight at most $w_T(H')\le O(\log n)\cdot w(H)$,
and such that for every pair of vertices $u,v$ at hop-distance at most
$\red{h}$ (w.r.t. $H$), $H'$ contains two copies $u'\in f(u)$ and $v'\in f(v)$
in the same connected component.
For compression, in the path-distortion guarantee we require (morally) that
the entire subgraph $H$ has hop-diameter $\red{h}$, and get a single subgraph
$H'$ of $T$ containing a copy of each vertex in $H$. In contrast, here there
is no limitation on $H$, and $H'$ provides a ``localized'' guarantee where
near-by vertices (w.r.t. hops) should have copies in the same connected
component in $H'$. Note that this is a much stronger guarantee that can be
used for problems where the solution does not necessarily have bounded
hop-diameter (for example hop-constrained group Steiner forest).

\begin{restatable}[Subgraph Preserving
Embedding]{theorem}{SubgraphPreservingEmbedding}\label{thm:Subgraph}
	Consider an $n$-point graph $G=(V,E,w)$ with polynomial aspect ratio,
	parameters $\red{h},k\in\N$ and vertex $r\in V$.
	Then there is a path-tree one-to-many embedding $f$ of $G$ into a tree $T$
	with hop bound $\red{O(k\cdot\log^3n)\cdot h}$, $|f(V)|=(2n)^{1+\frac1k}$,
	$|f(r)|=1$, and such that for every sub-graph $H$, there is a subgraph
	$H'$ of $T$ where for every $u,v\in H$ with $\hop_H(u,v)\le\red{h}$, $H'$
	contains a pair of copies $u'\in f(u)$ and $v'\in f(v)$ in the same
	connected component. Further, $w_T(H')\le O(k\cdot\log n)\cdot w(H)$.
\end{restatable}
In \Cref{thm:SubgraphAlt} we prove an alternative version of
\Cref{thm:Subgraph} where the hop bound is $\red{\tilde{O}(k\cdot\log^2n)\cdot
h}$ and the bound on the weight of $H'$ is only $\tilde{O}(k\cdot\log
n)\cdot w(H)$.
In the application of \Cref{thm:Subgraph} in our paper we will use $k=O(1)$ as
we want to minimize the hop and weight bounds, and the total number of copies
is insignificant (as long as it is polynomial). Nevertheless, one can choose
$k=\log n$ to obtain linear number of copies in total (which might be
desirable for efficient algorithms).

Haeupler \etal \cite{HHZ21} had a similar embedding
to  \Cref{thm:Subgraph},
which they refer to as ``copy tree embedding'' (see theorem 8 in the full
\href{https://arxiv.org/pdf/2011.06112.pdf}{arXiv version}).
Haeupler \etal \cite{HHZ21} had hop bound of $O(\log^3n)$, and the cost of the
corresponding copy of $H$ was at most $O(\log^3n)\cdot w(H)$. Thus we get a
significant polynomial savings in these aspects, which directly translates to
the performance of our approximation algorithms.
An advantage of \cite{HHZ21} is that every vertex is guaranteed to have at
most $O(\log n)$ copies (implying $O(n\log n)$ total copies), while our
\Cref{thm:Subgraph} doesn't have such a guarantee even for $k=\log n$, when
the total number of copies is linear.
While such a guarantee could be useful for future applications, it is
insignificant for the applications in this paper.

\subsubsection{Approximation algorithms}
\begin{table}[p]
	\centering
	\begin{tabular}{|l|l|c|l|}
		\hline
		\textbf{Hop-Constrained Problem} & \textbf{Cost Apx.} & \textbf{Hop
		Apx.} &  \textbf{Reference} \\ \hline
		\multicolumn{4}{|c|}{ \textbf{Application of the clan
		embedding}} \\ \hline
		\multirow{ 3}{*}{H.C. Group Steiner Tree } & $O(\log ^ 4 n\cdot\log
		k)$ & $O(\log ^ 3 n)$ &  \cite{HHZ21}\\\cline{2-4}
		& $O(\log ^ 2 n\cdot\log k)$ & $O(\log ^ 2 n)$ &
		\Cref{thm:GroupSteinerTree} $^{(\checkmark)}$\\\cline{2-4}
		& $\tilde{O}(\log ^ 2 n\cdot\log k)$ & $\tilde{O}(\log n)$ &
		\Cref{thm:GroupSteinerTree}\\ \hline
		\multirow{ 3}{*}{H.C. Online Group Steiner Tree} & $O(\log ^ 5
		n\cdot\log \kappa)$ & $O(\log ^ 3 n)$ &   \cite{HHZ21}\\ \cline{2-4}
		& $O(\log ^ 3 n\cdot\log \kappa)$ & $O(\log ^ 2 n)$ &
		\Cref{thm:OnlineGroupSteinerTree} $^{(\checkmark)}$\\ \cline{2-4}
		& $\tilde{O}(\log ^ 3 n\cdot\log \kappa)$ & $\tilde{O}(\log n)$ &
		\Cref{thm:OnlineGroupSteinerTree}\\ \hline

		\multirow{ 3}{*}{H.C. Group Steiner Forest } & $O( \log ^ 6 n\cdot\log
		k)$ & $O(\log ^ 3 n)$ &  \cite{HHZ21} \\ \cline{2-4}
		& $O( \log ^ 4 n\cdot\log k)$ & $O(\log ^ 3 n)$ &
		\multirow{2}{*}{\Cref{thm:GroupSteinerForest}}  \\ \cline{2-3}
		& $\tilde{O}( \log ^ 4 n\cdot\log k)$ & $\tilde{O}(\log^2  n)$
		&  \\ \hline

		\multirow{ 3}{*}{H.C. Online Group Steiner Forest} &
		$O(\log ^ 7 n\cdot\log \kappa)$ & $O(\log ^ 3 n)$ &
		\cite{HHZ21} \\ \cline{2-4}
		& $O(\log ^ 5 n\cdot\log \kappa)$ & $O(\log ^ 3 n)$ &
		\Cref{thm:OnlineGroupSteinerForest} $^{(\checkmark)}$ \\ \cline{2-4}
		& $\tilde{O}(\log ^ 5 n\cdot\log \kappa)$ & $\tilde{O}(\log^2 n)$
		&\Cref{thm:OnlineGroupSteinerForest} \\\hline\hline

		\multicolumn{4}{|c|}{\textbf{Application of the Ramsey type
		embedding}} \\ \hline
		\multirow{4}{*}{H.C. Relaxed $k$-Steiner Tree} & $O(\log ^ 3 n)$&
		$O(\log n)$ & \cite{HKS09}  \\ \cline{2-4}
		& $O(\log ^ 2 n)$& $O(\log ^ 3 n)$ & \cite{HHZ21}  \\ \cline{2-4}
		& $O(\log n)$& $O(\log ^ 3 n)$ &
		\multirow{2}{*}{\Cref{cor:HHZ21Application}}  \\ \cline{2-3}
		& $\tilde{O}(\log n)$& $\tilde{O}(\log ^ 2 n)$ &    \\ \hline
		\multirow{5}{*}{H.C. $k$-Steiner Tree}& $O(\log ^ 3 n\cdot\log k)$&
		$O(\log^2 n)$ & \cite{HKS09}  \\ \cline{2-4}
		& $O(\log ^ 2 n)$& $O(\log	 n)$ & \cite{KS16}  \\ \cline{2-4}
		& $O(\log ^ 2 n\cdot\log k)$ & $O(\log ^ 3 n)$ &   \cite{HHZ21}
		\\ \cline{2-4}
		& $O(\log n\cdot\log k)$& $O(\log ^ 3 n)$ &
		\multirow{2}{*}{\Cref{cor:HHZ21Application}}  \\ \cline{2-3}
		& $\tilde{O}(\log n\cdot\log k)$& $\tilde{O}(\log ^ 2 n)$ & \\ \hline
		\multirow{3}{*}{H.C. Oblivious Steiner Forest} &$O(\log ^ 3 n)$ &
		$O(\log ^ 3 n)$ &  \cite{HHZ21}\\\cline{2-4}
		&$O(\log^2 n)$ & $O(\log ^ 3 n)$ &
		\multirow{2}{*}{\Cref{cor:HHZ21Application}}  \\ \cline{2-3}
		&$\tilde{O}(\log^2 n)$ & $\tilde{O}(\log ^ 2 n)$ & \\	  \hline
		\multirow{3}{*}{H.C. Oblivious Network Design} & $O(\log ^ 4 n)$ &
		$O(\log ^ 3 n)$ &   \cite{HHZ21}\\ \cline{2-4}
		&$O(\log^3 n)$ & $O(\log ^ 3 n)$ &
		\multirow{2}{*}{\Cref{cor:HHZ21Application}}  \\ \cline{2-3}
		&$\tilde{O}(\log^3 n)$ & $\tilde{O}(\log ^ 2 n)$ & \\	  \hline

	\end{tabular}\caption{\footnotesize Our bicriteria approximation results.
	H.C. stands for hop-constrained.
	All approximation guarantees are in worst case, other than for the two
	online problems where the approximation is in expectation.
		All input graphs are assumed to have polynomial (in $n$) aspect ratio.
		The parameter $\kappa$ in online group Steiner tree/forest denotes the
		size of an initial set $\mathcal{K}$ containing all
		potential requests.
		\newline
		Our results for the group Steiner tree/forest, and online group
		Steiner tree/forest are based on our hop-constrained clan embedding.
		All the other results are applications of our Ramsey type
		hop-constrained embedding, and are obtained from \cite{HHZ21} in a
		black box manner.
		The results marked with $(\checkmark)$ match the state of the art
		approximation factor (when one ignores hop constraints).
		Note that we obtain the first improvement over the cost approximation to
		the well studied h.c. $k$-Steiner tree problems (since the 2011
		conference version of \cite{KS16}).
	} \label{tab:ApproxAlgs}
\end{table}
\cite{HHZ21} used their Ramsey type embeddings to obtain many different
approximation algorithms for hop-constrained (abbreviated h.c.) problems. Our
improved Ramsey embeddings directly imply improved approximation factors and
hop-stretch for all these problems.
Most prominently, we improve the cost approximation for the well studied h.c.
$k$-Steiner tree problem\footnote{In the $k$-Steiner tree problem we are given
a root $r\in V$, and a set $K\subseteq V$ of at least $k$ terminals. The goal
is to find a connected subgraph $H$ spanning the root $r$, and at least $k$
terminals, of minimum weight. In the hop-constrained version we are
additionally required to guarantee that the hop diameter of $H$ will be at
most $\red{h}$.\label{foot:kSteinerTree}}
over \cite{KS16} (the \cite{KS16} approximation is superior to \cite{HHZ21}).
See \Cref{cor:HHZ21Application} and \Cref{tab:ApproxAlgs} for a summary.

We go beyond applying our embeddings into \cite{HHZ21} as a black box and
obtain further improvements.
A subgraph $H$ of $G$ is $\red{h}$-respecting if for every $u,v\in V(H)$,
$d\rhop{h}_G(u,v)\le d_H(u,v)$.
In \Cref{cor:hBoundedSubgraph} we show that there is a one-to-many embedding
$f:V(G)\rightarrow 2^{T}$ into a tree, such that for every connected
$\red{h}$-respecting subgraph $H$, there is a connected subgraph $H'$ of $T$
of weight $O(\log n)\cdot w(H)$ and hop diameter $\red{O(\log^2n)\cdot h}$
containing at least one vertex from $f(u)$ for every $u\in V(H)$
(alternatively, in \Cref{cor:hBoundedSubgraphAlt} we are guaranteed a similar
subgraph of weight $\tilde{O}(\log n)\cdot w(H)$ and hop diameter
$\red{\tilde{O}( \log n)\cdot h}$).
We use \Cref{cor:hBoundedSubgraph} to construct a bicriteria approximation
algorithm for the h.c. group Steiner tree problem, the cost of which matches
the state of the art for the (non h.c.) group Steiner tree problem, thus
answering \Cref{question:Approximation}.
Later, we apply \Cref{cor:hBoundedSubgraph} to the online h.c. group Steiner
tree problem, and obtain a competitive ratio that matches the competitive ratio
of the non-constrained version (see \Cref{thm:OnlineGroupSteinerTree}).
\begin{restatable}[]{theorem}{GroupSteinerTree}
	\label{thm:GroupSteinerTree}
	Given an instance of the $\red{h}$-h.c. group Steiner tree problem with
	$k$ groups on a graph with polynomial aspect ratio,
	there is a poly-time algorithm that returns a subgraph $H$ such that
	$w(H)\le  O(\log^2n\cdot\log k)\cdot\opt$, and $\hop_H(r,g_i)\le
	\red{O(\log^2n)\cdot h}$ for every $g_i$, where $\opt$ denotes the weight
	of the optimal solution.~
	Alternatively, one can return a subgraph $H$ such that $w(H)\le
	\tilde{O}(\log^2n\cdot \log k)\cdot\opt$, and $\hop_H(r,g_i)\le
	\red{\tilde{O}(\log n)\cdot h}$ for every $g_i$.
\end{restatable}

Next we study the h.c. group Steiner forest problem, where we are given
a parameter $\red{h}$, and $k$ subset pairs $(S_1,R_1),\dots,(S_k,R_k)\subseteq
V$, where for every $i\in[k]$, $S_i,R_i\subseteq V$. The goal is to construct
a minimum-weight forest $F$, such that for every $i$, there is a path with at
most $\red{h}$ hops from a vertex in $S_i$ to a vertex in $R_i$.
As the optimal solution to the h.c. group Steiner forest is not necessarily
$\red{h}$-respecting, we cannot use \Cref{cor:hBoundedSubgraph}. Instead, we
use our subgraph-preserving embeddings
(\Cref{thm:Subgraph,thm:SubgraphAlt}) to obtain a bicriteria approximation. A
similar result for the online h.c. group Steiner forest problem is presented
in \Cref{thm:OnlineGroupSteinerForest}.
\begin{restatable}[]{theorem}{GroupSteinerForest}
	\label{thm:GroupSteinerForest}
	Given an instance of the $\red{h}$-h.c. group Steiner forest problem with
	$k$ groups on a graph with polynomial aspect ratio, there is a poly-time
	algorithm that returns a subgraph $H$ such that $w(H)\le
	O(\log^4n\cdot\log k)\cdot\opt$, and $\hop_H(S_i,R_i)\le
	\red{O(\log^3n)\cdot h}$ for every $i$, where $\opt$ denotes the weight of
	the optimal solution.~
	Alternatively, one can return a subgraph $H$ such that $w(H)\le
	\tilde{O}(\log^4n\cdot \log k)\cdot\opt$, and $\hop_H(S_i,R_i)\le
	\red{\tilde{O}(\log^2 n)\cdot h}$ for every $i$.
\end{restatable}

\subsubsection{Hop-constrained metric data structures}
In metric data structures, our goal is to construct a data structure that will
store (estimated) metric distances compactly, and answer distance (or routing)
queries efficiently.
As the distances returned by such a data structure do not have to respect the
triangle inequality, one might hope to avoid any hop-stretch.
This is impossible in general. Consider a complete graph with edge weights
sampled randomly from $\{1,\alpha\}$ for some large $\alpha$. Clearly, from
information-theoretic considerations, in order to estimate $d_G^{\red{(1)}}$
with arbitrarily large distortion (but smaller than $\alpha$), one must use
$\Omega(n^2)$ space.
Therefore, in our metric data structures we will allow hop-stretch.

\paragraph{Distance Oracles.} A distance oracle is a succinct data structure
$\DO$ that (approximately) answers distance queries.
Chechik \cite{C15} (improving over previous results
\cite{TZ05,RTZ05,MN07,W13,C14}) showed that any metric (or graph) with $n$
points has a distance oracle of size $O(n^{1+\frac{1}{k}})$,\footnote{We
measure size in machine words, each word is $\Theta(\log n)$
bits.\label{foot:wordsize}} that can report any distance in $O(1)$ time with
stretch at most $2k-1$ (which is asymptotically optimal assuming Erd\H{o}s'
girth conjecture \cite{Erdos64}). That is, on query $u,v$, the answer
$\DO(u,v)$ satisfies
$d_G(u,v)\le\DO(u,v)\le(2k-1)\cdot d_G(u,v)$.

Here we introduce the study of h.c. distance oracles. That is, given a
parameter $\red{h}$, we would like to construct a distance oracle that on
query $u,v$, will return $d\rhop{h}(u,v)$, or some approximation of it.
Our result is the following:
\begin{restatable}[Hop-constrained Distance Oracle]{theorem}{DistanceOracle}
	\label{thm:DistanceOracle}
	For every weighted graph $G=(V,E,w)$ on $n$ vertices with polynomial
	aspect ratio, and parameters $k\in\N$, $\eps\in(0,1)$, $\red{h}\in\N$,
	there is an efficiently constructible distance oracle $\DO$ of size
	$O(n^{1+\frac1k}\cdot\log n)$, that for every distance query $(u,v)$, in
	$O(1)$ time returns a value $\DO(u,v)$ such that
	$d_{G}^{\red{(O(\frac{k\cdot\log\log n}{\epsilon})\cdot h)}}(u,v)\le
	\DO(u,v)\le (2k-1)(1+\epsilon)\cdot d_{G}\rhop{h}(u,v)$.
\end{restatable}

\paragraph{Distance Labeling.}
A \emph{distance labeling} is a distributed version of a distance oracle.
Given a graph $G=(V,E,w)$, each vertex $v$ is assigned a label $\ell(v)$, and
there is an algorithm $\A$, that given $\ell(u),\ell(v)$ returns a value
$\A(\ell(u),\ell(v))$ approximating $d_G(u,v)$.
In their celebrated work, Thorup and Zwick \cite{TZ05} constructed a distance
labeling scheme with labels of size $O(n^{\frac1k}\cdot\log n)$
,$^{\ref{foot:wordsize}}$ and such that in $O(k)$ time, $\A(\ell(u),\ell(v))$
approximates $d_G(u,v)$ within a $2k-1$ factor.
We refer to \cite{FGK24} for further details on distance labeling (see also
\cite{Peleg00Labling,GPPR04,EFN18}).

Matou{\v{s}}ek \cite{Mat96} showed that every metric space $(X,d)$ could be
embedded into $\ell_\infty$ of dimension $O(n^{\frac1k}\cdot k\cdot\log n)$
with distortion $2k-1$ (for the case $k=O(\log n)$, \cite{ABN11} later
improved the dimension to $O(\log n)$).
As was previously observed, this embedding can serve as a labeling scheme
(where the label of each vertex will be the representing vector in the
embedding). From this point of view, the main contribution of \cite{TZ05} is
the small $O(k)$ query time (as the label size/distortion tradeoff
is similar).

Here we introduce the study of h.c. labeling schemes, where the goal is to
approximate $d\rhop{h}_G(u,v)$.
Note that $d\rhop{h}_G$ is not a metric function, and in particular does not
embed into $\ell_\infty$ with bounded distortion.
Hence there is no trivial labeling scheme which is embedding based.
Our contribution is the following:
\begin{restatable}[Hop-constrained Distance
Labeling]{theorem}{DistanceLabeling}
	\label{thm:DistanceLabeling}
	For every weighted graph $G=(V,E,w)$ on $n$ vertices with polynomial
	aspect ratio, and parameters $k,\red{h}\in\N$, $\eps\in(0,1)$, there is an
	efficient construction of a distance labeling that assigns each node $v$ a
	label $\ell(v)$ of size  $O(n^{\frac1k}\cdot\log^2 n)$, and such that
	there is an algorithm $\A$ that on input $\ell(u),\ell(v)$, in $O(k)$ time
	returns a value such that $d_G^{\red{(O(\frac{k\cdot\log\log
	n}{\epsilon})\cdot h)}}(v,u)\le \A(\ell(v),\ell(u)) \le(2k-1)(1+\eps)\cdot
	d_G\rhop{h}(v,u)$.
\end{restatable}

\paragraph{Compact Routing Scheme.}
A \emph{routing scheme} in a network is a mechanism that allows packets to be
delivered from any node to any other node. The network is represented as a
weighted undirected graph, and each node can forward incoming data by using
local information stored at the node,  called a \emph{routing table}, and the
(short) packet's \emph{header}. The routing scheme has two main phases: in the
preprocessing phase, each node is assigned a routing table and a short
\emph{label}; in the routing phase, when a node receives a packet, it should
make a local decision, based on its own routing table and the packet's header
(which may contain the label of the destination, or a part of it), of where to
send the packet.
The {\em stretch} of a routing scheme is the worst-case ratio between the
length of a path on which a packet is routed to the shortest possible path.
For fixed $k$, the state of the art is by Chechik \cite{C13} (improving
previous works \cite{PU89,ABLP90,AP92,Cowen01,EGP03,TZ01b}) who obtain stretch
$3.68k$ while using $O(k\cdot n^{\frac{1}{k}})$ size tables and labels of size
$O(k\cdot\log n)$.$^{\ref{foot:wordsize}}$
For stretch $\Omega(\log n)$ further improvements were obtained by Abraham
\etal \cite{ACEFN20}, and the author and Le \cite{FL21}.

Here we introduce the study of h.c. routing schemes, where the goal is to
route the packet in a short path with a small number of hops. Specifically, for
parameter $\red{h}$, we are interested in a routing scheme that for every
$u,v$ will use a route with at most  $\red{\beta\cdot h}$ hops, and weight
$\le t\cdot d\rhop{h}_G(u,v)$, for some parameters $\red{\beta},t$. If
$\hop_G(u,v)\ge\red{h}$, the routing scheme is not required to do anything,
and any outcome will be accepted.
We answer \Cref{question:Routing} in the following:
\begin{restatable}[Hop-constrained Compact Routing
Scheme]{theorem}{RoutingScheme}
	\label{thm:RoutingScheme}
	For every weighted graph $G=(V,E,w)$ on $n$ vertices with polynomial
	aspect ratio, and parameters $k,\red{h}\in\N$, $\eps\in(0,1)$, there is an
	efficient construction of a compact routing scheme that assigns each node
	a table of size $O(n^{\frac1k}\cdot k\cdot\log n)$, a label of size
	$O(k\cdot\log^2 n)$. For every $u,v$ such that $\hop_G(u,v)\le \red{h}$,
	the routing of a packet from $u$ to $v$ will be done using a path $P$
	such that  $\hop(P)=\red{O(\frac{k^2\cdot\log\log n}{\eps})\cdot h}$, and
	$w(P) \le3.68k\cdot(1+\eps)\cdot d_G\rhop{h}(u,v)$.
\end{restatable}

\subsection{Related work}
Hop-constrained network design problems, and in particular routing, received
considerable attention from the operations research community, see
\cite{WA88,Gou95,Gou96,Voss99,GR01,GM03,AT11,RAJ12,BFGP13,%
BHR13,BFG15,TCG15,DGMG16,Lei16,BF18,DMMY18}
for a sample of papers.

Approximation algorithms for h.c. problems were previously constructed. They
are usually considerably harder than their non h.c. counterparts, and often
require bicriteria approximation. Previously studied problems include
minimum depth spanning tree \cite{AFHKRS05}, degree bounded minimum diameter
spanning tree \cite{KLS05}, bounded depth Steiner tree \cite{KP97,KP09}, h.c.
MST \cite{Rav94}, Steiner tree \cite{MRSRRH98}, and $k$-Steiner tree
\cite{HKS09,KS16}. We refer to \cite{HHZ21} for further details.
Several other approximation problems are also closely related, as they
optimize cost subject to explicit or implicit distance/depth requirements
(even though they don't have explicit hop constraints). Examples include
shallow-light trees and shallow-light Steiner trees \cite{KRY95,ES15ShallowLight},
cost-distance Steiner tree variants \cite{MMP08}, and buy-at-bulk network
design problems \cite{SCRS01,HKS09,KS16,GRTU17}.

Recently, Ghaffari \etal \cite{GHZ23} obtained a hop-constrained $\poly(\log)$
competitive oblivious routing scheme using techniques based on \cite{HHZ21}
hop-constrained Ramsey trees.
Notably, even though the $O(\frac{\log n}{\eps})$ distortion in
\Cref{thm:UltrametricRamsey} is superior to the $O(\frac{\log^2 n}{\eps})$
distortion in a similar theorem in \cite{HHZ21}, \cite{HHZ21} showed that the
expected distortion in their embedding is only
$\tilde{O}(\log\frac{n}{\eps})$. This additional property turned out to be
important for the oblivious routing application. See further discussion in
\Cref{subsec:Conclusion}, \questionref{question:Expected}.

Given a graph $G=(V,E,w)$, a \emph{hop-set} is a set of edges $G'$ that when
added to $G$, $d_{G\cup G'}\rhop{h}$ well approximates $d_G$. Formally, an
$(t,\red{h})$-hop-set is a set $G'$ such that $\forall u,v\in V$, $d_G(u,v)\le
d_{G\cup G'}\rhop{h}(u,v)\le t\cdot d_G(u,v)$.
Elkin and Neiman \cite{EN19hop} constructed
$\left(1+\eps,\red{\left(\frac{\log k}{\eps}\right)^{\log k-2}}\right)$
hop-sets with $O_{\eps,k}(n^{1+\frac1k})$ edges.
Hop-sets have been extensively studied; we refer to \cite{EN20} for a survey.
Another related problem is $\red{h}$-hop $t$-spanners. Here we are given a
metric space $(X,d)$, and the goal is to construct a graph $G$ over $X$ such
that $\forall x,y\in X$, $d_X(x,y)\le d_{G}\rhop{h}(x,y)\le t\cdot d_X(x,y)$.
See \cite{AMS94,Sol13,HIS13,LMS23} for Euclidean metrics, \cite{FN22} for
different metric spaces, \cite{FL22,FGN24SoCG} for reliable $\red{2}$-hop
spanners, and \cite{ASZ20} for low-hop emulators (where $d_G^{\red{(O(\log\log
n))}}$ respects the triangle inequality).

The idea of one-to-many embedding of graphs originated with Bartal and
Mendel \cite{BM04multi}, who for $k\ge1$ constructed an embedding into
ultrametric with $O(n^{1+\frac1k})$ nodes and path distortion $O(k\cdot\log
n\cdot \log\log n)$ (see \Cref{def:OTM-PathTreeEmbedding}, and ignore all hop
constraints). The path distortion was later improved to $O(k\cdot\log n)$
\cite{FL21,Bar21}.
Recently, Haeupler \etal \cite{HHZ21NonHop} studied approximate copy tree
embeddings, which are essentially equivalent to one-to-many tree embeddings.
Their ``path-distortion'' is inferior: $O(\log^2n)$, however they were able to
bound the number of copies of each vertex by $O(\log n)$ in the worst case
(not obtained by previous works).
One-to-many embeddings were also studied in the context of minor-free graphs,
where Cohen-Addad \etal \cite{CFKL20} constructed an embedding into low treewidth
graphs with expected additive distortion. Later, the author and Le \cite{FL21}
also constructed clan and Ramsey type embeddings of minor-free graphs into
low treewidth graphs.

Bartal \etal \cite{BFN22} showed that even when the metric space is the
shortest path metric of a planar graph with constant doubling dimension, the
general metric Ramsey type embedding \cite{MN07} cannot be substantially
improved. Finally, there are also versions of Ramsey type embeddings
\cite{ACEFN20}, and clan embeddings \cite{FL21} into spanning trees of the
input graph. These embeddings lose a $\log\log n$ factor in the distortion
compared with the embeddings into (non-spanning) trees.

\subsection{Conclusion and open problems}\label{subsec:Conclusion}
Hop-constrained shortest path distances are very challenging to work with as
they do not constitute a metric space. Haeupler, Hershkowitz, and Zuzic
\cite{HHZ21} were the first to study metric embeddings into trees of h.c.
distances. While such a metric  embedding is impossible in general, Haeupler
\etal managed to do so by introducing hop-stretch (that is additional slack in
the number of hops), and inclusion probability (that is Ramsey-type
embeddings). Haeupler \etal then used their Ramsey-type embedding to construct
a ``copy tree embedding'' where each subgraph in the original graph has a
corresponding copy in the image. Finally they used their embeddings to get a
plethora of h.c. approximation and online algorithms.

Haeupler \etal's constructions are based on stochastic decompositions ala
\cite{Bar96}. We employ more sophisticated metric embedding techniques (based
on \cite{ACEFN20}, which itself is based on \cite{MN07}) to obtain
significantly better distortion and hop-stretch. We then employ ideas of clan
embeddings \cite{BM04multi,FL21} to obtain one-to-many embeddings with path
distortion guarantees. Next, we use the clan embeddings to obtain path-tree
subgraph-preserving embeddings, which are similar to the \cite{HHZ21} copy tree
embedding, while ours have much superior guarantees. We then use our
embeddings to improve the parameters for the plethora of h.c. approximation
and online algorithms studied in \cite{HHZ21} (and previous papers).
Finally, we study h.c. metric data structures (not studied theoretically
before). Surprisingly, using a combination of our metric embeddings (to obtain
a rough estimate of the hop-constrained distance), together with previously
known metric data structures (in a black-box fashion) we are able to construct
h.c. metric data structures with parameters almost matching their non h.c.
counterparts.
Our work leaves several open questions and directions:
\begin{enumerate}
	\item Metric embeddings of topologically restricted graphs: A large
	portion of the metric embeddings literature is concerned with metric
	embedding of topologically restricted graph families such as planar
	graphs, or fixed-minor-free graphs.
	It will be interesting to understand how h.c. shortest path distances in
	such graphs behave, and to find additional structure that will allow for
	better metric embeddings.
	While we do not expect to find better Ramsey type or clan embeddings into
	trees (as there are no such non h.c. counterparts), it is plausible that
	there are good metric embeddings with low additive distortion into low
	treewidth graphs. See \cite{FKS19,CFKL20,FL21,FL22tw,CCLMST23,CLPP23}.
	\item Approximate Nearest Neighbor Search (ANNS): A very well studied
	metric data structure is that of ANNS (see the survey \cite{AIR18}). While
	most of the previous work is on norm spaces, there has also been some
	work on general metrics \cite{Fil23} and planar graphs
	\cite{ACKW15,Fil23}. It will be interesting to see if these works can be
	extended to the case of h.c. distances.
	\item Optimizing the parameters in our h.c. Ramsey-type embeddings: The
	first question will be whether it is possible to remove the $O(\log\log
	n)$ factors in \Cref{thm:UltrametricRamseyAlt} (similarly in
	\Cref{thm:ClanHopUltrametricAlt}). Secondly, it will be interesting to
	optimize the distortion (hop-stretch) and remove the asymptotic notation
	there. Note that the $k$ parameter governs the exponent in the inclusion
	probability, and hence is important.
	\item\label{question:Expected} Expected distortion: Haeupler \etal
	\cite{HHZ21} also studied the notion of expected distortion. Consider an
	ultrametric $U$, and a subset $M\subseteq V$ of saved vertices. Given a
	pair of vertices $u,v\in V$, set $d_{U,M}(u,v)=\begin{cases}
		d_{U}(u,v) & u,v\in M\\
		0 & \text{otherwise}
	\end{cases}$.
	The main result in \cite{HHZ21} is an embedding with inclusion probability
	$1-\eps$, distortion $O(\frac{\log^2 n}{\eps})$, hop stretch
	$O(\frac{\log^2 n}{\eps})$, and in addition, for every pair of vertices,
	the expected distortion is bounded by
	$\mathbb{E}_{U,M}\left[d_{U,M}(u,v)\right]\le O(\log n\cdot\log(\frac{\log
	n}{\eps}))\cdot d_G^{\red{(h)}}(u,v)$. Note that for inclusion probability
	$1-\eps$, our \Cref{thm:UltrametricRamsey} only guarantees distortion
	$O(\frac{\log n}{\eps})$ (in the worst case). Thus  \cite{HHZ21}
	relaxation for expected distortion improves the dependence on $\eps$
	exponentially. As it turns out, in the application of \cite{HHZ21} for
	oblivious routing (see \cite{GHZ23}) one needs inclusion probability
	$1-\eps$ for $\eps=\poly(\frac1n)$, thus this difference is significant.
	It will be interesting to optimize the expected distortion w.r.t. the
	other parameters using the more sophisticated tools in this paper.
	\item H.c. Ramsey type embeddings into spanning trees (subgraph): Our
	embeddings in this paper are mostly into ultrametrics (or into abstract
	non-subgraph trees in \Cref{thm:Subgraph,thm:SubgraphAlt}). For non h.c.
	distances, stochastic, Ramsey-type, and clan embeddings were studied into
	spanning (sub-graph) trees as well \cite{EEST08,AN19,ACEFN20,FL21}. It is
	interesting to see if it is possible to obtain similar results for h.c.
	distances (perhaps by allowing the distances in the tree to be also h.c.).
	\item Relaxed notions of distortion: In the metric embeddings literature,
	there are many works studying relaxed notions of distortion such as
	scaling distortion \cite{CBCDGKNS05,CDG06,ABN11,BFN19} (where $1-\eps$ of
	the pairs are guaranteed to have distortion $f(\frac1\eps)$), terminal
	distortion \cite{EFN17,EN21,MMMR18,NN19,CN24} (where there is distortion
	guarantee only for the pairs $K\times V$ for a subset $K\subseteq V$ of
	terminals), and prioritized distortion \cite{EFN18,BFN19,EN22,FGK24}
	(where there is a priority order over the vertices, and the distortion
	guarantee is w.r.t. the order). It will be interesting to construct metric
	embeddings of h.c. distances with refined notions of distortion.
\end{enumerate}

\section{Preliminaries}\label{sec:prelim}
The $\tilde{O}$ notation hides poly-logarithmic factors, that is
$\tilde{O}(g)=O(g)\cdot\polylog(g)$.
Consider an undirected weighted graph $G=(V,E,w)$, and a hop
parameter $\red{h}$.
A graph is called unweighted if all its edges have unit weight.
In general, in the input graphs we will assume that all weights are finite
positive numbers, while in our output graphs we will also use $\infty$ as a
weight (this will represent that either vertices are disconnected or every
path between them has more than $\red{h}$ hops).
We denote $G$'s vertex set and edge set by $V(G)$ and $E(G)$, respectively.
Often we will abuse notation and write $G$ instead of $V(G)$.
$d_{G}$ denotes the shortest path metric in $G$, i.e., $d_G(u,v)$ equals
the minimum weight of a path between $u$ and $v$.
A path $P=(v_0,v_1,\dots,v_h)$ is said to have $\hop(P)=\red{h}$ hops. The
$\red{h}$-hop distance between two vertices $u,v$ is
$$d_G\rhop{h}(u,v)=\min\{w(P)~\mid~P\mbox{ is a path between $u$ and $v$ with
} \hop(P)\le \red{h}\}~.$$
If there is no path from $u$ to $v$ with at most $\red{h}$ hops, then
$d_G\rhop{h}(u,v)=\infty$.
$\hop_G(u,v)=\min \left\{h \mid d_G\rhop{h}(u,v)<\infty\right\}$ is the
minimum number of hops in a $u$-$v$ path. For two subsets $S,T\subseteq V$,
$\hop_G(S,T)=\min_{s\in S,t\in T} \left\{\hop_G(s,t)\right\}$, for a vertex
$v$, we can also write $\hop_G(v,S)=\hop_G(\{v\},S)=\min_{s\in
S}\left\{\hop_G(v,s)\right\}$.
The $\red{h}$-hop diameter of the graph is
$\diam\rhop{h}(G)=\max\{d_G\rhop{h}(u,v)\}$,
the maximal $\red{h}$-hop distance between a pair of vertices. Note that in
many cases, this will be $\infty$.
$B_G(v,r)=\{u\mid d_G(u,v)\le r\}$ is the closed ball around $v$ of radius
$r$, and similarly $B\rhop{h}_G(v,r)=\{u\mid d\rhop{h}_G(u,v)\le r\}$ is the
$\red{h}$-hop-constrained ball.

Throughout the paper we will assume that the minimum weight edge in the input
graph has weight $1$. Note that due to scaling this assumption does not
lose generality.
The \emph{aspect ratio},$^{\ref{foot:aspectRatio}}$ is the ratio between the
maximum and minimum weight in $G$ $\frac{\max_{e\in
E}w(e)}{\min_{e\in E}w(e)}$,
or simply $\max_{e\in E}w(e)$ as we assumed the minimum distance is $1$.
In many places (similarly to \cite{HHZ21}) we will assume that the aspect
ratio is polynomial in $n$ (actually in all results other than
\Cref{thm:UltrametricRamseyAlt,thm:ClanHopUltrametricAlt}). This assumption
will always be stated explicitly.

For a subset $A\subseteq V$ of vertices,
$G[A]=(A,E_A=E\cap{A\choose2},w_{\upharpoonright E_{A}})$ denotes the subgraph
induced by $A$.
The diameter of a cluster $S\subseteq V$, denoted by $\dm(S)$ is $\max_{u,v
\in S}d_{G[S]}(u,v)$.
\footnote{This is often called \emph{strong} diameter. A related notion is the
\emph{weak} diameter of a cluster $S$, defined as $ \max_{u,v \in S}d_{G}(u,v)$.
Note that for a metric space, weak and strong diameter are equivalent. See
\cite{Fil19padded}.}
Similarly, for a subset of edges $\tilde{E}\subseteq E$, where $A\subseteq V$
is the subset of vertices in $\cup_{e\in\tilde{E}}e$, the graph induced by
$\tilde{E}$ is $G[\tilde{E}]=(A,\tilde{E},w_{\upharpoonright \tilde{E}})$.

An ultrametric $\left(Z,d\right)$ is a metric space satisfying a
strong form of the triangle inequality, that is, for all $x,y,z\in Z$,
$d(x,z)\le\max\left\{ d(x,y),d(y,z)\right\} $.
In \Cref{def:ultra} we define a hierarchically well-separated tree (HST). It
is well known that HSTs and ultrametrics are equivalent (see \cite{BLMN05}),
and we will use these two terms interchangeably.
\begin{definition}[Hierarchically well-separated tree - HST]\label{def:ultra}
	An HST is a metric space $\left(Z,d\right)$ whose elements
	are the leaves of a rooted labeled tree $T$. Each $z\in T$ is associated
	with a label $\ell\left(z\right)\ge0$ such that if $q\in T$ is a
	descendant of $z$ then $\ell\left(q\right)\le\ell\left(z\right)$
	and $\ell\left(q\right)=0$ iff $q$ is a leaf. The distance between
	leaves $z,q\in Z$ is defined as
	$d_{T}(z,q)=\ell\left(\mbox{lca}\left(z,q\right)\right)$
	where $\mbox{lca}\left(z,q\right)$ is the least common ancestor of
	$z$ and $q$ in $T$.
\end{definition}

\paragraph{Metric Embeddings}
Classically, a metric embedding is defined as a function $f:X\rightarrow Y$
between the points of two metric spaces $(X,d_X)$ and $(Y,d_Y)$.
A metric embedding $f$ is said to be \emph{dominating} if for every pair of
points $x,y\in X$, it holds that $d_X(x,y)\le d_Y(f(x),f(y))$.
The distortion of a dominating embedding $f$ is $\max_{x \not= y\in
X}\frac{d_Y(f(x),f(y))}{d_X(x,y)}$.
Here we will also study a more permissive generalization of metric embedding
introduced by Cohen-Addad \etal \cite{CFKL20}, which is called
\emph{one-to-many} embedding.
\begin{definition}[One-to-many embedding]\label{def:one-to-many}
	A \emph{one-to-many embedding} is a function $f:X\rightarrow2^Y$  from the
	points of a metric space $(X,d_X)$ into non-empty subsets of points of a
	metric space $(Y,d_Y)$, where the subsets $\{f(x)\}_{x\in X}$ are
	disjoint.	For a point $x'\in Y$,
	$f^{-1}(x')$ denotes the unique point $x\in X$ such that $x'\in f(x)$. If
	no such point exists, $f^{-1}(x')=\emptyset$.
	A point $x'\in f(x)$ is called a \emph{copy} of $x$, while $f(x)$ is
	called the \emph{clan} of $x$.
	For a subset $A\subseteq X$ of vertices, denote $f(A)=\cup_{x\in A}f(x)$.
	We say that $f$ is \emph{dominating} if for every pair of points $x,y\in
	X$, it holds that $d_X(x,y)\le \min_{x'\in f(x),y'\in f(y)}d_Y(x',y')$.
\end{definition}

Here we will study the new notion of clan embeddings introduced by the author
and Le \cite{FL21}.%
\begin{definition}[Clan embedding]\label{def:clan}
	A clan embedding from metric space $(X,d_X)$ into a metric space $(Y,d_Y)$
	is a pair $(f,\chi)$ where $f:X\rightarrow2^{Y}$ is a dominating
	one-to-many embedding, and $\chi:X\rightarrow Y$ is a classical embedding,
	where for every $x\in X$, $\chi(x)\in f(x)$.  $\chi(x)$ is referred to as
	the chief of the clan of $x$ (or simply the chief of $x$).
	We say that the clan embedding $f$ has distortion $t$ if for every
	$x,y\in X$, \mbox{$\min_{y'\in f(y)}d_{Y}(y',\chi(x))\le t\cdot
	d_{X}(x,y)$}.
\end{definition}

\paragraph{Bounded hop distances}
We say that an embedding $f:V\rightarrow X$ has distortion $t$, and
\emph{hop-stretch} $\red{\beta}$, if for every $u,v\in V$ it holds that
$$d_G^{\red{(\beta h)}}(u,v)\le d_X(f(u),f(v))\le t\cdot d_G\rhop{h}(u,v)~.$$
\cite{HHZ21} showed an example of a graph $G$ where in every classical embedding
of $G$ into a metric space, either the hop-stretch $\red{\beta}$, or the
distortion $t$ must be polynomial in $n$. In particular, there is no hope for
a classical embedding (or even a stochastic one) with both hop-stretch and distortion
being sub-polynomial.

We will therefore study hop-constrained clan embeddings, and Ramsey
type embeddings.
\begin{definition}[Hop-distortion of Ramsey type embedding]\label{def:Ramsey}
	\sloppy An embedding $f$  from a weighted graph $G=(V,E,w)$ to a metric
	space $(X,d_X)$  has Ramsey hop distortion $(t,M,\red{\beta},\red{h})$ if
	$M\subseteq V$, for every $u,v\in V$  it holds that $d_G^{\red{(\beta\cdot
	h)}}(u,v)\le d_X(f(v),f(u))$, and for every $u\in V$ and $v\in M$,
	$d_X(f(v),f(u))\le t\cdot d_G\rhop{h}(u,v)$.
\end{definition}
\begin{definition}[Hop-distortion of clan embedding]
	\sloppy A clan embedding $(f,\chi)$ from a weighted graph $G=(V,E,w)$ to a
	metric space $(X,d_X)$ has hop distortion $(t,\red{\beta},\red{h})$ if for
	every $u,v\in V$,  it holds that $d_G^{\red{(\beta\cdot
	h)}}(u,v)\le\min_{v'\in f(v),u'\in f(u)}d_X(v',u')$, and $\min_{v'\in
	f(v)}d_X(v',\chi(u))\le t\cdot d_G\rhop{h}(u,v)$.
\end{definition}

The following definitions will be useful to argue that the clan embedding
preserves properties of subgraphs, and not only vertex pairs.

\begin{definition}[$\red{h}$-respecting]\label{def:hRespecting}
	A subgraph $H$  of $G$ is $\red{h}$-\emph{respecting} if for every $u,v\in
	H$ it holds that $d\rhop{h}_G(u,v)\le d_H(u,v)$. Often we will abuse
	notation and say that a path $P$ is $\red{h}$-\emph{respecting}, meaning
	that the subgraph induced by the path is $\red{h}$-respecting.
\end{definition}
\begin{definition}[Hop-path-distortion]\label{def:HopPathDist}
	We say that the one-to-many embedding $f: V \to Y$ has \emph{hop path
	distortion} $(t,\red{h})$ if for every $\red{h}$-respecting path
	$P=\left(v_0,v_1,\dots,v_m\right)$ there are vertices $v'_0,\dots,v'_m$
	where $v'_i\in f(v_i)$ such that $\sum_{i=0}^{m-1}d_Y(v'_i,v'_{i+1})\le t
	\cdot \sum_{i=0}^{m-1}w(v_i,v_{i+1})$.
\end{definition}

The $\red{h}$-respecting condition in \Cref{def:HopPathDist} is necessary as
it rules out low-weight paths that use more than $\red{h}$ hops. Such paths
cannot generally be preserved by an h.c. embedding, since vertices with large
or infinite $\red{h}$-hop distance in $G$ may be placed far apart in the
embedding. On the other hand, the condition is useful as it allows us to apply
path-distortion to connected subgraphs: after doubling the edges of an
$\red{h}$-respecting subgraph, its Euler tour is again an
$\red{h}$-respecting path.

\paragraph{Minimax principle.}
We will use the following standard finite-dimensional form of von Neumann's
minimax theorem; see, e.g., \cite{OR94}. For a finite set $S$, let
$\mathsf{Dist}(S)$ denote the set of probability distributions over $S$.
Let $\Omega,\mathcal{A}$ be finite sets, and let
$\Phi:\mathcal{A}\times\Omega\rightarrow\mathbb{R}$
be a payoff function. Then
\[
	\max_{\cD\in\mathsf{Dist}(\mathcal{A})}\min_{\mu\in\mathsf{Dist}(\Omega)}
	\mathbb{E}_{A\sim\cD,\,\omega\sim\mu}[\Phi(A,\omega)]=\min_{\mu\in
	\mathsf{Dist}(\Omega)}\max_{A\in\mathcal{A}}\mathbb{E}_{\omega\sim\mu}[
	\Phi(A,\omega)]~.
\]
In particular, if for every distribution $\mu\in\mathsf{Dist}(\Omega)$
there is an object $A\in\mathcal{A}$ such that
$\mathbb{E}_{\omega\sim\mu}[\Phi(A,\omega)]\ge \alpha$,
then there is a distribution $\cD\in\mathsf{Dist}(\mathcal{A})$ such
that for every $\omega\in\Omega$,
$\mathbb{E}_{A\sim\cD}[\Phi(A,\omega)]\ge\alpha$.
Similarly,
\[
\min_{\cD\in\mathsf{Dist}(\mathcal{A})}\max_{\mu\in\mathsf{Dist}(\Omega)}
\mathbb{E}_{A\sim\cD,\,\omega\sim\mu}[\Phi(A,\omega)]=\max_{\mu\in
\mathsf{Dist}(\Omega)}\min_{A\in\mathcal{A}}\mathbb{E}_{\omega\sim\mu}[
\Phi(A,\omega)]~.
\]
In particular, if for every distribution $\mu\in\mathsf{Dist}(\Omega)$
there is an object $A\in\mathcal{A}$ such that
$\mathbb{E}_{\omega\sim\mu}[\Phi(A,\omega)]\le \beta$,
then there is a distribution $\cD\in\mathsf{Dist}(\mathcal{A})$ such
that for every $\omega\in\Omega$,
$\mathbb{E}_{A\sim\cD}[\Phi(A,\omega)]\le \beta$.

\section{Paper overview}
\addtocontents{toc}{\protect\setcounter{tocdepth}{1}}
\subsection{Ramsey type embeddings}
\paragraph{Previous approach.}
The construction of Ramsey type embeddings in \cite{HHZ21} is based on padded
decompositions such as in \cite{Bar96,Fil19padded}. Specifically, Haeupler
\etal show that for every parameter $\Delta$, one can partition the vertices
into clusters $\mathcal{C}$ such that for every $C\in\mathcal{C}$,
$\diam\rhop{O(\log^2 n)\cdot h}(C)=\max_{u,v\in C}d_G\rhop{O(\log^2 n)\cdot
h}(u,v)\le \Delta$ and every ball $B\rhop{h}(v,r)=\{u\mid d\rhop{h}_G(u,v)\le
r\}$, for $r=\Omega(\frac{\Delta}{\log^2n})$, belongs to a single cluster
(i.e. $v$ is ``padded'') with probability $\approx1-\frac{1}{\log n}$.
\cite{HHZ21} then recursively partitions the graph to create clusters with
geometrically decreasing diameter $\Delta$. The resulting hierarchical
partition defines an ultrametric, where the vertices that are padded in all
the levels belong to the set $M$. By a union bound each vertex belongs to $M$
with probability at least $\frac12$, and the distortion and hop-stretch equal
to the parameters of the decomposition: $O(\log^2n)$.

Our goal here is to improve the different parameters: stretch, hop-stretch,
and inclusion probability. We begin by optimizing the trade-off between
stretch and inclusion probability (which mostly resembles classic non
hop-constrained work). Then we will further optimize the hop-stretch (by
introducing small slack to the stretch).

\paragraph{\Cref{thm:UltrametricRamsey}: optimal distortion - inclusion
probability tradeoff.}
We use a more sophisticated approach, similar to previous works on clan
embeddings \cite{FL21}, and the deterministic construction of Ramsey trees
\cite{ACEFN20}.
The minimax theorem is a standard tool for converting an average-type
guarantee into a distribution that is good for each element separately.
Informally, if for every way of averaging over the vertices there exists
an embedding whose saved set has a large portion of the total mass, then there
is a distribution over embeddings in which every vertex is saved with large
probability.
In our setting, the averaging is given by a probability measure
$\mu:V\rightarrow[0,1]$, and the payoff of an embedding with saved set $M$ is
$\mu(M)$. Thus, for the minimax step, it is enough to show that for
every probability measure $\mu$, there exists a single embedding for
which $\mu(M)$ is large.
As a corollary, we then obtain a distribution where the inclusion probability
for each individual vertex is accordingly large.

Rather than prove this probability-measure statement directly, we prove
a slightly more convenient version for $(\ge1)$-measures, that is, where the
measure of each vertex is at least $1$. The lower bound $\mu(v)\ge1$
is useful in the recursive ball-growing proof (where singleton clusters should
not have vanishing mass).
Specifically, given a parameter $k$ and a measure
$\mu:X\rightarrow\R_{\ge1}$, \Cref{lem:UltraMeasure} constructs a
Ramsey type embedding with Ramsey hop-distortion
$(16k,M,\red{O(k\cdot\log n)},\red{h})$ such that
$\mu(M)\ge \mu(X)^{1-\frac1k}$. This lemma is later converted back to
probability measures and fed into the minimax argument.
The construction of the distributional \Cref{lem:UltraMeasure} is
a deterministic recursive ball growing algorithm.
Our algorithm closely resembles the Ramsey (non-spanning) trees construction
from \cite{ACEFN20} (Algorithm 1 in the Journal version). In particular, if
one sets $\red{h}=\infty$, up to some technical differences, it will be the
same algorithm.

Given a cluster $X$ and scale parameter $2^i$, we partition $X$ into clusters
$X_1,\dots,X_s$ such that each $X_i$ is contained in a hop-bounded ball
$B_{G[X]}^{\red{(i\cdot O( k)\cdot h)}}(v,2^{i-1})$.
Then we recursively construct a Ramsey type embedding for each cluster $X_q$
(with scale $2^{i-1}$), and combine the obtained ultrametrics using a new root
with label $2^i$.
Note that there is a dependence between the radius of the balls and the number
of hops they are allowed. Specifically, in each level we decrease the radius
by a multiplicative factor of $2$, and the number of hops by an additive factor
of $\red{h'=O(k)\cdot h}$.
We maintain a set of active vertices $M\subseteq X$. Initially all the
vertices are active. An active vertex $v\in M$ will cease to be active if the
ball $B_G\rhop{h}(v,\Theta(\frac{2^i}{k}))$ intersects with more than one
cluster. The set of vertices that remain active until the end of the algorithm,
i.e. vertices that were ``padded'' in every level, will constitute the set $M$
in \Cref{lem:UltraMeasure}.
The ``magic'' happens in the creation of the clusters
(\Cref{alg:CreateCluster}), so as to ensure that a large enough portion of the
vertices remain active until the end of the algorithm.
On an (inaccurate and) intuitive level, inductively, at level $i-1$, the
``probability'' of a vertex $v$ to be padded in all future levels is
$\left(|B^{\red{(i\cdot h')}}(v,2^{i-3})|\right)^{-\frac1k}$.
By a ball growing argument, $v$ is padded in the $i$'th level with
``probability'' $\left(\frac{|B^{\red{((i+1)\cdot
h')}}(v,2^{i-2})|}{|B^{\red{(i\cdot h')}}(v,2^{i-3})|}\right)^{-\frac1k}$.
Hence by the induction hypothesis, $v$ is padded at level $i$ and all the
future levels with probability $\left(|B^{\red{((i+1)\cdot
h')}}(v,2^{i-2})|\right)^{-\frac1k}$. Note that for these cancellations to work
out, we are paying an $\red{h'=O(k\cdot h)}$ additive factor in the
hop-stretch in each level, as we need that the ``minimum possible cluster''
(among the possibly considered clusters, minimum w.r.t. containment) at the
current level will contain the ``maximum possible cluster'' (among the
possibly considered clusters, maximum w.r.t. containment) in the next one.
In other words, consider an $i$-level cluster $X$ which contains an
$i-1$-level cluster $Y$. $X$ and $Y$ were chosen from a set of $O(k)$ nested
clusters $X\in \left\{A^X_{0}\subseteq A^X_{1}\subseteq\cdots \subseteq
A^X_{2k}\right\}$, and $Y\in\left\{A^Y_{0}\subseteq A^Y_{1}\subseteq\cdots
\subseteq A^Y_{2k}\right\}$. Our analysis requires that $A^Y_{2k}\subseteq
A^X_{0}$ (regardless of the $A^Y_{j}$ implemented).

\paragraph{\Cref{thm:UltrametricRamseyAlt}: $\log\log n$ hop-stretch.}
The construction of the embedding for \Cref{thm:UltrametricRamseyAlt} is the
same as that of \Cref{thm:UltrametricRamsey} in all aspects other than the
cluster creation. In particular, we first prove a distributional lemma
(\Cref{lem:UltraMeasureAlt}) and use the minimax theorem.
The analysis of the algorithm of \Cref{thm:UltrametricRamsey} required a very
strong ``containment of potential created clusters'' guarantee
($A^Y_{2k}\subseteq A^X_{0}$ above).
This dependence between the different scale levels contributes an additive
factor of $\red{O(kh)}$ to the hop allowance per level. Thus for polynomial
aspect ratio it has $\red{O(k\cdot\log n)}$ hop-stretch.

To decrease the hop-stretch, here we create clusters w.r.t. ``hop diameter''
at most $\red{O(k\cdot\log\log n)\cdot h}$, regardless of the current scale.
As a result, it is harder to relate between the different levels (as the
``maximum possible cluster'' in the next level is not contained in the
``minimum possible one'' in the current level).
To compensate for that, instead, we use a stronger per-level guarantees.
A similar phenomenon where one cannot directly compare consecutive levels also
appeared in the spanning (sub-graph) Ramsey trees of Abraham \etal
\cite{ACEFN20} (see also \cite{AN19,FL21}). In particular, our algorithm for
\Cref{thm:UltrametricRamseyAlt} is morally similar to their algorithm 4 (in
the journal version). However, as \cite{ACEFN20} implemented their algorithm
in the petal decomposition framework, the technicalities are quite different.
An additional advantage of this approach is that it holds for graphs with
unbounded aspect ratio (as the hop-stretch is unrelated to the
number of scales).

In more detail, at level $i$ we create a ball cluster inside
$B^{\red{(O(k\cdot\log\log n)\cdot h)}}(v,2^{i-2})$, however instead of using
the ball growing argument on the entire spectrum of possibilities, we use it only
in a smaller ``strip'', losing a $\log\log n$ factor, but obtaining that the
probability of $v$ to be padded in the current $i$ level is
$\left(\frac{|B^{\red{(O(k\cdot\log\log n)\cdot
h)}}(v,2^{i-2})|}{|X|}\right)^{\frac{1}{k}}$. Inductively we assume that the
``probability'' of a vertex $v\in X$ to be padded in all the levels is
$|X|^{-\frac1k}$. As the cardinality of the resulting cluster is smaller than
$\left|B^{\red{(O(k\cdot\log\log n)\cdot h)}}(v,2^{i-2})\right|$, the
probability that $v$ is padded in all the levels (including $i$) is indeed
$|X|^{-\frac1k}$.

\subsection{Clan embeddings}
The constructions of the clan embeddings
\Cref{thm:ClanHopUltrametric,thm:ClanHopUltrametricAlt} are similar to the
Ramsey type embeddings. In particular, in both cases the key step is a
distributional lemma
(\Cref{lem:clanHopUltraMeasure,lem:clanHopUltraMeasureAlt}), where the
algorithms for the distributional lemmas are deterministic ball
growing algorithms.
The main difference is that while in the Ramsey type embeddings we partition
the graph into clusters,
in the clan embedding we create a cover.
Specifically, given a cluster $X$ and scale $2^i$, the created cover is a
collection of clusters $X_1,\dots,X_s$ such that each $X_i$ is contained in a
hop-bounded ball $B_{G[X]}^{\red{(O(i\cdot k)\cdot h)}}(v,2^{i-2})$, and every
ball $B_{G[X]}\rhop{h}(u,\Omega(\frac{2^i}{k}))$ is fully contained in some
cluster. Each vertex $u$ might belong to several clusters.
In the Ramsey type embedding our goal was to minimize the measure of the
``close to the boundary'' vertices, as they were deleted from the set $M$.
On the other hand, here the goal is to bound the combined multiplicity of all
the vertices (w.r.t. the measure $\mu$), which governs the clan sizes.

\paragraph{Path distortion and subgraph preserving embedding}
An additional property of the clan embedding is path distortion.
Specifically, given an ($\red{h}$-respecting) path
$P=\left(v_0,v_1,\dots,v_m\right)$ in $G$, we are guaranteed that there are
vertices $v'_0,\dots,v'_m$ where $v'_i\in f(v_i)$ such that
$\sum_{i=0}^{m-1}d_U(v'_i,v'_{i+1})\le O(k\cdot\log n) \cdot
\sum_{i=0}^{m-1}w(v_i,v_{i+1})$.
The proof of the path distortion property (\Cref{lem:ClanPathDistortion}) is
recursive and follows similar lines to \cite{BM04multi}.
Specifically, for a scale $2^i$ we iteratively partition the path vertices
$v_0,v_1,\dots,v_m$ into the clusters $X_1,\dots,X_s$, while minimizing the
number of ``split edges'': consecutive vertices assigned to different
clusters. We inductively construct ``copies'' for each sub-path internal to a
cluster $X_j$. The ``split edges'' are used as evidence that the path $P$
has weight comparable to the scale $2^i$ times the number of split edges, and
thus we can afford them (as the maximal price of an edge in the
ultrametric is $2^i$).

Next, we construct the subgraph preserving embedding. Specifically, we
construct a clan embedding $f:V\rightarrow 2^T$, such that all the vertices in
$T$ are copies of $G$ vertices: $V(T)=\cup_{v\in V}f(v)$, and every edge
$e=\{u',v'\}\in T$ for $u'\in f(u)$, $v'\in f(v)$ is associated with a path
$P_e$ from $u$ to $v$ in $G$ of weight $\le d_T(u',v')$. By concatenating edge
paths we obtain an induced path for every pair of vertices $u',v'$ in $T$.
In \Cref{cor:hBoundedSubgraph} we construct such a subgraph preserving
embedding, where the number of hops in every induced path is bounded by
$\red{O(\log^2n)\cdot h}$, and such that for every connected
($\red{h}$-respecting) subgraph $H$ of $G$, there is a connected subgraph $H'$
of $T$ of weight $w_T(H')\le O(\log n)\cdot w_G(H)$, where for every $v\in H$,
$H'$ contains some vertex from $f(v)$. In \Cref{cor:hBoundedSubgraphAlt} we
obtain an alternative tradeoff where the number of hops in every induced path
is bounded by $\red{\tilde{O}(\log n)\cdot h}$, and the weight of the subgraph
$H'$ is bounded by $w_T(H')\le \tilde{O}(\log n)\cdot w_G(H)$.

Finally, we obtain a ``subgraph-preserving embedding'' for arbitrary (that is
not necessarily $\red{h}$-respecting) subgraphs
(\Cref{thm:Subgraph,thm:SubgraphAlt}).
Specifically, given an arbitrary subgraph $H$ of $G$, we are guaranteeing that $T$
will contain a (not necessarily connected) subgraph $H'$ of $T$ of weight at
most $w_T(H')\le O(\log n)\cdot w_G(H)$, such that for every pair of points
$u,v$ where $\hop_H(u,v)\le\red{h}$, $H'$ will contain two vertices from
$f(u)$ and $f(v)$ in the same connected component.

A key component in the proof of \Cref{thm:Subgraph} is a construction of a
\emph{sparse cover} (\Cref{lem:SparseCover}). Here we are given a parameter
$\Delta$, and partition the graph into clusters of diameter $O(\log
n)\cdot\Delta$, such that every ball of radius $\Delta$ is contained in some
cluster, and the total sum of weights of all the cluster-induced subgraphs is
at most constant times larger than the weight of $G$.

\subsection{Approximation algorithms}
Following the approach of \cite{HHZ21}, our improved Ramsey type embeddings
(\Cref{thm:UltrametricRamsey,thm:UltrametricRamseyAlt}) imply improved
approximation algorithms for all the problems studied in \cite{HHZ21}. We
refer to \cite{HHZ21} for details, and to \Cref{cor:HHZ21Application} (and
\Cref{tab:ApproxAlgs}) for a summary.

For the group Steiner tree problem (and its online version), using the
subgraph preserving embedding, we are able to obtain a significant improvement
beyond the techniques of \cite{HHZ21}. In fact, the cost of our approximation
algorithm matches the state of the art for the non-hop-restricted case.
Our approach is as follows: first we observe that the optimal solution $H$ for
the h.c. group Steiner tree problem is a tree, and hence
$\red{2h}$-respecting. It follows that a copy $H'$ of the optimal solution $H$
(of weight $O(\log n)\cdot \opt$) could be found in the image $T$ of the
subgraph preserving embedding $f$. Next we use known (non h.c.) approximation
algorithms for the group Steiner tree problem on trees \cite{GKR00} to find a
valid solution $H''$ in $T$ of weight $O(\log^2 n\cdot\log k)\cdot \opt$.
Finally, by combining the associated paths we obtain an induced solution in
$G$ of the same cost. The bound on the hops in the induced paths ensures an
$\red{O(\log^2n)}$ hop-stretch.
The competitive algorithm for the online group Steiner tree problem follows
similar ideas.
A similar approach is then used on h.c. group Steiner forest (and its online
version). However, as the optimal solution is not necessarily
$\red{h}$-respecting, we use \Cref{thm:Subgraph} instead of
\Cref{cor:hBoundedSubgraph}.

\subsection{Hop-constrained metric data structures}
At first glance, as hop-constrained distances are non-metric, constructing
metric data structures for them seems to be very challenging, and indeed, the
lack of any previous work on this very natural problem serves as evidence.

Initially, using our Ramsey type embedding we construct distance labelings in
a similar fashion to the distance labeling scheme of \cite{MN07}: construct
$O(n^{\frac1k}\cdot k)$ Ramsey type embeddings into ultrametrics
$U_1,U_2,\dots$ such that $U_i$ is ultrametric over $V$ with the set $M_i$ of
saved vertices. It will hold that every vertex belongs to some set $M_i$.
For trees (and ultrametrics) there are very efficient distance labeling
schemes (constant size and query time, and $1.5$-distortion \cite{FGNW17}).
We define the label of each vertex $v$ to be the union of the labels for the
ultrametrics, and the index $i$ of the ultrametric where $v\in M_i$.
As a result we obtain a distance labeling with label size $O(n^{\frac1k}\cdot
k)$, constant query time, and such that for every $u,v$,
$d^{\red{(O(k\cdot\log\log n)\cdot h)}}_{G}(u,v)\le \A(\ell(v),\ell(u))\le
O(k\cdot\log\log n)\cdot  d\rhop{h}_{G}(u,v)$.
While the space, query time, and hop-stretch are quite satisfactory, the
distortion leaves much to be desired.

To obtain an improvement, our next step is to construct data structures for
a fixed scale. Specifically, for every scale $2^i$, we construct a graph $G_i$
such that for every pair $u,v$ where $d_G\rhop{h}(u,v)\approx 2^i$, it holds
that $d_G^{\red{(O(\frac{k\cdot\log\log n}{\eps})\cdot h)}}(u,v)\le
d_{G_i}(u,v)\le (1+\eps)\cdot(2k-1)\cdot d_G\rhop{h}(u,v)$.
Thus, if we were to know that $d_G\rhop{h}(u,v)\approx 2^i$, then we could
simply use state of the art distance labeling for $G_i$, ignoring hop
constraints. Now, using our first Ramsey based distance labeling scheme we can
obtain the required coarse approximation!
Our final construction consists of Ramsey based distance labeling, and in
addition state of the art distance labeling \cite{TZ05} for the graphs $G_i$
for all possible distance scales. Given a query $u,v$, we first coarsely
approximate $d_G\rhop{h}(u,v)$ to obtain some value $2^i$. Then we return
the answer from the distance labeling scheme that was prepared in
advance for $G_i$.
Our approach for distance oracles, and compact routing schemes is similar.

\addtocontents{toc}{\protect\setcounter{tocdepth}{3}}

\section{Hop-constrained Ramsey type embedding}\label{sec:Ramsey-Partitions}
\addtocontents{toc}{\protect\setcounter{tocdepth}{2}}
We restate \Cref{thm:UltrametricRamsey} for convenience:
\RamseyType*

First, we will prove a ``distributional'' version of
\Cref{thm:UltrametricRamsey}. That is, we will receive a distribution $\mu$
over the points, and deterministically construct a single Ramsey type
embedding such that $\mu(M)$ will be large. Later, we will use the Minimax
theorem to conclude \Cref{thm:UltrametricRamsey}.
We begin with some definitions: a \emph{measure} over a finite set $X$, is
simply a function $\mu:X\rightarrow \R_{\ge 0}$. The measure of a subset
$A\subseteq X$, is $\mu(A)=\sum_{x\in A}\mu(x)$.
We say that $\mu$ is a \emph{probability measure} if $\mu(X)=1$. We say that
$\mu$ is a \emph{$(\ge1)$-measure} if for every $x\in X$, $\mu(x)\ge1$.

\begin{lemma}\label{lem:UltraMeasure}
	Consider an $n$-vertex weighted graph $G=(V,E,w)$ with polynomial aspect
	ratio, $(\ge1)$-measure $\mu:V\rightarrow\mathbb{R}_{\ge1}$,
	integer parameters $k\ge1$, $h\in [n]$ and subset
	$\mathcal{M}\subseteq V$.
	Then there is a Ramsey type embedding into an ultrametric $U$ with Ramsey
	hop distortion
	$(16k,M,\red{O(k\cdot\log n)},\red{h})$ where $M\subseteq \mathcal{M}$ and
	$\mu(M)\ge \mu(\mathcal{M})^{1-\frac1k}$.
\end{lemma}
Note that \Cref{lem:UltraMeasure} guarantees that for every $u,v\in V$,
$d_G^{\red{(O(k\cdot\log n)\cdot h)}}(u,v)\le d_U(u,v)$. Further, if $u\in M$,
then $d_U(u,v)\le 16k\cdot d_G\rhop{h}(u,v)$.
The freedom to choose $\mathcal{M}$ will be useful in our construction of
metric data structures (but will not be used during the proof of
\Cref{thm:UltrametricRamsey}).
The proof of \Cref{lem:UltraMeasure} is deferred to
\Cref{subsec:ProofOfUltraMeasure}. Now we will prove
\Cref{thm:UltrametricRamsey} using \Cref{lem:UltraMeasure}.
\begin{proof}[Proof of \Cref{thm:UltrametricRamsey}]
	We can assume that $k<n$ (and also $\eps>\frac1n$), as otherwise we are
	allowed $\Omega(n)$ hops and stretch. Thus we can ignore the hop
	constraints, and \Cref{thm:UltrametricRamsey} will simply follow from
	classic Ramsey type embeddings (e.g. \cite{MN07}).

	We begin with an auxiliary claim, which translates from the language of
 	$(\ge1)$-measures used in \Cref{lem:UltraMeasure} to that of probability
 	measures (while increasing the distortion from $16k$ to $16(k+1)$).
 	\begin{claim}\label{clm:probabilisticRamsey}
 		Consider an $n$-vertex weighted graph $G=(V,E,w)$ with polynomial
 		aspect ratio, probability measure $\mu:V\rightarrow[0,1]$,
 		and integer parameters $k\ge1$, $\red{h}\in [n]$.
 		Then there is a Ramsey type embedding into an ultrametric with Ramsey
 		hop distortion
 		$(16(k+1),M,\red{O(k\cdot\log n)},\red{h})$ where $\mu(M)\ge
 		(k+1)^{-\frac{1}{k}}\left(1-\frac{1}{k+1}\right)\cdot
 		n^{-\frac{1}{k}}$.
 	\end{claim}
\begin{proof}
	Fix $s=\frac{n^{\nicefrac{1}{k}}}{\delta}$ for
	$\delta=(k+1)^{-\frac{k+1}{k}}$.
	Define the following probability
	measure $\widetilde{\mu}$: $\forall v\in V$,
	$\widetilde{\mu}(v)=\frac{1}{sn}+\frac{s-1}{s}\mu(v)$.
	Set the following $(\ge1)$-measure
	$\widetilde{\mu}_{\ge1}(v)=sn\cdot\widetilde{\mu}(v)$. Note that
	$\widetilde{\mu}_{\ge1}(V)=sn$.
	We execute \Cref{lem:UltraMeasure} w.r.t. the $(\ge1)$-measure
	$\widetilde{\mu}_{\ge1}$, parameter $k+1$, and $\mathcal{M}=V$. As a
	result, we obtain a Ramsey type embedding into an ultrametric $U$ with
	Ramsey hop distortion
	$(16(k+1),M,\red{O(k\cdot\log n)},\red{h})$. Further, it holds that
	\[
	(sn)^{1-\frac{1}{k+1}}=\widetilde{\mu}_{\ge1}(V)^{1-\frac{1}{k+1}}~\le~
	\widetilde{\mu}_{\ge1}(M)=sn\cdot\widetilde{\mu}(M)=sn\cdot\left(
	\frac{|M|}{sn}+\frac{s-1}{s}\mu(M)\right)~.
	\]
	As $|M|\le n$, $\frac{s-1}{s}<1$, and using
	$s=\frac{n^{\nicefrac{1}{k}}}{\delta}$ we conclude
	\begin{equation*}
	\mu(M)\ge\frac{s}{s-1}\cdot
	\left((sn)^{-\frac{1}{k+1}}-\frac{|M|}{sn}\right)
	\ge(sn)^{-\frac{1}{k+1}}-\frac{1}{s}
	=\left(\frac{1}{\delta}\cdot n^{1+\frac{1}{k}}\right)^{
	-\frac{1}{k+1}}
	-\delta\cdot n^{-\frac{1}{k}}
	=\left(\delta^{\frac{1}{k+1}}-\delta\right)\cdot
	n^{-\frac{1}{k}}~.%
	\end{equation*}
	By plugging in the value of $\delta$ we have
	\begin{align*}
	\delta^{\frac{1}{k+1}}-\delta &
	=\left((k+1)^{-\frac{k+1}{k}}\right)^{\frac{1}{k+1}}-(k+1)^{-\frac{k+
	1}{k}}=(k+1)^{-\frac{1}{k}}-(k+1)^{-\frac{k+1}{k}}=(k+1)^{-\frac{1}{k}}
	\left(1-\frac{1}{k+1}\right)~.
	\end{align*}
	The claim now follows.
\end{proof}

	We obtain that for every probability measure $\mu$, there is a Ramsey type
	embedding into an ultrametric with the distortion guarantees as above,
	such that the measure of $M$ is at least $\alpha_k\cdot n^{-\frac{1}{k}}$,
	for $\alpha_k=(k+1)^{-\frac{1}{k}}\left(1-\frac{1}{k+1}\right)$.
	Let $\cA$ be the family of candidate pairs $(U,M)$ with the above
	Ramsey hop-distortion guarantees.
	Using the minimax principle from \Cref{sec:prelim}, with payoff
	function $\mu(M)$,
	\[\max_{\cD\in\mathsf{Dist}(\cA)}
	\min_{\mu\in\mathsf{Dist}(V)}
	\mathbb{E}_{(U,M)\sim\cD}[\mu(M)]
	=
	\min_{\mu\in\mathsf{Dist}(V)}
	\max_{(U,M)\in\cA}
	\mu(M)
	\ge \alpha_k\cdot n^{-\frac1k}~.\]

	In words, $\cD$ here is a distribution over pairs $(U,M)$, and $\mu$ is a
	probability measure over the vertices. Given $U,M,\mu$, the payoff is
	$\mu(M)$, the measure of the saved set. \Cref{clm:probabilisticRamsey}
	says that for every probability measure $\mu$, one can find a Ramsey type
	embedding $(U,M)$ with payoff $\mu(M)\ge \alpha_k\cdot n^{-\frac{1}{k}}$.
	Accordingly, by the minimax principle, there is a distribution $\cD$ over
	pairs $(U,M)$ such that for every probability measure $\mu$ over the
	vertices, $\mathbb{E}_{(U,M)\sim\cD}[\mu(M)]\ge\alpha_k\cdot
	n^{-\frac{1}{k}}$.
	Let $\mathcal{D}$ be this distribution from above. For every vertex $v\in
	V$, denote by $\mu_{v}$
	the probability measure where $\mu_v(v)=1$ (and $\mu_v(u)=0$ for $u\ne
	v$). Then when sampling an ultrametric $U\sim\mathcal{D}$ we have:
	\[
	\Pr_{(U,M)\sim\cD}[v\in M]
	=
	\mathbb{E}_{(U,M)\sim\cD}[\mu_v(M)]
	\ge \alpha_k\cdot n^{-\frac1k}~.
	\]
	To conclude the first part of \Cref{thm:UltrametricRamsey}, note that
	$\alpha_k$ is a monotonically increasing function of $k$, hence
	$\alpha_{k}\ge\alpha_{1}=\frac{1}{2}\cdot\frac{1}{2}=\frac{1}{4}$.

	For the second assertion, let $k=\frac{4\cdot\ln n}{\eps}$.
	Note that
	$(k+1)^{-\frac{1}{k}}=e^{-\frac{1}{k}\ln(k+1)}\ge1-\frac{\ln(k+1)}{k}\ge1-
	\frac{\ln(2k)}{k}$, and $1-\frac{1}{k+1}>1-\frac1k$.
	Hence $\alpha_k>\left(1-\frac{\ln(2k)}{k}\right)\left(1-\frac{1}{k}
	\right)>1-\frac{\ln(4k)}{k}$.
	Applying the preceding minimax argument with this value of $k$ gives a
	distribution with the desired distortion guarantees, such that
	every vertex
	belongs to $M$ with probability at least
	\[
	\alpha_{k}\cdot n^{-\frac{1}{k}}>\left(1-\frac{\ln(4k)}{k}\right)\cdot
	n^{-\frac{1}{k}}>\left(1-\frac{\epsilon}{4}\cdot\frac{\ln(4n)}{\ln
	n}\right)\cdot
	e^{-\frac{\epsilon}{4}}>(1-\frac{\epsilon}{2})^{2}>1-\epsilon~,
	\]
	where we used that $k\le n$. The theorem follows.

	Note that we provided here only an existential proof of the desired
	distribution.
	A constructive proof (with polynomial construction time, and $O(n\cdot\log
	n)$ support size) can be obtained using the multiplicative weights update
	(MWU) method.
	The details are similar to the constructive proof of clan embeddings in
	\cite{FL21}, and we will skip them.
\end{proof}

\subsection{Proof of \Cref{lem:UltraMeasure}: distributional h.c. Ramsey type
embedding}\label{subsec:ProofOfUltraMeasure}
For simplicity, initially we assume that $\diam\rhop{h}(G)<\infty$.
Later, in \Cref{subsec:inftyDistance} we will remove this assumption.
Let $\phi=\left\lceil\log_2 \diam\rhop{h}(G) \right\rceil$ be the minimum
integer such that all the $\red{h}$-hop distances are bounded by $2^\phi$.
Recall that we assume that the minimum edge weight is $1$.
As we are assuming here a polynomial aspect ratio, it follows that
$\phi=O(\log n)$.

\sloppy The construction of the ultrametric is illustrated in the recursive
algorithm \texttt{Ramsey-type-embedding}: \Cref{alg:hier-Ramsey}.
We treat here ultrametrics as HSTs (according to \Cref{def:ultra}).
We will have a set of marked vertices which is denoted $M$.
Initially, all the vertices are marked $M=V$.
In order to create the embedding we will make the call
\texttt{Ramsey-type-embedding}$(G[V],\mu,M=V,\red{h},k,2^\phi)$.
At any point when a call
\texttt{Ramsey-type-embedding}$(G[X],\mu,M,\red{h},k,2^i)$ is called, it is
assumed that $\diam^{\red{((i+2)\cdot 4kh)}}(G[X])\le 2^i$, $M\subseteq X$ (it
is the current subset of marked vertices), $\red{h},k\in\N$ are parameters,
and $\mu$ is a $(\ge1)$-measure (over a superset of $X$).
We use the notation $\mu_M(X)=\mu(X\cap M)$ to denote the measure of the
marked vertices in $X$.
A cluster $X$ on which we called \texttt{Ramsey-type-embedding} with scale
parameter $2^i$ is called an $i$-level cluster (corresponding to the
$i$-level of the associated hierarchical cover).
The algorithm \texttt{Ramsey-type-embedding} first partitions the set of
vertices $X$ into clusters $\{X_q\}_{q=1}^s$, each with a subset $M_q\subseteq
X\cap M$ of marked vertices. Note that some vertices in $M$ might
become unmarked.
Each cluster $X_q$ will have diameter at most $\diam^{\red{((i+2)\cdot
4kh)}}(G[X_q])\le2^{i-1}$.
Note that as the algorithm progresses, both the diameter, and the allowed
number of hops decrease. This is used to balance the number of vertices that stay
marked during the entire execution of the algorithm.
For every $q\in[1,s]$, the algorithm will perform a recursive call on the
induced graph $G[X_q]$ with $M_q$ as the set of marked vertices, and scale
$2^{i-1}$, to obtain an ultrametric $U_q$ over $X_q$. In particular each such
cluster $X_q$ is an $i-1$-level cluster.
Then all these ultrametrics are joined to a single ultrametric root at $r_U$
with label $2^i$.

The algorithm \texttt{Padded-partition} partitions a cluster $X$ into clusters
$\{X_q\}_{q=1}^s$ iteratively.
When no marked vertices remain, the algorithm simply partitions the remaining
vertices into singletons and is done. Otherwise, while the remaining graph is
nonempty, the algorithm carves a new cluster $X_q$ using a call to the
\texttt{Create-Cluster} procedure. The vertices near the boundary of $X_q$ are
unmarked, where $M_q$ is the set of marked vertices in the interior of $X_q$. The
remaining graph is denoted $Y_q$ (initially $Y_0=X$), where $M_q^Y$ is the set
of remaining marked vertices (not in $X_q$ or near its boundary).
See \Cref{fig:RamseyAlgExample} for an illustration.

The \texttt{Create-Cluster} procedure returns a triplet
 $(A_{j(v)},A_{j(v)+1},A_{j(v)+2})$, 
 where $A_{j(v)+1}$ is the cluster itself (denoted
$X_{q}$ in the \texttt{Padded-partition} procedure), 
$A_{j(v)}$ (also denoted $\underline{X_q}$) is the interior of the cluster
(all vertices far from the boundary), and $A_{j(v)+2}$ (also denoted $\overline{X_q}$)
is the cluster itself plus its exterior (vertices not in $A_{j(v)+1}=X_q$, but close to
its boundary).
The procedure starts by picking a center $v$, which maximizes the measure of
the marked nodes in the ball $B_{H}^{\red{(i\cdot h')}}(v,2^{i-3})$. This
choice will later be used in the inductive argument lower bounding the overall
number of remaining marked vertices.
The cluster is chosen somewhere between $B_{H}^{\red{((i+1)\cdot
h')}}(v,2^{i-3})$ and $B_{H}^{\red{(i\cdot h')}}(v,2^{i-2})$, where
\red{$h'=2kh$}.
See \Cref{fig:RamseyAlgExample} for an illustration.

During the construction of the ultrametric, we implicitly create a
hierarchical partition $\{{\cal X}_{\logdiam},{\cal
X}_{\logdiam-1},\dots,{\cal X}_0\}$ of $X$.
That is, each $\mathcal{X}_i$ is a set of disjoint clusters, where
$V=\cup_{X\in \mathcal{X}_i}X$, and each $X\in \mathcal{X}_i$ fulfills
$\diam^{\red{((i+2)\cdot h')}}(G[X])\le 2^{i}$.
Further, $\mathcal{X}_{i-1}$ refines $\mathcal{X}_i$, that is, for every
cluster $X\in \mathcal{X}_{i-1}$, there is a cluster $X'\in \mathcal{X}_i$
such that $X\subseteq X'$. Finally ${\cal X}_0=\{\{v\}\mid v\in V\}$ is the
set of all singletons, while ${\cal X}_\logdiam=\{V\}$ is the
trivial partition.
In our algorithm, $\mathcal{X}_i$ will be all the clusters that were created
by calls to the \texttt{Padded-partition} procedure with scale $2^i$, in any
step of the recursive algorithm.

Denote $\rho_i=\frac{2^{i-3}}{2k}$. For a node $v$ and index $i$, we say that
$v$ is $i$-\emph{padded},
if there exists a subset $X \in {\cal X}_i$, such that
$B_{G}\rhop{h}(v,\rho_i) \subseteq X$.
We will argue that a vertex is marked at stage $i$, (i.e. belongs to $M$) only
if it was padded in all the previous steps of the algorithm.
We would like to maximize the measure of the nodes that are padded in
all the levels.
We denote by $M$ the set of \emph{marked} vertices. Initially $M=V$, and
iteratively the algorithm \emph{unmarks} some of the nodes.
The nodes that will remain marked by the end of the process are the nodes that
are padded on all levels.

\begin{algorithm}[p]
	\caption{$U=\texttt{Ramsey-type-embedding}(G[X],\mu,M,\red{h},k,
	2^i)$}\label{alg:hier-Ramsey}
	\DontPrintSemicolon
	\SetKwInOut{Input}{input}\SetKwInOut{Output}{output}
	\Input{Induced graph $G[X]$, parameters $\red{h},k\in\N$, $\ge1$-measure
	$\mu$, set of padded vertices $M$, scale $i$.}
	\Output{Ultrametric $U$ with $X$ as leaves.}
	\If{$|X|=1$}{\Return $G[X]$}
	let $\left(\left( X_{q},M_q\right)
	_{q=1}^{s}\right)=\texttt{Padded-partition}(G[X],\mu,M,\red{h},k,2^i)$\;
	\For{each $q\in[1,\dots,s]$}{
		$U_{q}=\texttt{Ramsey-type-embedding}(G[X_q],\mu,M_q,\red{h},k,
		2^{i-1})$\;
	}
	let $U$ be the HST formed by taking a root $r_U$ with label $2^{i}$, and
	set the ultrametrics $U_{1},\dots,U_{s}$ as its children\;
	\Return $U$\;
\end{algorithm}

\begin{algorithm}[p]
	\caption{$\left(\left(
	X_{q},M_q\right)_{q=1}^{s}\right)$=
	\texttt{Padded-partition}$(G[X],\mu,M,\red{h},k,
	2^i)$}\label{alg:RamseyPartition}
	\DontPrintSemicolon
	\SetKwInOut{Input}{input}\SetKwInOut{Output}{output}
	\Input{Induced graph $G[X]$, parameters $\red{h},k\in\N$, $\ge1$-measure
	$\mu$, set of padded vertices $M$, scale $i$.}
	\Output{Partition $\{X_{q}\}_{q=1}^{s}$ of $X$, and sets $M_q\subseteq
	M\cap X_q$.}
	\BlankLine

	set $Y_0=X$, $M^Y_0=M$, and $q=1$\;
	\While{$Y_{q-1}\ne\emptyset$}{
		\If{$M^Y_{q-1}=\emptyset$}{
			partition $Y_{q-1}$ into singleton clusters $\{X_{q'}\}_{q'\ge
			q}$, where $M_{q'}=\emptyset$ for each such $X_{q'}$\;
			skip to \cref{line:returnPaddedPartition}\;
		}
		\Else{
			$(\underline{X_q},X_q,\overline{X_q})=
			\texttt{Create-Cluster}(G[Y_{q-1}],M^Y_{q-1},\mu,\red{h},k,
			2^i)$\label{line:callCreateCluster}\;
			set $M_q=M\cap \underline{X_q}$\;
			set $Y_q=Y_{q-1}\backslash X_q$ and $M^Y_{q}\leftarrow
			M^Y_{q-1}\backslash \overline{X_q}$\tcp*{all the nodes in
			$\overline{X}_q\backslash\underline{X}_q$ are unmarked.}

			$q\leftarrow q+1$\;
		}
	}
	\Return $\left(\left(
	X_{q},M_q\right)_{q=1}^{s}\right)$\label{line:returnPaddedPartition}
	\tcp*{$s$ is the number of created clusters}
\end{algorithm}

\begin{algorithm}[p]
	\caption{$(\underline{A},A,\overline{A})$=\texttt{Create-Cluster}$(H=G[X],
	M,\mu,\red{h},k,2^i)$}\label{alg:CreateCluster}
	\DontPrintSemicolon
	\SetKwInOut{Input}{input}\SetKwInOut{Output}{output}
	\Input{Induced graph $G[X]$, parameters $\red{h},k\in\N$, $\ge1$-measure
	$\mu$, set of padded vertices $M$, scale $i$.}
	\Output{Three clusters $(\underline{A},A,\overline{A})$, where
	$\underline{A}\subseteq A\subseteq \overline{A}$.}
	\BlankLine
	pick a node $v \in X$ with maximal $\mu_{M}\left(B_{H}^{\red{(i\cdot
	h')}}(v,2^{i-3})\right)$ for
	$\red{h'}=\red{2kh}$\label{line:CreateClusterChooseCenter}\;
	denote $A_{j}=B_{H}^{\red{(i\cdot h'+j\cdot
	h)}}\left(v,2^{i-3}+j\cdot\rho_i\right)$, for
	$\rho_i=\frac{2^{i-3}}{2k}=\frac{2^i}{16\cdot k}$\;
	let $j(v)\ge 0$ be the minimal integer such that $
	\frac{\mu_{M}\left(A_{j+2}\right)}{\mu_{M}\left(A_{j}\right)}\leq\left(
	\frac{\mu_{M}\left(A_{2k}\right)}{\mu_{M}\left(A_{0}\right)}\right)^{
	\frac{1}{k}}$
	\label{line:RG}\;
	\Return $(A_{j(v)},A_{j(v)+1},A_{j(v)+2})$\;
\end{algorithm}

	\begin{figure}[t]
	\centering{\includegraphics[width=\textwidth]{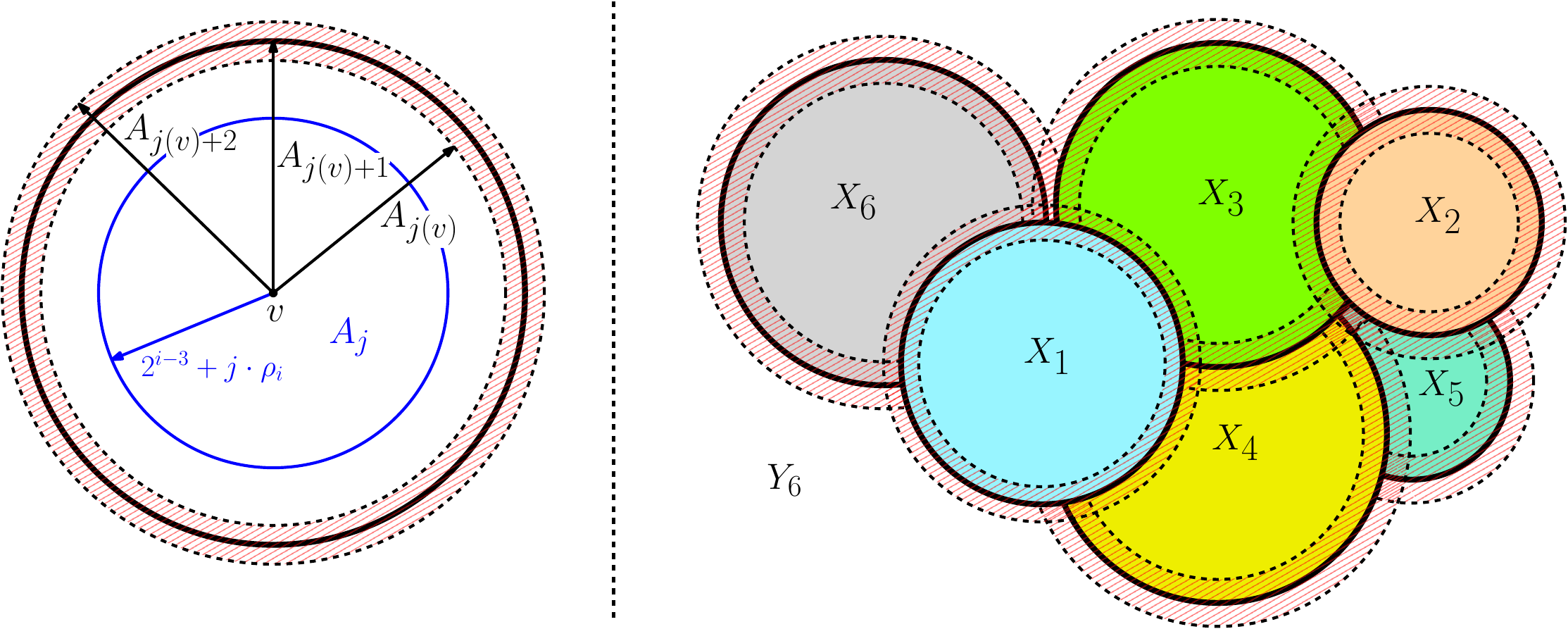}}
	\caption{\label{fig:RamseyAlgExample}\small
		Left: Illustration of the \texttt{Create-Cluster} procedure
		(\Cref{alg:CreateCluster}). The algorithm chooses a center $v$ maximizing
		the marked measure of the initial ball
		$A_0=B_H^{\red{(i\cdot h')}}(v,2^{i-3})$. Around this center it considers
		the nested balls $
		A_j=B_H^{\red{(i\cdot h'+j\cdot h)}}
		\left(v,2^{i-3}+j\cdot\rho_i\right)$, with increasing radii and
		hop budgets. 
		The blue ball depicts $A_j$. The procedure chooses the minimum index
		$j(v)$ satisfying the growth condition in \cref{line:RG}. It then returns
		the three consecutive balls
		$(A_{j(v)},A_{j(v)+1},A_{j(v)+2})$.\\
		Right: Illustration of the \texttt{Padded-partition} procedure
		(\Cref{alg:RamseyPartition}). In iteration $q$, the current working set is
		$Y_{q-1}$. The procedure calls \texttt{Create-Cluster} and obtains a triple
		$\underline{X}_{q}\subseteq X_q\subseteq\overline{X}_{q}$, corresponding to
		$A_{j(v)},A_{j(v)+1},A_{j(v)+2}$. The middle set
		$X_q=A_{j(v)+1}$ becomes the next cluster, while the boundary
		$\overline{X}_{q}\setminus\underline{X}_{q}
		=A_{j(v)+2}\setminus A_{j(v)}$
		is unmarked. The working set is updated to
		$Y_q=Y_{q-1}\setminus X_q$. In the illustration, the procedure creates
		$X_1,\dots,X_6$, leaving $Y_6$ as the remaining working set. All vertices
		covered by the red hatching become unmarked.
	}
\end{figure}

We first show that the index chosen in \Cref{alg:CreateCluster} is bounded
by $2(k-1)$.
\begin{claim}\label{claim:j-v}
	Consider a call to the \texttt{Create-Cluster} procedure, the index $j(v)$
	defined in \Cref{alg:CreateCluster} satisfies $j(v) \leq 2(k-1)$.
\end{claim}
\begin{proof}
	Using the terminology from the  \texttt{Create-Cluster} procedure, let
	$j\in\{0,1,2,\dots,k-1\}$ be the index minimizing
	$\frac{\mu_{M}\left(A_{2(j+1)}\right)}{\mu_{M}\left(A_{2j}\right)}$. Then
	it holds that
	\[
	\frac{\mu_{M}\left(A_{2k}\right)}{\mu_{M}\left(A_{0}\right)}=\frac{\mu_{M}
	\left(A_{2k}\right)}{\mu_{M}\left(A_{2(k-1)}\right)}\cdot\frac{\mu_{M}
	\left(A_{2(k-1)}\right)}{\mu_{M}\left(A_{2(k-2)}\right)}\cdots\frac{
	\mu_{M}\left(A_{2}\right)}{\mu_{M}\left(A_{0}\right)}\ge\left(\frac{
	\mu_{M}\left(A_{2(j+1)}\right)}{\mu_{M}\left(A_{2j}\right)}\right)^{k}~,
	\]
	the claim follows as
	$\frac{\mu_{M}\left(A_{2(j+1)}\right)}{\mu_{M}\left(A_{2j}\right)}\le
	\left(\frac{\mu_{M}\left(A_{2k}\right)}{\mu_{M}\left(A_{0}\right)}
	\right)^{\frac{1}{k}}$,
	and hence there is indeed such an index as required.
\end{proof}

The following is a straightforward corollary,
\begin{observation}\label{obs:diamBound}
	Every $i$-level cluster $X$ (i.e. cluster on which we called
	\texttt{Ramsey-type-embedding} with scale $2^i$) has diameter at most
	$\diam^{\red{(2(i+2)\cdot h')}}(G[X])\le 2^{i}$.
\end{observation}
\begin{proof}
	In the first call to the \texttt{Ramsey-type-embedding} algorithm we use
	the scale $2^\phi$ for $\phi=\left\lceil\log_2 \diam\rhop{h}(G)
	\right\rceil$. In particular, $\diam^{\red{(2(i+2)\cdot
	h')}}(G[X])\le\diam\rhop{h}(G[X])\le 2^\phi$.
	For any cluster $X$ other than the first one, it was created during a call
	to the \texttt{Create-Cluster} procedure with scale $2^{i+1}$. By
	\Cref{claim:j-v}, $X=A_{j(v)+1}=B_{H}^{\red{((i+1)\cdot h'+(j(v)+1)\cdot
	h)}}\left(v,2^{i-2}+(j(v)+1)\cdot\rho_{i+1}\right)$
	 for $j(v) \leq 2(k-1)$. In particular, for every $u\in X$,
	 $d_{G[X]}\rhop{(i+1)\cdot h'+2k\cdot h}(u,v)=d_{G[X]}\rhop{(i+2)\cdot
	 h'}(u,v)\le2^{i-2}+2k\cdot\rho_{i+1}=2^{i-1}$. In particular, by the
	 triangle inequality, $\diam^{\red{(2(i+2)\cdot h')}}(G[X])\le 2^{i}$.
\end{proof}

For a cluster $X$, denote by $\hat{X}$ the set of vertices in $X$ that remain
marked until the end of the algorithm. We next prove the distortion guarantee.
\begin{lemma}\label{lem:RamseyDistortion}
	For every pair of vertices $u,v\in V$ it holds that
	$d_{G}^{\red{(O(\phi)\cdot h')}}(u,v)\le d_{U}(u,v)$. Further, if $u\in
	\hat{V}$ then $d_{U}(u,v)\le16k\cdot d_{G}\rhop{h}(u,v)$.
\end{lemma}
\begin{proof}
	The proof is by induction on the scale $i$ used in the call to the
	\texttt{Ramsey-type-embedding}$(G[X],\mu,M,\red{h},k,2^i)$ algorithm.
	Specifically, we will prove that in the returned ultrametric $U$ it holds
	that $d_{G}^{\red{(2(i+2))\cdot h')}}(u,v)\le d_{U}(u,v)$ for arbitrary
	pair of vertices, while for $u\in\hat{X}$ it holds that
	$d_{U}(u,v)\le16k\cdot d_{G}\rhop{h}(u,v)$.
	The base case $i=0$ is easy, as the entire graph is a singleton.
	Consider the ultrametric $U$ returned by a call with scale $i$.
	The algorithm partitioned $X$ into clusters $\{X_{q}\}_{q=1}^{s}$ with
	marked sets $\{M_{q}\}_{q=1}^{s}$, and reconnected them under a root $r_U$
	with label $2^{i}$.

	We first prove the lower bound on the distance between all vertex pairs.
	By \Cref{obs:diamBound}, $\diam^{\red{(2(i+2)\cdot h')}}(G[X])\le 2^{i}$,
	and the label of the root of $U$ is $2^i$.
	If $u$ and $v$ belong to different clusters, then
	$d_{G[X]}^{\red{(2(i+2)\cdot h')}}(u,v)\le\diam^{\red{(2(i+2)\cdot
	h')}}(G[X])\le2^{i}= d_{U}(u,v)$. Otherwise, there is some index $q$ such
	that $u,v\in X_q$. By the induction hypothesis, it holds that
	$d_{G[X]}^{\red{(2(i+1)\cdot h')}}(u,v)\le d_{G[X_{q}]}^{\red{(2i\cdot
	h')}}(u,v)\le d_{U_{q}}(u,v)=d_{U}(u,v)$.

	Next we prove the upper bound, thus we assume that $u\in \hat{X}$.
	Denote $B=B_{G[X]}\rhop{h}(u,\rho_i)$. We argue that the entire ball $B$
	is contained in a single cluster $X_{\tilde{q}}$. We will use the notation
	from the \texttt{Padded-partition} procedure.
	Let $q$ be the minimal index such that there is a vertex $z\in B\cap
	X_{q}$. By the minimality of $q$, it follows that $B\subseteq Y_{q-1}$.
	Thus for every $z'\in B$, $d\rhop{h}_{G[Y_{q-1}]}(u,z')\le \rho_{i}$.
	Denote by $v_q$ the center vertex of the cluster $X_q$ created in the
	\texttt{Create-Cluster} procedure. Recall that
	$A_{j}=B_{G[Y_{q-1}]}^{\red{(i\cdot h'+j\cdot
	h)}}\left(v,2^{i-3}+j\cdot\rho_i\right)$, and there is some index $j(v_q)$
	such that $X_q=A_{j(v_q)+1}$.
	As $z\in X_q$, it holds that $d_{G[Y_{q-1}]}^{\red{((i\cdot
	h'+(j(v_{q})+1)\cdot h)}}(v_{q},z)\le2^{i-3}+(j(v_{q})+1)\cdot\rho_{i}$.
	Thus by the triangle inequality
	\begin{align*}
	d_{G[Y_{q-1}]}^{\red{((i\cdot h'+(j(v_{q})+2)\cdot h)}}(v_{q},u) & \le
	d_{G[Y_{q-1}]}^{\red{((i\cdot h'+(j(v_{q})+1)\cdot
	h)}}(v_{q},z)+d_{G[Y_{q-1}]}\rhop{h}(z,u)\\
	& \le2^{i-3}+(j(v_{q})+1)\cdot\rho_{i}+\rho_{i}=2^{i-3}+(j(v_{q})+2)
	\cdot\rho_{i}~,
	\end{align*}
	implying $u\in A_{j(v_q)+2}=\overline{X}_q$.
	As $u\in \hat{X}$ it follows that $u\in\hat{X}_q$, and thus $u\in
	\underline{X}_q=A_{j(v_q)}$ (as otherwise it would have been unmarked).
	Hence $d_{G[Y_{q-1}]}^{\red{((i\cdot h'+j(v_{q})\cdot
	h)}}(v_{q},u)\le2^{i-3}+j(v_{q})\cdot\rho_{i}$. Using the triangle
	inequality, we conclude that for every $z'\in B$,
	\begin{align*}
	d_{G[Y_{q-1}]}^{\red{((i\cdot h'+(j(v_{q})+1)\cdot h)}}(v_{q},z') & \le
	d_{G[Y_{q-1}]}^{\red{((i\cdot h'+j(v_{q})\cdot
	h)}}(v_{q},u)+d_{G[Y_{q-1}]}\rhop{h}(u,z')\\
	& \le2^{i-3}+j(v_{q})\cdot\rho_{i}+\rho_{i}=2^{i-3}+(j(v_{q})+1)
	\cdot\rho_{i}
	\end{align*}
	It follows that $B\subseteq A_{j(v_q)+1}=X_q$ as required.

	We are finally ready to upper bound $d_U(u,v)$.
	If $v\notin B$ then $d\rhop{h}(u,v)>\rho_i$, as the maximum distance in
	$U$ is $2^i$, it follows that
	$d_{U}(u,v)\le 2^{i}=16k\cdot\rho_{i}<16k\cdot d\rhop{h}(u,v)$.
	Else, $v\in B$. Then as $B=B_{G[X]}\rhop{h}(u,\rho_i)\subseteq X_q$, the
	entire shortest $\red{h}$-hop path from $u$ to $v$ is contained in $X_q$
	and hence $d_{G[X]}\rhop{h}(u,v)=d_{G[X_q]}\rhop{h}(u,v)$. Using the
	induction hypothesis on $i-1$, it follows that
	\[
	d_{U}(u,v)\le d_{U_{q}}(u,v)\le16k\cdot d_{G[X_{q}]}\rhop{h}(u,v)=16k\cdot
	d_{G[X]}\rhop{h}(u,v)~.
	\]

\end{proof}

Consider a cluster $C$ created during a call
\texttt{Create-Cluster}$G([H=Y_{q-1}],M,\mu,\red{h},k,2^i)$.
Specifically, there was an $i$-level cluster $X'$, from which we created
$X_1,\dots,X_s$. $C=X_q$ was created in the $q$'th step, where the set of yet
unclustered vertices was $Y_{q-1}=X'\setminus\cup_{j=1}^{q-1}X_j$.
Let $r(C)=r(X_q)$ be the center vertex of $C=X_q$ (chosen in
\cref{line:CreateClusterChooseCenter}), and also denote by $Y(C)=Y_{q-1}$ the set
of yet unclustered vertices. $C=X_q$ will be called an $i-1$-cluster (as by
\Cref{obs:diamBound} its diameter is bounded by $2^{i-1}$).
Denote by $M(C)=M\cap C$ the set of marked vertices in $C$ after
its creation.
Accordingly, $M(Y(C))=M(Y_{q-1})$ is the set of yet unclustered, but also
still marked vertices after the creation of $X_{q-1}$ (or all of
$M$ for $Y_0$).
Specifically, following the notation in \texttt{Create-Cluster},
$C=X_q=A_{j(v)+1}$, and $M(C)=M(X_q)=M(Y_{q-1})\cap\underline{X_q}=M(Y_{q-1})\cap
A_{j(v)}$. In particular, $M(Y_q)=M(Y_{q-1})\setminus
\overline{X_q}=M(Y_{q-1})\setminus A_{j(v)+2}$.

We now proceed to proving the inductive bound on the measure of the saved set.
Recall that for a set of marked
vertices $S$ and a subset $A\subseteq V$, we write
$\mu_S(A)=\mu(A\cap S)$ for the measure of the marked vertices of $A$.
The first cluster $V\in{\cal X}_\phi$, is a $\phi$-level cluster. We set
$Y(V)=V$, $M(V)=V$, and we choose its center vertex $r(V)$ to be the vertex
maximizing $\mu_{M}\left(B_{G}^{\red{((\phi+1)\cdot
h')}}(v,2^{\phi-2})\right)$.
Next we bound the number of vertices in $\hat{V}$.
\begin{lemma}\label{lem:frac-survival}
	For every $i$-level cluster $X$ with $M(X)\ne\emptyset$, it holds
	$\mu(\hat{X})\ge\frac{\mu(M(X))}{\mu_{M(Y(X))}\left(B_{Y(X)}^{\red{(i\cdot
	h')}}(r(X),2^{i-2})\right)^{\frac{1}{k}}}
	$.
\end{lemma}
\begin{proof}
	We prove the lemma by induction on $i$.
	Consider first the base case where $i=0$ and $M(X)\ne\emptyset$.
	As the minimum distance in $G$ is $1$, $X$ is the singleton vertex $r(X)$.
	It follows that
	$\mu(\hat{X})=\mu(r(X))=\mu(X)\ge\nicefrac{\mu_{M(X)}(X)}{\mu_{M(Y(X))}
	\left(B_{Y(X)}^{\red{(0)}}(r(X),\frac{1}{4})\right)^{\frac{1}{k}}}=
	\mu(r(X))^{1-\frac{1}{k}}$,
	which holds as $\mu(r(X))\ge 1$.

	Consider an $i$-level cluster $X$, and assume that the lemma holds for every
	$(i-1)$-level cluster.
	Let $v =r(X)$, and denote for ease of notation the set of marked vertices
	by $M=M(X)$.
	The algorithm performs the procedure
	\texttt{Padded-partition}$(G[X],\mu,M,\red{h},k,2^i)$ to obtain the
	$i-1$-level clusters $X_1,...,X_{s}$, each with the set
	$M_q\subseteq M\cap X_q$.
	For every $1\leq q \leq s$, denote by $v_q=r(X_q)$ the center of the cluster
	$X_q$. Using the notation from the \texttt{Padded-partition} procedure,
	$Y(X_q)=Y_{q-1}$ (where $Y_0=X$, $Y_q=Y_{q-1}\backslash X_q$).
	Recall that we choose an index $j_q=j(v_q)$ which in addition to $X_q$
	also defined sets $\underline{X_q}\subseteq X_q
	\subseteq\overline{X_q}$ such that
	\begin{equation}
		\mu_{M(X_{q})}(X_{q})=\mu_{M(Y_{q-1})}\left(\underline{X}_{q}\right)~,
		\label{eq:MeasureXq}
	\end{equation}
 and $M(Y_q)=M(Y_{q-1})\setminus \overline{X_q}$.
	By the choice of $j(v_q)$ in \Cref{alg:CreateCluster}, we have
	that
	\begin{equation}
	\frac{\mu_{M(Y_{q-1})}\left(\underline{X}_{q}\right)}{\mu_{M(Y_{q-1})}(
	\overline{X}_{q})}\ge\left(\frac{\mu_{M(Y_{q-1})}\left(A_{0}\right)}{
	\mu_{M(Y_{q-1})}\left(A_{2k}\right)}\right)=\left(\frac{\mu_{M(Y_{q-1})}
	\left(B_{Y_{q-1}}^{\red{(i\cdot h')}}(v_{q},2^{i-3})\right)}{
	\mu_{M(Y_{q-1})}\left(B_{Y_{q-1}}^{\red{((i+1)\cdot h')}}(v_{q},2^{i-2})
	\right)}\right)^{\frac{1}{k}}~.
	\label{eq:Int_Res_relation}
	\end{equation}
	Finally, by the choice of $r(X)$ to be the vertex maximizing
	$\mu_{M(Y(X))}\left(B_{Y(X)}^{\red{((i+1)\cdot h')}}(r(X),2^{i})\right)$,
	it holds that
	\begin{equation}
	\mu_{M(Y(X))}\left(B_{Y(X)}^{\red{((i+1)\cdot
	h')}}(r(X),2^{i-2})\right)\geq\mu_{M(Y(X))}\left(B_{Y(X)}^{\red{((i+1)
	\cdot h')}}(v_{q},2^{i-2})\right)=\mu_{M(Y_{q-1})}\left(B_{Y_{q-1}}^{
	\red{((i+1)\cdot h')}}(v_{q},2^{i-2})\right)~.
	\label{eq:measureInH(X)VersusH(X_q)}
	\end{equation}
	Using the induction hypothesis on $X_q$, we obtain
	\begin{align*}
		\mu(\hat{X}_{q}) &
		\ge\frac{\mu_{M(X_{q})}(X_{q})}{\mu_{M(Y_{q-1})}\left(B_{Y_{q-1}}^{
		\red{(i\cdot h')}}(v_{q},2^{i-3})\right)^{\frac{1}{k}}}\\
		& \overset{(\ref{eq:MeasureXq})}{=}\frac{\mu_{M(Y_{q-1})}\left(
		\underline{X}_{q}\right)}{\mu_{M(Y_{q-1})}\left(B_{Y_{q-1}}^{\red{(i
		\cdot h')}}(v_{q},2^{i-3})\right)^{\frac{1}{k}}}\\
		& \overset{(\ref{eq:Int_Res_relation})}{\ge}\frac{\mu_{M(Y_{q-1})}(
		\overline{X}_{q})}{\mu_{M(Y_{q-1})}\left(B_{Y_{q-1}}^{\red{(i\cdot
		h')}}(v_{q},2^{i-3})\right)^{\frac{1}{k}}}\cdot\left(\frac{
		\mu_{M(Y_{q-1})}\left(B_{Y_{q-1}}^{\red{(i\cdot h')}}(v_{q},2^{i-3})
		\right)}{\mu_{M(Y_{q-1})}\left(B_{Y_{q-1}}^{\red{((i+1)\cdot
		h')}}(v_{q},2^{i-2})\right)}\right)^{\frac{1}{k}}\\
		& \overset{(\ref{eq:measureInH(X)VersusH(X_q)})}{\ge}\frac{
		\mu_{M(Y_{q-1})}(\overline{X}_{q})}{\mu_{M(Y(X))}\left(B_{Y(X)}^{
		\red{((i+1)\cdot h')}}(r(X),2^{i-2})\right)^{\frac{1}{k}}}~.
	\end{align*}
	As $M(Y_q)=M(Y_{q-1})\backslash \overline{X}_q$, it follows that
	$\mu(M(Y_{q-1}))=\mu(M(Y_q))+\mu_{M(Y_{q-1})}(\overline{X}_q)$, and by
	induction $\mu(M(X))=\sum_{q=1}^s\mu_{M(Y_{q-1})}(\overline{X}_q)$.
	We conclude that
	\[
	\mu(\hat{X})=\sum_{q=1}^{s}\mu(\hat{X}_{q})\ge\sum_{q=1}^{s}\frac{
	\mu_{M(Y_{q-1})}(\overline{X}_{q})}{\mu_{M(Y(X))}\left(B_{Y(X)}^{\red{((i+
	1)\cdot h')}}(r(X),2^{i-2})\right)^{\frac{1}{k}}}=\frac{\mu_{M(X)}(X)}{
	\mu_{M(Y(X))}\left(B_{Y(X)}^{\red{((i+1)\cdot h')}}(r(X),2^{i-2})\right)^{
	\frac{1}{k}}}~.
	\]
\end{proof}

Using \Cref{lem:frac-survival} on $V$ with $i=\logdiam$ and $M=\mathcal{M}$ we
have \\$\mu(\hat{V})\ge\nicefrac{\mu(V)}{\mu\left(B_{G}^{\red{(\phi\cdot
h')}}(r(X),2^{\phi-2})\right)^{\frac{1}{k}}}\ge\mu(V)^{1-\frac{1}{k}}$.
\Cref{lem:UltraMeasure} now follows.

\subsubsection{The case where
$\diam\rhop{h}(G)=\infty$}\label{subsec:inftyDistance}
In this subsection we will remove the assumption that
$\diam\rhop{h}(G)<\infty$.
Note that even in the simple unweighted path graph $P_n$, it holds that
$\diam\rhop{h}(G)=\infty$.
Let $D'=\max\{d_G\rhop{h}(u,v)\mid u,v\in V \mbox{ and
}d_G\rhop{h}(u,v)<\infty\}$ be the maximum distance between a pair of
vertices $u,v$ such that $\hop_G(u,v)\le \red{h}$.
We create an auxiliary graph $G'$ by taking $G$ and adding an edge of weight
$\omega=17k\cdot D'$ for every pair $u,v$ for which $\hop_G(u,v)>\red{h}$.
Note that $G'$ has a polynomial aspect ratio (as we can assume that $k$ is at
most polynomial in $n$), and the maximum $\red{h}$-hop distance is $\omega$.
Next, using the finite case, we construct an ultrametric $U'$ for $G'$, and
receive a set $M$ such that for every $u,v$, $d_{G'}^{\red{(O(k \log n)\cdot
h)}}(u,v)\le d_{U'}(u,v)$, while for $u\in M$ and $v\in V$, $d_{U'}(u,v)\le
16k\cdot d_{G'}\rhop{h}(u,v)$.
Finally, we create a new ultrametric $U$ from $U'$ by replacing each label $l$
such that $l\ge \omega$ by $\infty$. We argue that the ultrametric $U$
satisfies the conditions of \Cref{lem:UltraMeasure} w.r.t. $G$ and
the set $M$.

First note that as we used the same measure $\mu$, it holds that
$\mu(M)\ge\mu(V)^{1-\frac1k}$.
Next we argue that $U$ does not shrink any distance. For every pair $u,v\in V$,
if $d_U(u,v)=\infty$ then clearly $d_G^{\red{(O(k \log n)\cdot
h)}}(u,v)\le\infty= d_U(u,v)$.
Else $d_U(u,v)<\infty$, then
$d_{G'}^{\red{(O(k \log n)\cdot h)}}(u,v)\le d_{U'}(u,v)<\omega$. It follows
that the shortest $\red{O(k \log n)\cdot h}$-path in $G'$ from $u$ to $v$
did not use any of the newly added edges. In particular $d_{G}^{\red{(O(k
\log n)\cdot h)}}(u,v)=d_{G'}^{\red{(O(k \log n)\cdot h)}}(u,v)\le
d_{U'}(u,v)=d_{U}(u,v)$ as required.

Finally, consider $u\in M$ and $v\in V$. If $d_U(u,v)=\infty$, then $\omega\le
d_{U'}(u,v)\le 16k\cdot d_{G'}\rhop{h}(u,v)$, implying $d_{G'}\rhop{h}(u,v)\ge
\frac{\omega}{16k}>D'$. In particular $\hop_G(u,v)> \red{h}$, and hence
$d_{G}\rhop{h}(u,v)=\infty$.
Otherwise, if $d_U(u,v)<\infty$,
$$d_{U}(u,v)=d_{U'}(u,v)\le 16k\cdot d_{G'}\rhop{h}(u,v)= 16k\cdot
d_{G}\rhop{h}(u,v)~.$$

\subsection{Alternative Ramsey construction}\label{subsec:AltRamsey}
In this section, we modify the \texttt{Ramsey-type-embedding} algorithm by
making changes to the \texttt{Create-Cluster} procedure. As a result, we obtain
a theorem with slightly different trade-off between hop-stretch and distortion
from that of \Cref{thm:UltrametricRamsey}. Specifically, by adding an additional
factor of  $O(\log\log n)$ in the distortion, we can allow the hop-stretch to
be as small as $O(\log\log n)$, and  completely avoid any dependence on the
aspect ratio.
We restate the theorem for convenience:

\begin{customthm}{\ref{thm:UltrametricRamseyAlt}}[Hop-Constrained Ramsey
Embedding with small hop-stretch and arbitrary aspect-ratio]
	\sloppy
	Consider an $n$-vertex graph $G=(V,E,w)$, and parameters
	$k,\red{h}\in [n]$.
	Then there is a distribution $\mathcal{D}$ over dominating ultrametrics
	with $V$ as leaves, such that for every $U\in\supp(\mathcal{D})$ there is a
	set $M_U\subseteq V$ such that:
	\begin{enumerate}
		\item Every $U\in\supp(\mathcal{D})$, has Ramsey hop-distortion
		$(O(k\cdot\log\log n),M_U,\red{O(k\cdot\log\log n)},\red{h})$.
		\item For every $v\in V$, $\Pr_{U\sim\mathcal{D}}[v\in M_U]\ge
		\Omega(n^{-\frac1k})$.
	\end{enumerate}
	In addition, for every $\eps\in (0,1)$, there is a distribution
	$\mathcal{D}$ as above such that every $U\in\supp(\mathcal{D})$ has Ramsey
	hop-distortion $(O(\frac{\log n\cdot\log\log n}{\eps}),M,\red{O(\frac{\log
	n\cdot\log\log n}{\eps})},\red{h})$, and $\forall v\in V$, $\Pr[v\in
	M]\ge 1-\eps$.
\end{customthm}

Analogously to the role of \Cref{lem:RamseyDistortion} in the proof of
\Cref{thm:UltrametricRamsey}, the following lemma is the core of our
construction.
\begin{lemma}\label{lem:UltraMeasureAlt}
	\sloppy Consider an $n$-vertex weighted graph $G=(V,E,w)$,
	$(\ge1)$-measure $\mu:V\rightarrow\mathbb{R}_{\ge1}$,
	integer parameters $k,\red{h}\in [n]$, and subset
	$\mathcal{M}\subseteq V$.
	Then there is a Ramsey type embedding into an ultrametric with Ramsey
	hop distortion
	$(O(k\cdot\log\log\mu(\mathcal{M})),M,\red{O(k\cdot\log\log\mu(
	\mathcal{M}))},\red{h})$ where $M\subseteq \mathcal{M}$ and $\mu(M)\ge
	\mu(\mathcal{M})^{1-\frac1k}$.
\end{lemma}

The proof of \Cref{thm:UltrametricRamseyAlt} using \Cref{lem:UltraMeasureAlt}
follows exactly the same lines as the proof of \Cref{thm:UltrametricRamsey}
using \Cref{lem:UltraMeasure}, and thus we will skip it.

\subsection{Proof of \Cref{lem:UltraMeasureAlt}: alternative distributional
h.c. Ramsey type embedding}
For the embedding of \Cref{lem:UltraMeasureAlt} we will use the exact same
\texttt{Ramsey-type-embedding} (\Cref{alg:hier-Ramsey}), with the only
difference that instead of using the \texttt{Create-Cluster} procedure
(\Cref{alg:CreateCluster}), we will use the  \texttt{Create-Cluster-alt}
procedure (\Cref{alg:CreateClusterAlternative}).
While in \Cref{alg:CreateCluster} the allowed number of hops in the
construction gradually decreases as the algorithm progresses, here we
start from a clean sheet in each scale. This prevents us from relating the
different scales in the analysis. Instead, we give stronger
``per-level'' guarantees.

Similarly to the proof of \Cref{lem:UltraMeasure}, we will assume that
$\diam\rhop{h}(G)<\infty$. The exact same argument from
\Cref{subsec:inftyDistance} can be used to remove this assumption.
Similarly to the proof of \Cref{lem:UltraMeasure} we also denote
$\phi=\left\lceil\log_2 \diam\rhop{h}(G) \right\rceil$, and assume that the
minimum edge weight is $1$. However, here there is no assumption of
polynomial aspect ratio, and thus we don't have any bound on $\phi$.
We begin by making the call
\texttt{Ramsey-type-embedding}$(G[V],\mu,M=\mathcal{M},\red{h},k,2^\phi)$.

\begin{algorithm}[t]
	\caption{$(\underline{A},A,\overline{A})$=\texttt{Create-Cluster-alt}$(H=
	G[X],M,\mu,\red{h},k,2^i)$}\label{alg:CreateClusterAlternative}
	\DontPrintSemicolon
	\SetKwInOut{Input}{input}\SetKwInOut{Output}{output}
	\Input{Induced graph $G[X]$, parameters $\red{h},k\in\N$, $\ge1$-measure
	$\mu$, set of padded vertices $M$, scale $i$.}
	\Output{Three clusters $(\underline{A},A,\overline{A})$, where
	$\underline{A}\subseteq A\subseteq \overline{A}$.}
	\BlankLine
	set $L=\lceil 1+\log\log\mu(M)\rceil$ and $\Delta=2^i$\;
	let $v \in M$ with minimal $\mu_{M}\left(B_{H}^{\red{(2k\cdot L\cdot
	h)}}(v,\frac18\Delta)\right)$\;
	\If{$\mu_{M}\left(B_{H}^{\red{(2k\cdot L\cdot
	h)}}(v,\frac18\Delta)\right)>\frac12\cdot\mu_M(X)$}{
		let $S=\cup_{x\in M}B_{H}^{\red{(2k\cdot L\cdot
		h)}}(x,\frac18\Delta)$\tcp*{Note that for every $u,v\in M$,
		$d_H^{\red{(4k\cdot L\cdot h)}}(u,v)\le \frac\Delta4$}
		\Return $(S,S,S)$\label{line:AltCreateClusterReturnX}\;
	}
	denote $A_{a,j}=B_{H}^{\red{((2k\cdot a+j)\cdot
	h)}}\left(v,(a+\frac{j}{2k})\cdot\frac{\Delta}{8L}\right)$\;
	let $a\in[0,L-1]$ such that
	$\mu_{M}\left(A_{a,0}\right)\ge\frac{\mu_{M}\left(A_{a+1,0}\right)^{2}}{
	\mu_{M}(H)}$\label{line:firstStepRamsey}\;
	let $j\in [0,2(k-1)]$ such that
	$\mu_{M}\left(A_{a,j+2}\right)\le\mu_{M}\left(A_{a,j}\right)\cdot\left(
	\frac{\mu_{M}\left(A_{a+1,0}\right)}{\mu_{M}\left(A_{a,0}\right)}\right)^{
	\frac{1}{k}}$\label{line:SecondStepRamsey}\;
	\Return $(A_{a,j},A_{a,j+1},A_{a,j+2})$\label{line:RamseyCreateFinal}\;
\end{algorithm}

In the \texttt{Create-Cluster-alt} procedure, the main objects of interest are
the sets $A_{a,j}=B_{H}^{\red{((2k\cdot a+j)\cdot
h)}}\left(v,(a+\frac{j}{2k})\cdot\frac{\Delta}{8L}\right)$, where there are
two indices $a\in[0,L-1]$, $j\in [0,2(k-1)]$, governing the radius and the
allowed number of hops.
The sets $A_{a,j}$ range between
$A_{0,0}=B_{H}^{\red{(0)}}\left(v,0\right)=\{v\}$ and
$A_{L,0}=B_{H}^{\red{(2k\cdot L\cdot h)}}\left(v,\frac{\Delta}{8}\right)$.
Note that all these sets are nested.
The proof of \Cref{lem:UltraMeasureAlt} follows similar lines to
\Cref{lem:UltraMeasure}.
First we argue that we can choose indices $a,j$ as specified in
\cref{line:SecondStepRamsey} and \cref{line:firstStepRamsey} of
\Cref{alg:CreateClusterAlternative}.

\begin{claim}\label{clm:RamseyAltAChoise}
	In \cref{line:firstStepRamsey} of \Cref{alg:CreateClusterAlternative},
	there is an index $a\in [0,L-1]$ such that
	$\mu_{M}\left(A_{a,0}\right)\ge\frac{\mu_{M}\left(A_{a+1,0}\right)^{2}}{
	\mu_{M}(H)}$.
\end{claim}
\begin{proof}
	Seeking contradiction, assume that for every index $a$
	it holds that
	$\mu_{M}\left(A_{a,0}\right)<\frac{\mu_{M}\left(A_{a+1,0}\right)^{2}}{
	\mu_{M}(H)}$. Applying this for every $a\in[0,L-1]$ we have that
	\begin{align*}
	\mu_{M}\left(A_{L,0}\right) &
	>\mu_{M}\left(A_{L-1,0}\right)^{\frac{1}{2}}\cdot\mu_{M}(H)^{\frac{1}{2}}>
	\mu_{M}\left(A_{L-2,0}\right)^{\frac{1}{4}}\cdot\mu_{M}(H)^{\frac{3}{4}}\\
	& >\dots>\mu_{M}\left(A_{L-l,0}\right)^{2^{-l}}\cdot\mu_{M}(H)^{1-2^{-l}}>
	\mu_{M}\left(A_{0,0}\right)^{2^{-L}}\cdot\mu_{M}(H)^{1-2^{-L}}~.
	\end{align*}
	But $\mu_{M}(H)^{1-2^{-L}}=\frac{\mu_{M}(H)}
	{\mu_{M}(H)^{\frac{1}{2^{L}}}}>
	\frac{\mu_{M}(H)}{\mu_{M}(H)^{\frac{1}{\log(\mu_{M}(H))}}}=\frac{1}{2}\cdot
	\mu_{M}(H)$
	 (as $L=\lceil 1+\log\log\mu(M)\rceil>\log\log \mu_{M}(H)$ ), and as $v\in
	 M$, $\mu_{M}(A_{0,0})=\mu_{M}\left(B_{H}^{\red{(0)}}\left(v,0\right)
	 \right)=\mu_{M}(v)\ge1$.
	 Thus $\mu_{M}\left(A_{0,0}\right)^{2^{-L}}\ge1$. Hence we obtain
	$\mu_{M}\left(B_{H}^{\red{(2k\cdot L\cdot
	h)}}\left(v,\frac{\Delta}{8}\right)\right)=\mu_{M}\left(A_{L,0}\right)>
	\frac{1}{2}\cdot\mu_{M}(H)$,
	a contradiction (as the algorithm would have halted at
	\cref{line:AltCreateClusterReturnX}).
\end{proof}

\begin{claim}\label{clm:RamseyAltJChoise}
	In \cref{line:SecondStepRamsey} of \Cref{alg:CreateClusterAlternative},
	there is an index $j\in [0,2(k-1)]$ such that
	$\mu_{M}\left(A_{a,j+2}\right)\le\mu_{M}\left(A_{a,j}\right)\cdot\left(
	\frac{\mu_{M}\left(A_{a+1,0}\right)}{\mu_{M}\left(A_{a,0}\right)}\right)^{
	\frac{1}{k}}$.
\end{claim}
\begin{proof}
	Let $j\in\{0,2,\dots,2(k-1)\}$ be the index minimizing the ratio
	$\frac{\mu_{M}\left(A_{a,j+2}\right)}{\mu_{M}\left(A_{a,j}\right)}$.
	It holds that
	\[
	\frac{\mu_{M}\left(A_{a+1,0}\right)}{\mu_{M}\left(A_{a,0}\right)}=\frac{
	\mu_{M}\left(A_{a,2}\right)}{\mu_{M}\left(A_{a,0}\right)}\cdot\frac{
	\mu_{M}\left(A_{a,4}\right)}{\mu_{M}\left(A_{a,2}\right)}\cdots\frac{
	\mu_{M}\left(A_{a+1,0}\right)}{\mu_{M}\left(A_{a,2(k-1)}\right)}\ge\left(
	\frac{\mu_{M}\left(A_{a,j+2}\right)}{\mu_{M}\left(A_{a,j}\right)}
	\right)^{k}~,
	\]
	the claim follows.
\end{proof}

Next we observe that the diameter of the clusters in our laminar partition is
gradually decreasing. However here we argue this w.r.t. a fixed
number of hops.
\begin{observation}\label{obs::diamBoundAlt}
	Every $i$-level cluster $X$ (i.e. cluster on which we called
	\texttt{Ramsey-type-embedding} with scale $2^i$) has diameter at most
	$\diam^{\red{(8k\cdot L\cdot h)}}(G[X])\le 2^{i}$.
\end{observation}
\begin{proof}
	In the first call to the \texttt{Ramsey-type-embedding} algorithm we use
	the scale $2^\phi$ for $\phi=\left\lceil\log_2 \diam\rhop{h}(G)
	\right\rceil$. In particular, $\diam^{\red{(8k\cdot L\cdot
	h)}}(G[X])\le\diam\rhop{h}(G[X])\le 2^\phi$.
	For an $i$-level cluster $X$ other than the first one, it was created
	during a call to the \texttt{Create-Cluster-alt} procedure with scale
	$\Delta=2^{i+1}$ over a cluster $Y$.
	There are two different options for \Cref{alg:CreateClusterAlternative} to
	return a cluster.
	Suppose first that it returned a cluster $S=\cup_{x\in
	M}B_{H}^{\red{(2k\cdot L\cdot h)}}(x,\frac18\Delta)$ in
	\cref{line:AltCreateClusterReturnX}, and consider two vertices $u,v\in S$.
	By definition, there are $u',v'\in M$ such that
	$u\in B_{H}^{\red{(2k\cdot L\cdot h)}}(u',\frac18\Delta)$, and $v\in
	B_{H}^{\red{(2k\cdot L\cdot h)}}(v',\frac18\Delta)$. Furthermore, it holds
	that $\mu_{M}\left(B_{Y}^{\red{(2k\cdot L\cdot
	h)}}(v',\frac{\Delta}{8})\right)>\frac12\cdot\mu_M(Y)$ and
	$\mu_{M}\left(B_{Y}^{\red{(2k\cdot L\cdot
	h)}}(u',\frac{\Delta}{8})\right)>\frac12\cdot\mu_M(Y)$. Hence there is a
	point $z\in B_{Y}^{\red{(2k\cdot L\cdot h)}}(v',\frac{\Delta}{8})\cap
	B_{Y}^{\red{(2k\cdot L\cdot h)}}(u',\frac{\Delta}{8})$
	in the intersection of these two balls. It follows that
	\[
	d_{H}^{\red{(8k\cdot L\cdot h)}}(u,v)\le d_{H}^{\red{(2k\cdot L\cdot
	h)}}(u,u')+d_{H}^{\red{(2k\cdot L\cdot h)}}(u',z)+d_{H}^{\red{(2k\cdot
	L\cdot h)}}(z,v')+d_{H}^{\red{(2k\cdot L\cdot
	h)}}(v',v)\le4\cdot\frac{\Delta}{8}=\frac{\Delta}{2}=2^{i}~.
	\]
	Otherwise, $X=A_{a,j+1}$ for some $a\in [0,L-1]$ and $j\in[0,2(k-1)]$. As
	$(2ka+(j+1))\cdot h\le 2k(a+1)\le 2kL\cdot h$, and
	$(a+\frac{j+1}{2k})\cdot\frac{\Delta}{8L}<(a+1)\cdot\frac{\Delta}{8L}\le
	\frac{\Delta}{8}=2^{i-3}$, it follows that $X\subseteq
	B_{H}^{\red{(2k\cdot L\cdot h)}}(v,2^{i-1})$, the observation follows.
\end{proof}

Recall that for a cluster $X$, we denote by $\hat{X}$ the set of vertices in
$X$ that have been padded throughout the algorithm (that is, never belonged to
$\overline{X}_{j}\setminus\underline{X}_{j}$). In particular, $\hat{V}$
denotes the set $M$ returned at the end of the algorithm. The proof of the
following lemma follows the exact same lines as \Cref{lem:RamseyDistortion},
and thus we will skip it.
\begin{lemma}\label{lem:RamseyDistortionAlt}
	For every pair of vertices $u,v\in V$ it holds that
	$d_{G}^{\red{(O(k\cdot\log\log\mu(\mathcal{M}))\cdot h)}}(u,v)\le
	d_{U}(u,v)$. Further, if $u\in \hat{V}$ then $d_{U}(u,v)\le
	O(k\cdot\log\log \mu(\mathcal{M}))\cdot d_{G}\rhop{h}(u,v)$.
\end{lemma}

Finally, we will lower bound the measure of the vertices $\hat{V}$ that remain
marked until the end of the algorithm.
\begin{lemma}\label{lem:frac-survival-alt}
	For every cluster $X$ on which we execute the
	\texttt{Ramsey-type-embedding} using the \texttt{Create-Cluster-alt}
	procedure, it holds that $\mu_M(\hat{X})\ge \mu_M(X)^{1-\frac1k}$.
\end{lemma}
\begin{proof}
	We prove the lemma by induction on $|X|$ and $i$.
	The base cases, where either $X$ is a singleton or $i=0$, are both
	trivial. For the inductive step, assume we call
	\texttt{Ramsey-type-embedding} on $(H=G[X],\mu,k,\red{h},2^i)$ and the
	current set of marked vertices is $M$.
	The algorithm performs the procedure
	\texttt{Padded-partition}$(G[X],\mu,M,\red{h},k,2^i)$ to obtain the
	$i-1$-clusters $X_1,...,X_{s}$, each with the set $M_q\subseteq M\cap X_q$
	of marked vertices.
	If the algorithm returns a trivial partition (i.e. $s=1$) then we are done
	by induction (as $M_1=M$).
	Similarly, if the cluster $X_1$ was returned by the
	\texttt{Create-Cluster-alt} procedure in
	\cref{line:AltCreateClusterReturnX}, then $M_1=M$, and again we can simply
	apply induction.
	Hence we can assume that $s>1$, and the cluster $X_1$ was returned by the
	\texttt{Create-Cluster-alt} procedure in \cref{line:RamseyCreateFinal}.
	Denote by $\hat{X}_q$ the set of vertices from $X_q$ that remained marked
	at the end of the algorithm. Then $\hat{X}=\cup_{q=1}^{s}\hat{X}_q$.
	Recall that after creating $X_1$, the set of remaining vertices is denoted
	$Y_1=X\backslash X_1$, with the set $M_1^Y=M\backslash \overline{X}_1$ of
	marked vertices.
	Note that if we were to call \texttt{Ramsey-type-embedding} on input
	$(G[Y_1],\mu,M_1^Y,\red{h},k,2^i)$, it would first partition $Y_1$ into
	$X_2,\dots,X_s$, and then continue on these clusters in the exact same
	way as the original execution on $X$. In particular we would receive an
	embedding into an ultrametric $U_Y$.
	Therefore, we can analyze the rest of the process as if the algorithm did
	a recursive call on $Y_1$ rather than on each cluster $X_q$.
	Denote $\hat{Y}_1=\cup_{q=2}^{s}\hat{X}_q$.
	Since $|X_1|,|Y_1|<|X|$, the induction hypothesis implies that
	$\mu(\hat{X}_{1})\ge\mu(X_{1})^{1-\frac{1}{k}}$ and
	$\mu(\hat{Y}_1)\ge\mu(M_{1}^Y)^{1-\frac{1}{k}}$.

	As the \texttt{Padded-partition} algorithm did not return a trivial
	partition of $X$, it defined  $A_{a,j}=B_{H}^{\red{((2k\cdot a+j)\cdot
	h)}}\left(v,(a+\frac{j}{2k})\cdot\frac{\Delta}{8L}\right)$,
	and returned
	$(\underline{X}_1,X_1,\overline{X}_1)=(A_{a,j},A_{a,j+1},A_{a,j+2})$. The
	indices $a,j$ were chosen such that
	$\mu_{M}\left(A_{a,0}\right)\ge\frac{\mu_{M}\left(A_{a+1,0}\right)^{2}}{
	\mu_{M}(X)}$ and
	$\mu_{M}\left(A_{a,j+2}\right)\le\mu_{M}\left(A_{a,j}\right)\cdot\left(
	\frac{\mu_{M}\left(A_{a+1,0}\right)}{\mu_{M}\left(A_{a,0}\right)}\right)^{
	\frac{1}{k}}$,
	which is equivalent to $\mu_{M}\left(\underline{X}_{1}\right)\ge\mu_{M}
	\left(\overline{X}_{1}\right)\cdot\left(\frac{\mu_{M}\left(A_{a,0}\right)}
	{\mu_{M}\left(A_{a+1,0}\right)}\right)^{\frac{1}{k}}$.
	It holds that
	\begin{align*}
		\mu(\hat{X}_{1}) &
		\ge\frac{\mu(M_{1})}{\mu(M_{1})^{\frac{1}{k}}}=\frac{\mu_{M}(
		\underline{X}_{1})}{\mu_{M}(\underline{X}_{1})^{\frac{1}{k}}}\\
		& \ge\mu_{M}(\overline{X}_{1})\cdot\left(\frac{\mu_{M}\left(A_{a,0}
		\right)}{\mu(\underline{X}_{1})\cdot\mu\left(A_{a+1,0}\right)}
		\right)^{\frac{1}{k}}\\
		& \ge\mu_{M}(\overline{X}_{1})\cdot\left(\frac{\mu_{M}\left(A_{a+1,0}
		\right)}{\mu(\underline{X}_{1})\cdot\mu_{M}\left(X\right)}\right)^{
		\frac{1}{k}}\ge\frac{\mu_{M}(\overline{X}_{1})}{\mu_{M}\left(X
		\right)^{\frac{1}{k}}}~,
	\end{align*}
	where the third inequality follows as
	$\frac{\mu_{M}\left(A_{a,0}\right)}{\mu_{M}\left(A_{a+1,0}\right)}\ge
	\frac{\mu_{M}\left(A_{a+1,0}\right)}{\mu_{M}(X)}$, and the last inequality
	holds as $A_{a+1,0}\supseteq A_{a,j}=\underline{X}_{1}$. As $M_1^Y=M
	\backslash \overline{X}_{1}$, we conclude
	\begin{align*}
	\mu(\hat{X}) & =\mu(\hat{X}_{1})+\mu(\hat{Y}_{1})\\
	& \ge\frac{\mu_{M}(\overline{X}_{1})}{\mu_{M}\left(X\right)^{
	\frac{1}{k}}}+\mu_{M}(M_{1}^{Y})^{1-\frac{1}{k}}\\
	& \ge\frac{\mu_{M}(\overline{X}_{1})+\mu_{M}(M_{1}^{Y})}{\mu_{M}\left(X
	\right)^{\frac{1}{k}}}=\frac{\mu_{M}(\overline{X}_{1})+\mu_{M}(X)-\mu_{M}(
	\overline{X}_{1})}{\mu_{M}\left(X\right)^{\frac{1}{k}}}=\mu_{M}\left(X
	\right)^{1-\frac{1}{k}}~.
	\end{align*}
\end{proof}

\section{Hop-constrained clan embedding}
In this section we prove \Cref{thm:ClanHopUltrametric}.
We restate the theorem for convenience:
\ClanHopUltrametric*

The algorithm and its proof will follow lines similar to those of
\Cref{thm:UltrametricRamsey}. The main difference is that in the construction
of \Cref{thm:UltrametricRamsey} we had a set of marked vertices $M$ and each
time when creating a partition we removed all the boundary vertices from $M$,
while here we will create a cover where the boundary vertices will be
duplicated to multiple clusters.

Intuitively, in the construction of the Ramsey-type embedding we first created
a hierarchical partition. In contrast, here we will create a hierarchical cover.
A cover of $X$ is a collection $\mathcal{X}$ of clusters such that each vertex
belongs to at least one cluster (but possibly to many more).
A cover $\mathcal{X}$ refines a cover $\mathcal{X}'$ if for every cluster
$C\in \mathcal{X}$, there is a cluster $C'\in \mathcal{X}'$ such that
$C\subseteq C'$.
A hierarchical cover is a collection of covers $\{{\cal X}_0,{\cal
X}_1,\dots,{\cal X}_\phi\}$ such that ${\cal X}_\logdiam=\{V\}$ is the trivial
cover, ${\cal X}_0$ consists only of singletons, and for every $i$, ${\cal
X}_i$ refines ${\cal X}_{i+1}$.
However, to accurately use a hierarchical cover to create a clan embedding, one
needs to control for the different copies of a vertex, which makes this
notation tedious. Therefore, while we will actually create a hierarchical
cover, we will not use this notation explicitly.
As in the Ramsey-case, the main technical part will be the following
distributional lemma, the proof of which is deferred to
\Cref{subsec:ProofClanDistributionalLemma}.
Recall that a $(\ge1)$-measure is simply a function
$\mu:V\rightarrow\mathbb{R}_{\ge1}$. Given a one-to-many embedding
$f:X\rightarrow2^Y$, set  $\mathbb{E}_{v\sim\mu}[|f(v)|]=\sum_{v\in
X}\mu(v)\cdot |f(v)|$.
\begin{lemma}\label{lem:clanHopUltraMeasure}
	Consider an $n$-vertex graph $G=(V,E,w)$ with polynomial aspect ratio,
	$(\ge1)$-measure $\mu:V\rightarrow\mathbb{R}_{\ge1}$, and parameters
	$k,\red{h}\in [n]$. Then there is a clan embedding $(f,\chi)$ into an
	ultrametric with
	hop-distortion $\left(16(k+1),\red{O(k\cdot\log n)},\red{h}\right)$,
	hop-path-distortion $\left(\redno{O(k\cdot\log
	\mu(V))},\red{h}\right)$, and
	such that $\mathbb{E}_{v\sim\mu}[|f(v)|]\le\mu(V)^{1+\frac1k}$.
\end{lemma}

\begin{proof}[Proof of \Cref{thm:ClanHopUltrametric}]
	In a similar fashion to \Cref{clm:probabilisticRamsey}, we begin by
	translating \Cref{lem:clanHopUltraMeasure} to the language of
	probability measures.
	\begin{claim}\label{clm:probabilisticClan}
		Consider an $n$-vertex graph $G=(V,E,w)$ with polynomial aspect ratio,
		probability measure $\mu:V\rightarrow[0,1]$,
		and parameter $h\in [n]$. We can construct the following two clan
		embeddings $(f,\chi)$ into ultrametrics:
		\begin{enumerate}
			\item For every parameter  $k\in\N$, hop-distortion $\left(16\cdot
			(k+1),\red{O(k\cdot\log n)},\red{h}\right)$, hop-path-distortion
			$\left(\redno{O(k\cdot\log n)},\red{h}\right)$, and such that
			$\mathbb{E}_{x\sim\mu}[|f(x)|]\le
			O(n^{\frac{1}{k}})$.
			\item For every parameter  $\epsilon\in(0,1]$, hop-distortion
			$\left(O(\frac{\log n}{\eps}),\red{O(\frac{\log^2
			n}{\eps})},\red{h}\right)$, hop-path-distortion
			$\left(\redno{O(\frac{\log^2 n}{\eps})},\red{h}\right)$, and such
			that $\mathbb{E}_{x\sim\mu}[|f(x)|]\le 1+\eps$.
		\end{enumerate}
	\end{claim}
	\begin{proof}
		We define the following probability measure $\widetilde{\mu}$:
		$\forall x\in V$, $\widetilde{\mu}(x)=\frac{1}{2n}+\frac{1}{2}\mu(x)$.
		Set the following $(\ge1)$-measure $\widetilde{\mu}_{\ge1}(x)=2n\cdot
		\tilde{\mu}(x)$. Note that $\widetilde{\mu}_{\ge1}(V)=2n$.
		We execute \Cref{lem:clanHopUltraMeasure} w.r.t. the $(\ge1)$-measure
		$\widetilde{\mu}_{\ge1}$, and parameter $\frac1\delta\in\N$ to be
		determined later. It holds that
		\[
		\widetilde{\mu}_{\ge1}(V)\cdot\mathbb{E}_{x\sim\widetilde{
		\mu}}[|f(x)|]=\mathbb{E}_{x\sim\widetilde{\mu}_{\ge1}}[|f(x)|]\le
		\widetilde{\mu}_{\ge1}(V)^{1+\delta}=\widetilde{\mu}_{\ge1}(V)
		\cdot(2n)^{\delta}~,
		\]
		implying
		\[
		(2n)^{\delta}\ge\mathbb{E}_{x\sim\widetilde{\mu}}[|f(x)|]=\frac{1}{2}
		\cdot\mathbb{E}_{x\sim\mu}[|f(x)|]+\frac{\sum_{x\in V}|f(x)|}{2n}\ge
		\frac{1}{2}\cdot\mathbb{E}_{x\sim\mu}[|f(x)|]+\frac{1}{2}~.
		\]
		\begin{enumerate}
			\item  Set
			$\delta=\frac1k$, then we have hop-distortion $\left(16\cdot
			(k+1),\red{O(k\cdot\log n)},\red{h}\right)$, hop-path-distortion
			$\left(\redno{O(k\cdot\log n)},\red{h}\right)$, and
			$\mathbb{E}_{x\sim\mu}[|f(x)|]\le2\cdot(2n)^{\delta}=O(n^{
			\frac{1}{k}})$.
			\item Choose $\delta\in(0,1]$ such that
			$\frac{1}{\delta}=\left\lceil
			\frac{\ln(2n)}{\ln(1+\epsilon/2)}\right\rceil $,
			note that $\delta\le\frac{\ln(1+\epsilon/2)}{\ln(2n)}$.
			Then we have hop-distortion $\left(O(\frac{\log
			n}{\eps}),\red{O(\frac{\log^2 n}{\eps})},\red{h}\right)$,
			hop-path-distortion $\left(\redno{O(\frac{\log^2
			n}{\eps})},\red{h}\right)$, and such that
			\\				$\mathbb{E}_{x\sim\mu}[|f(x)|]\le2\cdot(2n)^{
			\delta}-1\le2\cdot e^{\ln(1+\epsilon/2)}-1=1+\epsilon$.
		\end{enumerate}
	\end{proof}

	We now apply the minimax principle, exactly as in the proof of
	\Cref{thm:UltrametricRamsey}.
	First consider the first assertion of
	\Cref{thm:ClanHopUltrametric}.
	Denote by $\cA$ the family of clan embeddings into an ultrametric with
	hop-distortion
	$\left(16\cdot (k+1),\red{O(k\cdot\log n)},\red{h}\right)$,
	hop-path-distortion $\left(\redno{O(k\cdot\log n)},\red{h}\right)$.
	By \Cref{clm:probabilisticClan}, for every
	probability measure $\mu$ over $V$, there is a clan embedding $(f,\chi)\in
	\cA$ such that $\mathbb{E}_{v\sim\mu}[|f(v)|]\le O(n^{\frac1k})$.
	Using the minimax
	principle (see \Cref{sec:prelim}, preliminaries) we have
	\[
	\min_{\cD\in\mathsf{Dist}(\cA)}
	\max_{\mu\in\mathsf{Dist}(V)}
	\mathbb{E}_{(f,\chi)\sim\cD,\,v\sim\mu}[|f(v)|]
	=
	\max_{\mu\in\mathsf{Dist}(V)}
	\min_{(f,\chi)\in\cA}
	\mathbb{E}_{v\sim\mu}[|f(v)|]
	\le O(n^{\frac1k})~.
	\]
	In words, $\cD$ here is a distribution over clan embeddings, and $\mu$
	is a probability measure over the vertices. Given $(f,\chi)$ and $\mu$,
	the payoff is $\mathbb{E}_{v\sim\mu}[|f(v)|]$, the $\mu$-average number
	of copies. \Cref{clm:probabilisticClan} says that for every probability
	measure $\mu$, one can find a clan embedding whose payoff is at most
	$O(n^{\frac1k})$. Accordingly, by the minimax principle, there is a
	distribution $\cD$ over clan embeddings such that for every probability
	measure $\mu$ over the vertices,
	$\mathbb{E}_{(f,\chi)\sim\cD,\,v\sim\mu}[|f(v)|]\le O(n^{\frac1k})$.
	Let $\cD$ be this distribution. For every vertex $z\in V$, denote by
	$\mu_z$ the probability measure where $\mu_z(z)=1$ (and $\mu_z(y)=0$
	for $y\ne z$). Then
	\[
	\mathbb{E}_{(f,\chi)\sim\cD}[|f(z)|]
	=
	\mathbb{E}_{(f,\chi)\sim\cD,\,v\sim\mu_z}[|f(v)|]
	\le O(n^{\frac1k})~.
	\]
	This proves the first assertion.

	The second assertion is identical, using the second item of
	\Cref{clm:probabilisticClan}. Let $\cA$ now be the family of clan
	embeddings with hop-distortion
	$\left(O(\frac{\log n}{\eps}),\red{O(\frac{\log^2
	n}{\eps})},\red{h}\right)$
	and hop-path-distortion
	$\left(\redno{O(\frac{\log^2 n}{\eps})},\red{h}\right)$.
	For every probability measure $\mu$ over $V$, \Cref{clm:probabilisticClan}
	gives some $(f,\chi)\in\cA$ such that
	$\mathbb{E}_{v\sim\mu}[|f(v)|]\le1+\eps$. Hence, by the minimax principle,
	\[
	\min_{\cD\in\mathsf{Dist}(\cA)}
	\max_{\mu\in\mathsf{Dist}(V)}
	\mathbb{E}_{(f,\chi)\sim\cD,\,v\sim\mu}[|f(v)|]
	=
	\max_{\mu\in\mathsf{Dist}(V)}
	\min_{(f,\chi)\in\cA}
	\mathbb{E}_{v\sim\mu}[|f(v)|]
	\le 1+\eps~.
	\]
	Thus there is a distribution $\cD$ over clan embeddings in $\cA$ such
	that, for every vertex $z\in V$,
	\[
	\mathbb{E}_{(f,\chi)\sim\cD}[|f(z)|]\le1+\eps~.
	\]

	Note that we provided here only an existential proof of the desired
	distribution.
	A constructive proof (with polynomial construction time, and $O(n\cdot\log
	n)$ support size) could be obtained using the multiplicative weights
	update (MWU) method.
	The details are exactly the same as in the constructive proof of clan
	embeddings in \cite{FL21}, and we will skip them.
\end{proof}

\subsection{Proof of \Cref{lem:clanHopUltraMeasure}: distributional h.c. clan
embedding}\label{subsec:ProofClanDistributionalLemma}
For simplicity, as in the Ramsey case, we assume that
$\diam\rhop{h}(G)<\infty$. This assumption can later be removed using the
exact same argument as in \Cref{subsec:inftyDistance}, and we will
not repeat it.
Denote $\phi=\left\lceil\log_2 D\rhop{h}(G) \right\rceil$.
As we assumed polynomial aspect ratio, it follows that $\phi=O(\log n)$.

Unlike in the Ramsey partitions, here we create a cover, and every vertex will
be padded in some cluster. However, some vertices might belong to multiple
clusters. Our goal will be to minimize the total measure of all the clusters
in the cover. In other words, to minimize the measure of
``repeated vertices''.
 The construction is given in \Cref{alg:hier-Clan} and the procedures
 \Cref{alg:clan-c} and \Cref{alg:clan-create-cluster}.

\sloppy In order to create the embedding we will make the call
\texttt{Hop-constrained-clan-embedding}$(G[V],\mu,\red{h},k,2^\phi)$
(\Cref{alg:hier-Clan}), which is a recursive algorithm similar to
\Cref{alg:hier-Ramsey}. The input is $G[X]$ (a graph $G=(V,E,w)$ induced over
a set of vertices $X\subseteq V$), a $\ge1$-measure $\mu$, a hop parameter
$\red{h}$, distortion parameter $k$ and scale $2^i$.
The algorithm invokes the \texttt{clan-cover} procedure (\Cref{alg:clan-c}) to
create clusters $\overline{X}_1,\overline{X}_2,\dots,\overline{X}_s$ of $X$
(for some integer $s$). The \texttt{clan-cover} algorithm also provides
sub-clusters $\underline{X}_q,X_q$ which are used to define the clan embedding
(specifically, to determine which vertices should be removed, and how to
choose the chiefs). See \Cref{fig:HopDistortion} for an illustration.
The \texttt{Hop-constrained-clan-embedding} algorithm now recurses on each
$G[\overline{X}_q]$ for $1\le q\le s$ (with scale parameter $2^{i-1}$), to get
clan embeddings $\{(f_q,\chi_q)\}_{1\le q\le s}$ into ultrametrics
$\{U_q\}_{1\le q\le s}$. Finally, it creates a global embedding into a
ultrametric $U$, by setting a new node $r_U$ to be its root with label $2^{i}$,
and ``hanging'' the ultrametrics $U_1,\dots,U_s$ as the children of $r_U$.
The one-to-many embedding $f$ is simply defined as the union of the one-to-many
embeddings $\{f_q\}_{1\le q\le s}$.
Note, however, that the clusters $\overline{X}_1,\dots,\overline{X}_s$ are
not disjoint.
The algorithm uses the previously provided sub-clusters $X_q\subseteq
\overline{X}_q$  to determine the chiefs (i.e. $\chi$ part) of the
clan embedding.

The algorithm \texttt{clan-cover} creates the clusters $\left\{
\underline{X}_q,X_q,\overline{X}_q\right\} _{q=0}^{s}$ in an iterative manner.
Initially $Y_0=X$. In each step we create a new cluster $\overline{X}_q$ with
sub-clusters $\underline{X}_q,X_q$ using the \texttt{Clan-create-cluster}
procedure (\Cref{alg:clan-create-cluster}). Then the vertices
$\underline{X}_q$ are removed to define $Y_q=Y_{q-1}\backslash
\underline{X}_q$. The process continues until all the vertices are removed.

The \texttt{Clan-create-cluster} procedure is very similar to
\Cref{alg:CreateCluster}.
It returns a triple $\underline{A},A,\overline{A}$, where $\overline{A}$ is
the cluster itself (denoted $A_{j(v)+2}$), and $\underline{A},A$ (denoted
$A_{j(v)},A_{j(v)+1}$) are sub-clusters (it holds that $\underline{A}\subseteq
A\subseteq \overline{A}$).
It will hold that vertices in $A$ are ``padded'' by $\overline{A}$, while
vertices in $\underline{A}$, as well as all their neighborhoods, are
``padded'' by $\overline{A}$. Therefore the procedure \texttt{clan-cover}
can remove the vertices $\underline{A}$ from $Y$.
The cluster $\overline{X}_q$ will have a center vertex $v_q$, and
index $j(v_q)$.
In contrast to the analogous
 \Cref{alg:CreateCluster}, here we also require that either
 $\mu(\overline{A})\le\frac23\mu(Y)$ or $\mu(\underline{A})>\frac13\mu(Y)$.
 This property will later be used in a recursive argument bounding the hop
 path distortion.
This is also the reason that we are using slightly larger parameters here:
$\red{h'}=\red{2(k+1)h}$ and $\rho_i=\frac{2^i}{16(k+1)}$.

\begin{algorithm}[p]
	\caption{$(U,f,\chi)=\texttt{Hop-constrained-clan-embedding}
		(G[X],\mu,\red{h},k,2^i)$}	\label{alg:hier-Clan}
	\DontPrintSemicolon
	\SetKwInOut{Input}{input}\SetKwInOut{Output}{output}
	\Input{Induced graph $G[X]$, parameters $\red{h},k\in\N$, $\ge1$-measure
	$\mu$, scale $i$.}
	\Output{Clan embedding $(f,\chi)$ into ultrametric $U$.}
	\BlankLine
	\If{$X=\{v\}$}{
		set $f(v)=\{v\}$ and $\chi(v)=v$\;
		\Return $(G[X],f,\chi)$
	}
	let $\left(\left\{ \underline{X}_{q},X_{q},\overline{X}_{q}\right\}
	_{q=1}^{s}\right)=\texttt{clan-cover}(G[X],\mu,\red{h},k,2^{i})$\;
	\For{each $q\in[1,\dots,s]$}{
		$(U_{q},f_q,\chi_q)=\texttt{Hop-constrained-clan-embedding}
		(G[\overline{X}_{q}],\mu,\red{h},k,2^{i-1})$\;
	}
	\For{each $z\in X$}{
		set $f(z)=\cup_{q=1}^sf_q(z)$\;
		let $q>0$ be the minimal index such that $z\in X_q$. Set
		$\chi(z)=\chi_q(z)$\;
	}
	let $U$ be the ultrametric formed by taking a root $r_U$ with label
	$2^{i}$, and taking the ultrametrics $U_{1},\dots,U_{s}$ as its children\;
	\Return $(U,f,\chi)$\;
\end{algorithm}
\begin{algorithm}[p]
	\caption{$\left(\left\{ \underline{X}_q,X_q,\overline{X}_q\right\}
	_{q=0}^{s}\right)=\texttt{clan-cover}(G[X],\mu,\red{h},k,2^i)$}
	\label{alg:clan-c}
	\DontPrintSemicolon
	\SetKwInOut{Input}{input}\SetKwInOut{Output}{output}
	\Input{Induced graph $G[X]$, parameters $\red{h},k\in\N$, $\ge1$-measure
	$\mu$, scale $i$.}
	\Output{Cover $\{\overline{X}_{q}\}_{q=1}^{s}$ of $X$, and sub-clusters
	$\underline{X}_q\subseteq X_q\subseteq  \overline{X}_q$.}
	\BlankLine
	let $Y_0=X$\;
	set $q=1$\;

	\While{$Y_{q-1}\ne\emptyset$}{
		$(\underline{X_q},X_q,\overline{X_q})=
		\texttt{Clan-create-cluster}(G[Y_{q-1}],\mu,\red{h},k,
		2^i)$\label{line:callClanCreateClaster}\;
		let $Y_q\leftarrow Y_{q-1}\backslash \underline{X}_{q}$\;
		let $q\leftarrow q+1$\;
	}
	\Return $\left(\left\{ \underline{X}_{q},X_{q},\overline{X}_{q}\right\}
	_{q=1}^{s}\right)$.\tcp*{$s$ is the maximal index of a set}
\end{algorithm}
\begin{algorithm}[p]
	\caption{$(\underline{A},A,\overline{A})=
	\texttt{Clan-create-cluster}(G[Y],\mu,\red{h},k,2^i)$}
	\label{alg:clan-create-cluster}
	\DontPrintSemicolon
	\SetKwInOut{Input}{input}\SetKwInOut{Output}{output}
	\Input{Induced graph $G[Y]$, parameters $\red{h},k\in\N$, $\ge1$-measure
	$\mu$, scale $i$.}
	\Output{Three clusters $(\underline{A},A,\overline{A})$, where
	$\underline{A}\subseteq A\subseteq \overline{A}$.}
	\BlankLine
	Fix $H=G[Y]$, $\rho_i=\frac{2^{i-3}}{2(k+1)}=\frac{2^i}{16(k+1)}$, and
	$\red{h'}=\red{2(k+1)h}$\;
	pick a node $v \in Y$ with maximal $\mu\left(B_{H}^{\red{(i\cdot
	h')}}(v,2^{i-3})\right)$\label{line:CoverCreateClusterChooseCenter}\;
	denote $A_{j}=B_{H}^{\red{(i\cdot h'+j\cdot
	h)}}\left(v,2^{i-3}+j\cdot\rho_i\right)$\;
	let $j(v)\ge 0$ be the minimal integer such that $
	\frac{\mu\left(A_{j+2}\right)}{\mu\left(A_{j}\right)}\leq\left(\frac{\mu
	\left(A_{2(k+1)}\right)}{\mu\left(A_{0}\right)}\right)^{\frac{1}{k}}$
	\newline\phantom{let $j(v)\ge 0$ be the minimal n} and either
	$\mu(A_{j})>\frac13\mu(Y)$ or $\mu(A_{j+2})\le\frac23\mu(Y)$
	\label{line:RG_Clan}\;
	\Return $(A_{j(v)},A_{j(v)+1},A_{j(v)+2})$\;
\end{algorithm}

In the next claim, which is analogous to \Cref{claim:j-v}, we show that the
index chosen in \cref{line:RG_Clan} of \Cref{alg:clan-c} is bounded by $2k$.
\begin{claim}\label{claim:j-v-clan}
	Consider a call to the \texttt{Clan-create-cluster} procedure, the index
	$j(v)$ defined in \cref{line:RG_Clan} of \Cref{alg:clan-create-cluster}
	satisfies $j(v) \leq 2k$.
\end{claim}
\begin{proof}
	Using the terminology from the \texttt{Clan-create-cluster} procedure, let
	$j\in\{0,1,2,\dots,k-1\}$ be the index minimizing
	$\frac{\mu\left(A_{2(j+1)}\right)}{\mu\left(A_{2j}\right)}$. Then
	it holds that
	\[
	\frac{\mu\left(A_{2k}\right)}{\mu\left(A_{0}\right)}=\frac{\mu
	\left(A_{2k}\right)}{\mu\left(A_{2(k-1)}\right)}\cdot\frac{\mu
	\left(A_{2(k-1)}\right)}{\mu\left(A_{2(k-2)}\right)}\cdots\frac{
	\mu\left(A_{2}\right)}{\mu\left(A_{0}\right)}\ge\left(\frac{
	\mu\left(A_{2(j+1)}\right)}{\mu\left(A_{2j}\right)}\right)^{k}~,
	\]
	then $\frac{\mu\left(A_{2(j+1)}\right)}{\mu\left(A_{2j}\right)}\leq\left(
	\frac{\mu\left(A_{2k}\right)}{\mu\left(A_{0}\right)}\right)^{\frac{1}{k}}
	\le\left(\frac{\mu\left(A_{2(k+1)}\right)}{\mu\left(A_{0}\right)}\right)^{
	\frac{1}{k}}$.
	If either $\mu(A_{2j})>\frac13\mu(Y)$ or $\mu(A_{2(j+1)})\le\frac23\mu(Y)$
	then we are done.
	Otherwise,
	$\frac{\mu\left(A_{2(j+1)}\right)}{\mu\left(A_{2j}\right)}>
	\frac{2/3}{1/3}=2$,
	while $\frac{\mu\left(A_{2(j+2)}\right)}{\mu\left(A_{2(j+1)}\right)}\le
	\frac{1}{2/3}<\frac{\mu\left(A_{2(j+1)}\right)}{\mu\left(A_{2j}\right)}\le
	\left(\frac{\mu\left(A_{2(k+1)}\right)}{\mu\left(A_{0}\right)}\right)^{
	\frac{1}{k}}$.
	Note that $2(j+1)\le2k$, and
	$\mu\left(A_{2(j+1)}\right)>\frac{2}{3}\mu(Y)>\frac{1}{3}\mu(Y)$.
	The claim follows.
\end{proof}

The following observation follows by the same argument as
\Cref{obs:diamBound}.
\begin{observation}\label{obs:ClanDiamBound}
	Every $i$-level cluster $X$   (i.e. cluster on which we called
	\texttt{Hop-constrained-clan-embedding} with scale $2^i$) has diameter at
	most $\diam^{\red{(2(i+2)\cdot h')}}(G[X])\le 2^{i}$.
\end{observation}
Using the notation from the \texttt{Clan-create-cluster} procedure which is
executed with parameter $2^{i+1}$ (the one where $i$-level clusters are
created), the observation holds as
$A_{j(v)+2}\subseteq A_{2k+2}=B_{H}^{\red{((i+1)\cdot h'+(2k+2)\cdot
h)}}\left(v,2^{i-2}+(2k+2)\cdot\rho_{i+1}\right)=B_{H}^{\red{((i+2)\cdot
h')}}\left(v,2^{i-1}\right)$.
We will skip the formal proof.

\begin{SCfigure}[][t]\caption{\it \footnotesize Illustration of the procedure
\texttt{clan-cover}, and the analysis of the hop-distortion of the resulting
clan embedding. The algorithm receives the graph induced by an $i$-level
cluster $X$. Initially $Y_0=X$. The procedure iteratively creates clusters
$\{\overline{X}_q\}_{q=1}^{s}$, where each cluster $\overline{X_q}$ comes with
two associated sub-clusters $\underline{X_q}\subseteq
X_q\subseteq\overline{X_q}$. The cluster $\overline{X_q}$ is created in the
graph induced by $Y_{q-1}=X\setminus\cup_{j=1}^{q-1}\underline{X_j}$.
Consider the ball $B=B_{G[X]}\rhop{h}(u,\rho_i)$, and let $q\in[1,s]$ be the
minimal index such that $u\in X_q$ ($q=4$ in the illustration).
For every $q'<q$, we argue that as $u\notin X_{q'}$, the ball $B$ is disjoint
from $\underline{X_{q'}}$. In particular $B\subseteq Y_{q-1}$.
Next we argue that as $u\in X_q$, the ball $B$ is fully contained in
$\overline{X_q}$. It follows that $u$ is ``fully padded'' in
$\overline{X_q}$.
		\label{fig:HopDistortion}}
	\includegraphics[width=.465\textwidth]{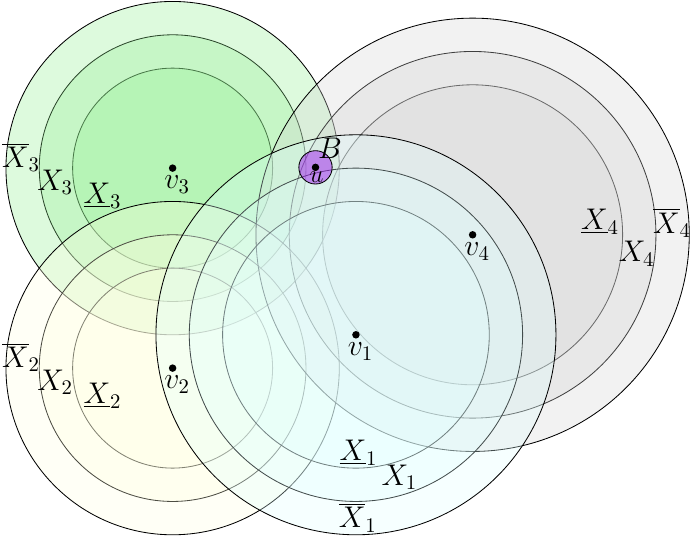}
\end{SCfigure}

\begin{lemma}\label{lem:ClanDistortion}
	The clan embedding $(f,\chi)$ has hop-distortion
	$(16(k+1),\red{O(k\cdot\phi)},\red{h})$.
\end{lemma}
\begin{proof}
	We will prove by induction on the scale $i$, that the algorithm
	\texttt{Hop-constrained-clan-embedding} given input
	$(G[X],\mu,\red{h},k,2^i)$ returns a clan embedding $(f,\chi)$ into
	ultrametric $U$ with hop-distortion $(16(k+1),\red{2(i+2)\cdot
	2(k+1)},\red{h})$. More specifically,
	the inductive hypothesis is that for every $u,v\in X$,
	$d_{G[X]}^{\red{(2(i+2)\cdot h')}}(u,v)\le\min_{v'\in f(v),u'\in
	f(u)}d_{U}(v',u')$
	and $\min_{v'\in f(v)}d_{U}(v',\chi(u))\le16(k+1)\cdot
	d_{G[X]}\rhop{h}(u,v)$.
	The base case is $i=0$, where the graph is a singleton, and is trivial.
	For the general case, consider a pair of vertices $u,v$. Let
	$\overline{X}_1,\dots,\overline{X}_s$ be the clusters created by the call
	to the \texttt{clan-cover} procedure. Note that each cluster
	$\overline{X}_q$ was created from a subset $Y_{q-1}\subseteq X$, has a
	center $v_q$, and index $j(v_q)$ such that
	$\overline{X}_q=A_{j(v_q)+2}=B_{Y_{q-1}}^{\red{((i\cdot h'+(j(v_q)+2)\cdot
	h)}}\left(v_q,2^{i-3}+(j(v_q)+2)\cdot\rho_i\right)$.

	We begin by proving the lower bound: $d_{G[X]}^{\red{(2(i+2)\cdot
	h')}}(u,v)\le\min_{v'\in f(v),u'\in f(u)}d_{U}(v',u')$.
	By \Cref{obs:ClanDiamBound}, $\diam^{\red{(2(i+2)\cdot
	h')}}(G[X])\le 2^{i}$.
	For every $q'\ne q$ such that $v\in \overline{X}_{q'},u\in
	\overline{X}_{q}$, it holds that $\min_{v'\in f_{q'}(v),u'\in
	f_{q}(u)}d_{U}(v',u')=2^{i}\ge\diam^{\red{(2(i+2)\cdot h')}}(G[X])\ge
	d_{G[X]}^{\red{(2(i+2)\cdot h')}}(u,v)$.
	For every $q$ such that $v,u\in \overline{X}_q$, the induction hypothesis
	implies that
	$\min_{v'\in f_{q}(v),u'\in f_{q}(u)}d_{U}(v',u')\ge
	d_{G[\overline{X}_{q}]}^{\red{(2(i+1)\cdot h')}}(u,v)\ge
	d_{G[X]}^{\red{(2(i+2)\cdot h')}}(u,v)$.
	We conclude that  $d_{G[X]}^{\red{(2(i+2)\cdot h')}}(u,v)\le\min_{v'\in
	f(v),u'\in f(u)}d_{U}(v',u')$.

	Next, we turn to the upper bound. See \Cref{fig:HopDistortion} for an
	illustration of our argument.
	If $d\rhop{h}_{G[X]}(u,v)>\rho_i$, then as the label of the root of $U$ is
	$2^i$, we have that
	\[
	\min_{v'\in
	f(v)}d_{U}(v',\chi(u))\le2^{i}=16(k+1)\cdot\rho_{i}<16(k+1)\cdot
	d_{G[X]}\rhop{h}(u,v)~,
	\]
	and thus the lemma holds.
	Otherwise,  $d\rhop{h}_{G[X]}(u,v)\le\rho_i$.
	We will use the terminology from the \texttt{clan-cover}, and
	\texttt{Clan-create-cluster} procedures.
	Let $q$ be the minimal index such that $u\in X_q$.
	Recall that $\chi(u)=\chi_q(u)\in f_q(u)$.
	Denote $B=B_{G[X]}\rhop{h}(u,\rho_i)$. We argue that the entire ball $B$
	is contained in $\overline{X}_q$.
	In particular, it will follow that $v\in \overline{X}_q$, and
	$d_{G[\overline{X}_q]}\rhop{h}(u,v)=d_{G[X]}\rhop{h}(u,v)$ (as the entire
	shortest $\red{h}$-hop path contained in $\overline{X}_q$).
	By induction, it will then follow that
	\[
	\min_{v'\in f(v)}d_{U}(v',\chi(u))\le\min_{v'\in
	f_{q}(v)}d_{U}(v',\chi_q(u))\le 16(k+1)\cdot
	d_{G[\overline{X}_{q}]}\rhop{h}(u,v)=16(k+1)\cdot d_{G[X]}\rhop{h}(u,v)~.
	\]
	We first show that $B\subseteq Y_{q-1}$.
	Suppose for the sake of contradiction otherwise.
	Hence there is some vertex $z\in B= B_{G[X]}\rhop{h}(u,\rho_i)$, such that
	$z\in\underline{X}_{q'}$ for $q'<q$. Let $q'$ be the minimal such index,
	and $z\in \underline{X}_{q'}$ for some $z\in B$. By the minimality of
	$q'$, $B\subseteq Y_{q'-1}$.
	Denote by $v_{q'}$ the center of $X_{q'}$.
	It follows that
	\begin{align*}
	d_{G[Y_{q'-1}]}^{\red{(i\cdot h'+(j(v_{q'})+1)\cdot h)}}(v_{q'},u) & \le
	d_{G[Y_{q'-1}]}^{\red{(i\cdot h'+j(v_{q'})\cdot
	h)}}\left(v_{q'},z\right)+d_{G[Y_{q'-1}]}\rhop{h}\left(z,u\right)\\
	& \le2^{i-3}+j(v_{q'})\cdot\rho_{i}+\rho_{i}=2^{i-3}+(j(v_{q'})+
	1)\rho_{i}~,
	\end{align*}
	thus $u\in A_{j(v_{q'})+1}=X_{q'}$, a contradiction to the choice of
	the index $q$
	as the minimal index such that $u\in X_{q}$.
	It follows that $B\subseteq Y_{q-1}$ as claimed. In particular, for every
	$z\in B$ it holds that
	$d_{G[Y_{q-1}]}\rhop{h}\left(u,z\right)=d_{G[X]}\rhop{h}\left(u,z\right)$.
	As $u\in X_q=A_{j(v_q)+1}$, for every $z\in B$ we have that
	\begin{align*}
	d_{G[Y_{q-1}]}^{\red{(i\cdot h'+(j(v_{q})+2)\cdot h)}}(v_{q},z) & \le
	d_{G[Y_{q-1}]}^{\red{(i\cdot h'+(j(v_{q})+1)\cdot
	h)}}\left(v_{q},u\right)+d_{G[Y_{q-1}]}\rhop{h}\left(u,z\right)\\
	& \le2^{i-3}+(j(v_{q})+1)\cdot\rho_{i}+\rho_{i}=2^{i-3}+(j(v_{q})+2)
	\cdot\rho_{i}~,
	\end{align*}
	implying that $B\subseteq \overline{X}_q$, as required.
\end{proof}

\begin{lemma}\label{lem:ClanPathDistortion}
	The one-to-many embedding $f$ has hop path distortion
	$(\redno{O(k\cdot\log \mu(X))},\red{h})$.
\end{lemma}
\begin{proof}
	The proof follows similar lines to the path distortion argument in
	\cite{BM04multi}. 	Fix $t=8(k+1)$. We will argue by induction on the
	scale $i$ and the number of vertices $n$
	that the algorithm \texttt{Hop-constrained-clan-embedding}, given input
	$(G[X],\mu,\red{h},k,2^i)$, returns a one-to-many embedding $f$ into
	ultrametric $U$ with hop path distortion $(2t\cdot \log_{\nicefrac32}
	\mu(X),\red{h})$. More specifically, for every $\red{h}$-respecting path
	$P=\left(v_0,v_1,\dots,v_m\right)$ (a path such that $\forall a,b\in[m]$,
	$d_{G}\rhop{h}(v_a,v_b)\le d_{P}(v_a,v_b)$) there are  vertices
	$v'_0,\dots,v'_m$ where $v'_i\in f(v_i)$ such that
	$\sum_{i=0}^{m-1}d_U(v'_i,v'_{i+1})\le 2t\cdot \log_{\nicefrac32}
	\mu(X)\cdot \sum_{i=0}^{m-1}w(v_i,v_{i+1})$.

	The base case $i=0$ is trivial, as there is a single vertex and
	no distortion.
	Consider the call \texttt{Hop-constrained-clan-embedding}
	$(G[X],\mu,\red{h},k,2^i)$, which first created the clusters
	$\overline{X}_1,\dots,\overline{X}_s$, and then, for each cluster
	$\overline{X}_q$ created the one-to-many embedding $f_q$, where
	$f=\cup_{q\ge1}f_q$.

	If $s=1$, then $\overline{X}_1=X$. In particular the algorithm will call
	\texttt{Hop-constrained-clan-embedding}
	$(G[X],\mu,\red{h},k,2^{i-1})$ and we will have path distortion  $2t\cdot
	\log_{\nicefrac32} \mu(X)$ by induction.
	Else, recall that $Y_1=X\backslash \underline{X}_1$. If we were to execute
	\texttt{Hop-constrained-clan-embedding}
	$(G[Y_1],\mu,\red{h},k,2^i)$, we would get the clusters
	$\overline{X}_2,\dots,\overline{X}_s$, and on each cluster
	$\overline{X}_q$, we would have obtained the one-to-many embedding $f_q$. Denote
	$\tilde{f}=\cup_{q\ge2}f_q$.
	Then for the global one-to-many embedding $f$ on $X$ it holds that
	$f=f_1\cup\tilde{f}$. Note that $\underline{X}_1\ne\emptyset$, thus
	$Y_1\neq X$. Moreover, $f_1$ was constructed on $\overline{X}_1$ with
	scale $2^{i-1}$.
	Hence by the induction hypothesis, $f_1$ has path distortion  $2t\cdot
	\log_{\nicefrac32} \mu(\overline{X}_1)$, while $\tilde{f}$ has path
	distortion $2t\cdot \log_{\nicefrac32} \mu(Y_1)$.
	By the choice of index $j(v_1)$ in \cref{line:RG_Clan} in
	\Cref{alg:clan-create-cluster}, it holds that either
	$\mu(\overline{X}_1)\le \frac23\mu(X)$, or
	$\mu(Y_1)=\mu(X)-\mu(\underline{X}_1)< \frac23\mu(X)$.
	We assume w.l.o.g. that $\mu(\overline{X}_1)\le \frac23\mu(X)$. Otherwise,
	in the rest of the proof one should swap the roles of
	$\overline{X}_1$ and $Y_1$.

	Consider an $\red{h}$-respecting path $P=\left(x_0,x_1,\dots,x_m\right)$.
	If all vertices of $P$ are fully contained in either $\overline{X_1}$, or
	$Y_1$, then again, we can use the induction hypothesis and be done. We
	will thus assume that it is not the case.
	Let $a_0=0$, and let $Z_0\in\{\overline{X}_1,Y_1\}$ be the cluster
	containing the maximal length of a prefix of $P$:
	$\left(x_{a_0},x_{a_0+1},\dots,x_{b_0}\right)$. In general, the sub-path
	$\{x_{a_q},\dots,x_{b_q}\}$  contained in $Z_q$.
	Unless $b_q=m$, set $a_{q+1}=b_q+1$, and let
	$\left(x_{a_{q+1}},\dots,x_{b_{q+1}}\right)$ be the maximal prefix (among
	the remaining vertices) contained in a single cluster
	$Z_{q+1}\in\{\overline{X}_1,Y_1\}$. Note that $Z_{q}\ne Z_{q+1}$.
	See \Cref{fig:PathDistortion} for an illustration.

	\begin{figure}[t]
		\centering{\includegraphics[scale=.85]{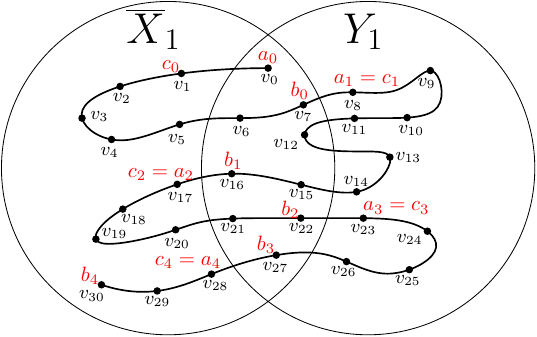}}
		\caption{\label{fig:PathDistortion}\small The path
		$P=\{v_0,\dots,v_{30}\}$ is illustrated in black. Here
		$\mathcal{X}=\{0,2,4\}$ corresponds to indices $i$ for which
		$Z_i=\overline{X}_1$, while $\mathcal{X}=\{1,3\}$ corresponds to
		indices $i$ for which $Z_i=$ $Y_1$.}
	\end{figure}

	We broke $P$ into sub-paths
	$(x_{a_{0}},\dots,x_{b_{0}}),\dots,(x_{a_{q}},\dots,x_{b_{q}}),\dots,
	(x_{a_{\bar{q}}},\dots,x_{b_{\bar{q}}})$. For every $q<\bar{q}$, by the
	maximality of $b_q$, $x_{a_{q+1}}\notin Z_{q}$.
	Furthermore, there is some index $c_0\in[a_0,b_0]$ such that
	$x_{c_0}\notin Z_1$, as otherwise, $x_0,\dots,x_{b_0},x_{a_1}$ are
	contained in $Z_1$, a contradiction to the choice of $b_0$.
	For every $q>0$, set $c_q=a_q$.
	Note that for every $q$, $c_q\notin Z_{q-1}=Z_{q+1}$.
	For every pair of vertices $u,v\in X$ such that $u\notin Y_1$ and $v\notin
	\overline{X}_1$, it holds that $u\in\underline{X}_1$. By the definition of
	$\overline{X}_1,\underline{X}_1$, it follows that
	$d_{G[X]}\rhop{h}(u,v)>2\rho_i$. In particular, for every $q$,
	$ d_{G}\rhop{h}(x{}_{c_{q}},x{}_{c_{q+1}})>2\rho_{i}$.

	For every $q$, using the induction hypothesis on $Z_q$, let
	$x'_{a_q},\dots,x'_{b_{q}}$ be vertices such that $x'_{j}\in f(x_j)$, and
	$\sum_{j=a_{q}}^{b_{q}}d_{U}(x'_{j},x'_{j+1})\le 2t\cdot
	\log_{\nicefrac32}
	\mu(Z_q)\cdot\sum_{j=a_{q}}^{b_{q}}d_{G[X]}(x{}_{j},x{}_{j+1})$.
	Let $\mathcal{X}\subseteq[0,\bar{q}]$ be all the indices that correspond
	to $Z_{q}=\overline{X}_{1}$, and
	$\mathcal{Y}=[0,\bar{q}]\backslash\mathcal{X}$ the indices corresponding
	to $Z_{q}=Y_{1}$. Note that $\mathcal{X}$ is all the even (or odd)
	indices.
	For notational convenience, whenever an index equals $-1$ (respectively,
	$m+1$), we interpret it as $0$ (respectively, $m$).
	 Then
	\begin{align}
	\sum_{j=0}^{m-1}d_{U}(x'_{j},x'_{j+1}) &
	=\ensuremath{\sum_{q\in\mathcal{X}}\sum_{j=a_{q}-1}^{b_{q}}d_{U}(x'_{j},
	x'_{j+1})}+\sum_{q\in\mathcal{Y}}\sum_{j=
	a_{q}}^{b_{q}-1}d_{U}(x'_{j},x'_{j+1})~.\label{eq:PathDistortion}
	\end{align}
	Note that in \cref{eq:PathDistortion}, each of
	$Z\in\{\overline{X_1},Y_1\}$ is ``responsible'' for the sub-paths that are
	fully contained in $Z$, while $\overline{X_1}$ is ``responsible'' for the
	edges connecting consecutive sub-paths (in fact, other than
	the edge connecting the sub-paths $0$ with $1$, and $\overline{q}-1$ with
	$\overline{q}$, each such edge is counted twice, but this is negligible).
	We first bound the second sum in \cref{eq:PathDistortion} (the
	``responsibility'' of $Y_1$).	By the induction hypothesis:
	\begin{align}
		\sum_{q\in\mathcal{Y}}\sum_{j=a_{q}}^{b_{q}-1}d_{U}(x'_{j},
		x'_{j+1}) & \le2t\cdot\log_{\nicefrac{3}{2}}\mu(Y_{1})\cdot\sum_{q\in
		\mathcal{Y}}\sum_{j=a_{q}}^{b_{q}-1}d_{G}(x{}_{j},x{}_{j+1})
		\nonumber \\
		& \le2t\cdot\log_{\nicefrac{3}{2}}\mu(X)\cdot\sum_{q\in\mathcal{
		\mathcal{Y}}}\sum_{j=a_{q}}^{b_{q}-1}d_{G}(x{}_{j},x{}_{j+
		1})\label{eq:PathHopDistortionY1}
	\end{align}
	Next we consider the first sum in \cref{eq:PathDistortion} (the
	``responsibility'' of $\overline{X_1}$).
	Observe that for every $q$, the sub-path of $P$ from $c_q$ to $c_{q+1}$
	has weight at most $\sum_{j=a_{q}}^{b_{q}}d_{G}(x_{j},x_{j+1})$. In
	particular,
	$d_{G[P]}(x_{c_{q}},x_{c_{q+1}})\le\sum_{j=
	a_{q}}^{b_{q}}d_{G}(x_{j},x_{j+1})$.
	As $P$ is $\red{h}$-respecting, for every $q\in\mathcal{X}$, it holds that
	\[
	\sum_{j=a_{q}}^{b_{q}}d_{G}(x_{j},x_{j+1})\ge
	d_{G[P]}(x_{c_{q}},x_{c_{q+1}})\ge
	d_{G}\rhop{h}(x_{c_{q}},x_{c_{q+1}})>2\rho_{i}=\frac{2^{i}}{8(k+1)}=
	\frac{2^{i}}{t}~.
	\]
	As the maximal possible distance in the ultrametric $U$ is $2^{i}$, using
	the induction hypothesis, we get
	\begin{align*}
	&\sum_{j=a_{q}-1}^{b_{q}}d_{U}(x'_{j},x'_{j+1})\\
	 &\qquad \le
	 d_{U}(x'_{a_{q}-1},x'_{a_{q}})+\sum_{j=a_{q}}^{b_{q}-1}d_{U}(x'_{j},
	 x'_{j+1})+d_{U}(x'_{b_{q}},x'_{b_{q}+1})\\
	&\qquad \le2^{i}+t\cdot2\log_{\nicefrac{3}{2}}\mu(\overline{X}_{1})\cdot
	\sum_{j=a_{q}}^{b_{q}-1}d_{G}(x{}_{j},x{}_{j+1})+2^{i}\\
	&\qquad \le\left(t+2t\cdot\log_{\nicefrac{3}{2}}\left(\frac{2}{3}\mu(X)
	\right)+t\right)\cdot\sum_{j=a_{q}}^{b_{q}}d_{G}(x{}_{j},x{}_{j+1})=2t
	\cdot\log_{\nicefrac{3}{2}}\mu(X)\cdot\sum_{j=a_{q}}^{b_{q}}d_{G}(x{}_{j},
	x{}_{j+1})~.
	\end{align*}
	Using the inequality above, and \cref{eq:PathHopDistortionY1} we conclude
	\begin{align*}
	\sum_{j=0}^{m-1}d_{U}(x'_{j},x'_{j+1}) &
	\le2t\cdot\log_{\nicefrac{3}{2}}\mu(X)\cdot\left(\ensuremath{\sum_{q\in
	\mathcal{X}}\sum_{j=a_{q}}^{b_{q}}d_{G}(x{}_{j},x{}_{j+1})}+\sum_{q\in
	\mathcal{Y}}\sum_{j=a_{q}}^{b_{q}-1}d_{G}(x{}_{j},x{}_{j+
	1})\right)\\
	& \le2t\cdot\log_{\nicefrac{3}{2}}\mu(X)\cdot\sum_{j=
	0}^{m-1}d_{U}(x{}_{j},x{}_{j+1})~.
	\end{align*}
\end{proof}

\begin{remark}\label{rem:AspectRatio}
	Similarly to \cite{BM04multi}, in cases where $\phi\ll \log n$, one can
	improve the distortion upper bound in \Cref{lem:ClanPathDistortion} to
	$O(k\cdot \phi)$. The proof technique is the same, where we prove by
	induction that the one-to-many embedding created for scale $2^i$ has hop
	path distortion $(2t\cdot i,\red{h})$.
\end{remark}

	Similarly to the notation introduced before \Cref{lem:frac-survival},
	we associate
	with every cluster $X$ a graph $H(X)$ from which it was created,
	and a center
	$r(X)$. In the present clan-cover construction, the clusters on
	which we recurse
	are the clusters $\overline{X}_q$ created by \texttt{clan-cover}
	(\Cref{alg:clan-c}) using \texttt{Clan-create-cluster}
	(\Cref{alg:clan-create-cluster}). Thus, if $X=\overline{X}_q$
	was created in
	the $q$th iteration of a call to \texttt{clan-cover}, when the
	remaining vertex
	set was $Y_{q-1}$ and the chosen center was $v_q$, then we set
	$H(X)=G[Y_{q-1}]$ and $r(X)=v_q$. For the initial cluster $V$, set
	$H(V)=G[V]$,
	and let $r(V)$ be a vertex maximizing
	$\mu\left(B_G^{\red{((\phi+1)\cdot h')}}(v,2^{\phi-2})\right)$.
\begin{lemma}\label{lem:ClanMesureBound}
	For an $i$-level cluster $X$, it holds that\\\phantom{a}\hfill
	$\mathbb{E}_{x\sim
	X}[|f(x)|]\le\mu(X)\cdot\mu\left(B_{H(X)}^{\red{((i+1)\cdot
	h')}}(r(X),2^{i-2})\right)^{\frac{1}{k}}$.
\end{lemma}
\begin{proof}
	Note that an $i$-level cluster was created from an $i+1$-level cluster. In
	particular, it was created by a call to the \texttt{Clan-create-cluster}
	procedure with parameter $2^{i+1}$.
	We prove the lemma by induction on $i$. Consider first the base case
	where $i=0$, or more generally, when
	$X$ is a singleton $\{v\}$. The embedding will be into a single vertex. In
	particular, $\mathbb{E}_{x\sim
	X}[|f(x)|]=\mu(\{v\})\le\mu(X)\cdot\mu\left(B_{H(X)}^{\red{(0)}}(r(X),
	2^{-2})\right)^{\frac{1}{k}}$.
	Next we assume by induction that the claim holds for every
	$i$-level cluster.
	Consider an $i$-level cluster $X$. $X$ was created from the graph $H(X)$,
	and has center $r(X)$ (for the case $i=\phi$, let $H(X)=V$, and $r(x)$ to
	be the vertex maximizing $\mu\left(B_{H(X)}^{\red{((\phi+1)\cdot
	h')}}(r(X),2^{\phi-2})\right)$~).

	The algorithm calls $\texttt{clan-cover}(G[X],\mu,\red{h},k,2^{i})$ to
	obtain the $i-1$-level
	clusters $\overline{X}_{1},\dots,\overline{X}_{s}$, each with center
	$v_{q}$, which was created when the graph was $Y_{q-1}$.
	In particular, there is a sub-cluster $\underline{X}_{q}$ such that
	$Y_{q}=Y_{q-1}\backslash\underline{X}_{q}$. Note that each vertex
	$v\in X$ belongs to exactly one sub-cluster $\underline{X}_{q}$,
	hence
	\begin{equation}
	\sum_{q=1}^{s}\mu(\underline{X}_{q})=\mu(X)\label{eq:SubClustersPartition}
	\end{equation}
	By our choice of index in \cref{line:RG_Clan} of
	\Cref{alg:clan-create-cluster}, for every $q$ it holds that
	\begin{equation}
	\mu(\overline{X}_{q})\le\mu(\underline{X}_{q})\cdot\left(\frac{\mu
	\left(B_{Y_{q-1}}^{\red{((i+1)\cdot h')}}(v_{q},2^{i-2})\right)}{\mu
	\left(B_{Y_{q-1}}^{\red{(i\cdot h')}}(v_{q},2^{i-3})\right)}\right)^{
	\frac{1}{k}}~.
	\label{eq:ClanMeasureClusterIneq}
	\end{equation}
	\sloppy For every $q$, as $Y_{q-1}\subseteq X\subseteq H(X)$, and by
	the choice of
	$r(X)$ as the vertex maximizing $\mu\left(B_{H(X)}^{\red{((i+1)\cdot
	h')}}(r(X),2^{i-2})\right)^{\frac{1}{k}}$,
	it holds that
	\begin{equation}
	\mu\left(B_{Y_{q-1}}^{\red{((i+1)\cdot
	h')}}(v_{q},2^{i-2})\right)\le\mu\left(B_{H(X)}^{\red{((i+1)\cdot
	h')}}(v_{q},2^{i-2})\right)\le\mu\left(B_{H(X)}^{\red{((i+1)\cdot
	h')}}(r(X),2^{i-2})\right)~.\label{eq:ClanMusureOfSuperVertex}
	\end{equation}
	Using the induction hypothesis, for every $q$ we obtain
	\begin{align*}
	\mathbb{E}_{x\sim\overline{X}_{q}}[|f_{q}(x)|] &
	\le\mu(\overline{X}_{q})\cdot\mu\left(B_{Y_{q-1}}^{\red{(i\cdot
	h')}}(v_{q},2^{i-3})\right)^{\frac{1}{k}}\\
	& \overset{(\ref{eq:ClanMeasureClusterIneq})}{\le}\mu(\underline{X}_{q})
	\cdot\left(\frac{\mu\left(B_{Y_{q-1}}^{\red{((i+1)\cdot h')}}(v_{q},
	2^{i-2})\right)}{\mu\left(B_{Y_{q-1}}^{\red{(i\cdot h')}}(v_{q},2^{i-3})
	\right)}\right)^{\frac{1}{k}}\cdot\mu\left(B_{Y_{q-1}}^{\red{(i\cdot
	h')}}(v_{q},2^{i-3})\right)^{\frac{1}{k}}\\
	& =\mu(\underline{X}_{q})\cdot\mu\left(B_{Y_{q-1}}^{\red{((i+1)\cdot
	h')}}(v_{q},2^{i-2})\right)^{\frac{1}{k}}\\
	& \overset{(\ref{eq:ClanMusureOfSuperVertex})}{\le}\mu(\underline{X}_{q})
	\cdot\mu\left(B_{H(X)}^{\red{((i+1)\cdot h')}}(r(X),2^{i-2})\right)^{
	\frac{1}{k}}~.
	\end{align*}
	As $f=\cup_{q=1}^{s}f_{q}$, we conclude
	\begin{align*}
	\mathbb{E}_{x\sim X}[|f(x)|] & =\sum_{q=1}^{s}\mathbb{E}_{x\sim
	X}[|f_{q}(x)|]\\
	& \le\sum_{q=1}^{s}\mu(\underline{X}_{q})\cdot\mu\left(B_{H(X)}^{\red{((i+
	1)\cdot h')}}(r(X),2^{i-2})\right)^{\frac{1}{k}}\\
	& \overset{(\ref{eq:SubClustersPartition})}{=}\mu(X)\cdot\mu
	\left(B_{H(X)}^{\red{((i+1)\cdot h')}}(r(X),2^{i-2})\right)^{
	\frac{1}{k}}~.
	\end{align*}
\end{proof}

Using \Cref{lem:ClanMesureBound} on $V$ with $i=\logdiam$ we have
$\mathbb{E}_{x\sim V}[|f(x)|]\le\mu(V)\cdot\mu\left(B_{G}^{\red{((\phi+1)\cdot
h')}}(r(V),2^{\phi-2})\right)^{\frac{1}{k}}\le\mu(V)^{1+\frac{1}{k}}$.
\Cref{lem:clanHopUltraMeasure} now follows by combining
\Cref{lem:ClanDistortion}, \Cref{lem:ClanPathDistortion}, and
\Cref{lem:ClanMesureBound}.

An observation that will be useful in some of our applications (specifically
for \Cref{thm:Subgraph} and \Cref{cor:hBoundedSubgraph}) is the following:
\begin{observation}\label{obs:singleCopy}
	Suppose that \Cref{lem:clanHopUltraMeasure} is applied on a measure $\mu$
	where for some vertex $r\in V$, $\mu(r)>\frac12\mu(V)$. Then in the
	resulting clan embedding $(f,\chi)$ it will hold that $|f(r)|=1$.
\end{observation}
\begin{proof}
	Note that in \Cref{alg:clan-create-cluster}, we pick a center $v_1$ that
	maximizes  $\mu\left(B_{H}^{\red{(i\cdot h')}}(v,2^{i-3})\right)$. Note
	that every ball containing $r$ has measure larger than $\frac12\mu(V)$,
	while every ball that does not contain $r$ has measure less than
	$\frac12\mu(V)$. Thus necessarily $r\in B_{H}^{\red{(i\cdot
	h')}}(v,2^{i-3})$.
	Next, observe that for every index $j(v_q)$ chosen in \cref{line:RG_Clan}
	it will hold that $B_{H}^{\red{(i\cdot
	h')}}(v,2^{i-3})\subseteq\underline{X}_1$. Hence $r\notin Y_1$, and
	therefore will not be duplicated: i.e. $r$ belong only to $\overline{X}_1$
	among all the created clusters. The observation follows now by induction.
\end{proof}

\subsection{Alternative clan embedding}
In this section, in a similar spirit to what we did in \Cref{subsec:AltRamsey},
we modify the \texttt{Hop-constrained-clan-embedding} algorithm by making
changes to the \texttt{Clan-create-cluster} procedure. As a result, we obtain a
theorem with a slightly different trade-off between hop-stretch and
hop-distortion from that of \Cref{thm:ClanHopUltrametric}. Specifically, by
adding an additional factor of $O(\log\log n)$ in the distortion, we can
allow the hop-stretch to be as small as $O(\log\log n)$, and  completely avoid
any dependence on the aspect ratio.
We restate the theorem for convenience:

\ClanHopUltrametricAlt*

As usual, the key step will be to prove a distributional lemma:
\begin{lemma}\label{lem:clanHopUltraMeasureAlt}
	Consider an $n$-vertex graph $G=(V,E,w)$ with polynomial aspect ratio,
	$(\ge1)$-measure $\mu:V\rightarrow\mathbb{R}_{\ge1}$, and parameters
	$k,\red{h}\in [n]$. Then there is a clan embedding $(f,\chi)$ into an
	ultrametric with
	hop-distortion $\left(O(k\cdot\log\log \mu(V)),\red{O(k\cdot\log\log
	\mu(V))},\red{h}\right)$, hop-path-distortion $\left(\redno{O(k\cdot\log
	\mu(V)\cdot\log\log \mu(V))},\red{h}\right)$, and
	such that $\mathbb{E}_{v\sim\mu}[|f(v)|]\le\mu(V)^{1+\frac1k}$.
\end{lemma}

The proof of \Cref{thm:ClanHopUltrametricAlt} using
\Cref{lem:clanHopUltraMeasureAlt} follows the exact same lines as the proof of
\Cref{thm:ClanHopUltrametric} using \Cref{lem:clanHopUltraMeasure}, and thus
we will skip it.

\subsection{Proof of \Cref{lem:clanHopUltraMeasureAlt}: alternative
distributional h.c. clan embedding}
For the embedding of \Cref{lem:clanHopUltraMeasureAlt} we will use the exact
same \texttt{Hop-constrained-clan-embedding} (\Cref{alg:hier-Clan}), with the
only difference that instead of using the \texttt{Clan-create-Cluster}
procedure (\Cref{alg:clan-create-cluster}), we will use the
\texttt{Clan-create-cluster-alt} procedure
(\Cref{alg:clan-create-cluster-alt}).
Similarly to the proof of \Cref{lem:clanHopUltraMeasure} we will assume that
$\diam\rhop{h}(G)<\infty$. The exact same argument from
\Cref{subsec:inftyDistance} can be used to remove this assumption.
We begin by making the call
\texttt{Hop-constrained-clan-embedding}$(G[V],\mu,\red{h},k,2^\phi)$, where
the call to \texttt{Clan-create-Cluster} in \cref{line:callClanCreateClaster}
of \Cref{alg:clan-c} is replaced by a call to
\texttt{Clan-create-cluster-alt}.

\begin{algorithm}[t]
	\caption{$(\underline{A},A,\overline{A})=
	\texttt{Clan-create-cluster-alt}(G[Y],\mu,\red{h},k,2^i)$}
	\label{alg:clan-create-cluster-alt}
	\DontPrintSemicolon
	\SetKwInOut{Input}{input}\SetKwInOut{Output}{output}
	\Input{Induced graph $G[Y]$, parameters $\red{h},k\in\N$, $\ge1$-measure
	$\mu$, scale $i$.}
	\Output{Three clusters $(\underline{A},A,\overline{A})$, where
	$\underline{A}\subseteq A\subseteq \overline{A}$.}
	\BlankLine
	set $L=\lceil 1+\log\log\mu(Y)\rceil$ and $\Delta=2^i$\;
	let $v \in Y$ with minimal $\mu\left(B_{G[Y]}^{\red{(2k\cdot L\cdot
	h)}}(v,\frac14\Delta)\right)$\;
	\If{$\mu\left(B_{G[Y]}^{\red{(2k\cdot L\cdot
	h)}}(v,\frac14\Delta)\right)>\frac12\cdot\mu(Y)$}{
		\Return $(Y,Y,Y)$\label{line:AltCreateClusterReturnXCaln}\tcp*{Note
		that for every $u,v\in Y$, $d_{G[Y]}^{\red{(4k\cdot L\cdot h)}}(u,v)\le
		\frac\Delta2$}
	}
	denote $A_{a,j}=B_{G[Y]}^{\red{((2k\cdot a+j)\cdot
	h)}}\left(v,(a+\frac{j}{2k})\cdot\frac{\Delta}{4L}\right)$\;
	let $a\in[0,L-1]$ such that
	$\mu\left(A_{a,0}\right)\ge\frac{\mu\left(A_{a+1,0}\right)^{2}}{
	\mu(Y)}$\label{line:firstStepCaln}\;
	let $j\in [0,2(k-1)]$ such that
	$\mu\left(A_{a,j+2}\right)\le\mu\left(A_{a,j}\right)\cdot\left(\frac{\mu
	\left(A_{a+1,0}\right)}{\mu\left(A_{a,0}\right)}\right)^{\frac{1}{k}}
	$\label{line:SecondStepCaln}\;
	\Return $(A_{a,j},A_{a,j+1},A_{a,j+
	2})$\label{line:AltCreateClusterReturnUsual}\;
\end{algorithm}

The proofs of the following two claims are identical to those of
\Cref{clm:RamseyAltAChoise} and \Cref{clm:RamseyAltJChoise}, so we
will skip them.

\begin{claim}
	In \cref{line:firstStepCaln} of \Cref{alg:clan-create-cluster-alt}, there
	is an index $a\in [0,L-1]$ such that
	$\mu\left(A_{a,0}\right)\ge\frac{\mu\left(A_{a+1,0}\right)^{2}}{\mu(Y)}$.
\end{claim}
\begin{claim}
	In \cref{line:SecondStepCaln} of \Cref{alg:clan-create-cluster-alt}, there
	is an index $j\in [0,2(k-1)]$ such that
	$\mu\left(A_{a,j+2}\right)\le\mu\left(A_{a,j}\right)\cdot\left(
	\frac{\mu\left(A_{a+1,0}\right)}{\mu\left(A_{a,0}\right)}\right)^{
	\frac{1}{k}}$.
\end{claim}

The following observation is analogous to \Cref{obs::diamBoundAlt}, however
the proof is somewhat simpler (as there is no set $M$ from which we have to
choose the center, and we take into account the measure of all vertices).

\begin{observation}
	Every cluster $i$-level cluster $X$ (i.e. cluster on which we called
	\texttt{Hop-constrained-clan-embedding} with scale $2^i$) has diameter at
	most $\diam^{\red{(4k\cdot L\cdot h)}}(G[X])\le 2^{i}$.
\end{observation}
The case where the cluster is returned in
\cref{line:AltCreateClusterReturnUsual} is identical (to the proof of
\Cref{obs::diamBoundAlt}). For the case where we returned the cluster in
\cref{line:AltCreateClusterReturnXCaln} it holds that for every $v\in Y$,
$\mu\left(B_{G[Y]}^{\red{(2k\cdot L\cdot
h)}}(v,\frac14\Delta)\right)>\frac12\cdot\mu(Y)$ (as $v$ is a minimizer of
the measure). Hence for every two points $x,y\in Y$, the balls
$B_{G[Y]}^{\red{(2k\cdot L\cdot h)}}(x,\frac14\Delta),B_{G[Y]}^{\red{(2k\cdot
L\cdot h)}}(y,\frac14\Delta)$ intersect and hence $d_{G[Y]}\rhop{4k\cdot
L\cdot h}(x,y)\le\frac{\Delta}{2}=2^{i}$. In particular, $\diam^{\red{(4k\cdot
L\cdot h)}}(G[X])\le 2^{i}$.

The proof of the following lemma follows the same lines as
\Cref{lem:ClanDistortion}, and we will skip it.
\begin{lemma}\label{lem:ClanDistortionAlt}
	The clan embedding $(f,\chi)$ has hop-distortion\\\phantom{a}\hfill
	$(O(k\cdot \log\log\mu(V)),\red{O(k\cdot\log\log\mu(V))},\red{h})$.
\end{lemma}

Our next goal is to prove the path distortion guarantee  (analogous to
\Cref{lem:ClanPathDistortion}).
We observe the following:
\begin{observation}\label{obs:halfTheMeasure}
	Consider a cluster $\overline{A}$ returned by the call\\
	\texttt{Clan-create-cluster-alt}$(G[Y],\mu,\red{h},k,2^i)$. Then either
	$\overline{A}=\underline{A}=Y$, or $\mu(\overline{A})\le \frac12\mu(Y)$.
\end{observation}
\begin{proof}
	Consider the execution of the \texttt{Clan-create-cluster-alt} procedure
	on input $(G[Y],\mu,\red{h},k,2^i)$.
	If the algorithm halts in \cref{line:AltCreateClusterReturnXCaln}, then
	$\overline{A}=\underline{A}=Y$ and we are done.
	Otherwise, we have that $\mu\left(B_{G[Y]}^{\red{(2k\cdot L\cdot
	h)}}(v,\frac14\Delta)\right)\le\frac12\cdot\mu(Y)$, while
	$\overline{A}=A_{a,j+2}$ for some $a\le L-1$ and $j\le 2(k-1)$. Hence $
	A_{a,j+2}\subseteq A_{L-1,2k}=B_{G[Y]}^{\red{((2k\cdot(L-1)+2k)\cdot
	h)}}\left(v,(L-1+\frac{2k}{2k})\cdot\frac{\Delta}{4L}\right)=B_{G[Y]}^{
	\red{(2k\cdot
	L\cdot h)}}\left(v,\frac{\Delta}{4}\right)$, implying that
	$\mu(\overline{A})\le \frac12\cdot\mu(Y)$, as required.
\end{proof}
Given \Cref{obs:halfTheMeasure}, the proof of the following lemma follows the
exact same lines as \Cref{lem:ClanPathDistortion} (replacing $O(k)$ with
$O(k\cdot\log\log\mu(X))$), and we will skip it (note that
\Cref{rem:AspectRatio} will also hold).
\begin{lemma}\label{lem:ClanPathDistortionAlt}
	The one-to-many embedding $f$ has hop path distortion\\
	$\left(\redno{O(k\cdot\log \mu(V)\cdot\log\log \mu(V))},\red{h}\right)$.
\end{lemma}

	Finally, we turn to bounding $\mathbb{E}_{x\sim\mu}[|f(x)|]$, which is the
	main difference between the two algorithms.
Recall that for a cluster $\underline{X}$, we denote by $H(X)$ the graph at
the time of the creation of $\underline{X},X,\overline{X}$.
\begin{lemma}\label{lem:ClanMesureBoundAly}
	It holds that $\mathbb{E}_{v\sim V}[|f(v)|]\le\mu(V)^{1+\frac{1}{k}}$.
\end{lemma}
\begin{proof}
	We prove the lemma by induction on $i$ and $|X|$.
	Specifically, the induction hypothesis is that for an $i$-level cluster
	$X$, it holds that $\mathbb{E}_{x\sim
	X}[|f(x)|]\le\mu(X)^{1+\frac{1}{k}}$.
	Consider first the base case
	where the cluster $X$ created with scale $i=0$, or when
	$X$ is a singleton $\{v\}$. The embedding will be into a single
	vertex. In particular, $\mathbb{E}_{x\sim
	X}[|f(x)|]=\mu(\{v\})\le\mu\left(X\right)^{1+\frac{1}{k}}$.

	Next we assume that the claim holds for every $i-1$-level cluster.
	Consider an $i$-level cluster $X$.
	We called
	 $\texttt{clan-cover}(G[X],\mu,\red{h},k,2^{i})$, and obtained clusters
 	$\overline{X}_{1},\dots,\overline{X}_{s}$. Then for each cluster
 	$\overline{X}_{q}$ the algorithm created a one-to-many embedding $f_q$. If
 	$s=1$, then there is only a single cluster and the lemma will follow by
 	the induction hypothesis on $i-1$.
	Else, we obtain a non-trivial cover.
	Recall that after creating $\overline{X}_1$, the set of remaining vertices
	is denoted $Y_1=X\backslash \underline{X}_1$.
	Note that if we were to call \texttt{clan-cover} on input
	$(G[Y_1],\mu,\red{h},k,2^i)$, the algorithm will first create a cover
	$\overline{X}_2,\dots,\overline{X}_s$ of $Y_1$, and then will recursively
	create one-to-many embeddings for each $\overline{X}_q$.
	Denote by  $\tilde{f}$ the one-to-many embedding created by such a call.
	It holds that $f=f_1\cup\tilde{f}$.
	Therefore, we can analyze the rest of the process as if the algorithm did
	a recursive call on $Y_1$ rather than on each cluster $\overline{X}_q$.
	Since $|\overline{X}_1|,|Y_1|<|X|$, the induction hypothesis implies that
	$\mathbb{E}_{x\sim
	\overline{X}_1}[|f_1(x)|]\le\mu(\overline{X}_1)^{1+\frac{1}{k}}$ and
	$\mathbb{E}_{x\sim Y_1}[|\tilde{f}(x)|]\le\mu(Y_1)^{1+\frac{1}{k}}$.
	As the \texttt{Padded-partition} algorithm did not return a trivial
	partition of $X$, it defined  $A_{a,j}=B_{G[Y]}^{\red{((2k\cdot a+j)\cdot
	h)}}\left(v,(a+\frac{j}{2k})\cdot\frac{\Delta}{4L}\right)$,
	and returned
	$(\underline{X}_1,X_1,\overline{X}_1)=(A_{a,j},A_{a,j+1},A_{a,j+2})$. The
	indices $a,j$ were chosen such that
	$\mu\left(A_{a,0}\right)\ge\frac{\mu\left(A_{a+1,0}\right)^{2}}{\mu(X)}$
	and $\mu\left(A_{a,j+2}\right)\le\mu\left(A_{a,j}\right)\cdot\left(\frac{
	\mu\left(A_{a+1,0}\right)}{\mu\left(A_{a,0}\right)}\right)^{\frac{1}{k}}$,
	implying
	 $$\mu\left(\overline{X}_{1}\right)\le\mu\left(\underline{X}_{1}\right)
	 \cdot\left(\frac{\mu\left(A_{a+1,0}\right)}{\mu\left(A_{a,0}\right)}
	 \right)^{\frac{1}{k}}\le\mu\left(\underline{X}_{1}\right)\cdot\left(
	 \frac{\mu\left(X\right)}{\mu\left(A_{a+1,0}\right)}\right)^{\frac{1}{k}}
	 ~.$$
	It follows that
	\begin{align*}
	\mathbb{E}_{x\sim\overline{X}_{1}}[|f_{1}(x)|] &
	\le\mu(\overline{X}_{1})^{1+\frac{1}{k}}=\mu(\overline{X}_{1})\cdot\mu(
	\overline{X}_{1})^{\frac{1}{k}}\\
	& \le\mu\left(\underline{X}_{1}\right)\cdot\left(\frac{\mu\left(X\right)}{
	\mu\left(A_{a+1,0}\right)}\right)^{\frac{1}{k}}\cdot\mu(
	\overline{X}_{1})^{\frac{1}{k}}\le\mu\left(\underline{X}_{1}\right)\cdot
	\mu\left(X\right)^{\frac{1}{k}}~,
	\end{align*}
	where the last inequality follows as $A_{a+1,0}\supseteq
	A_{a,j+2}=\overline{X}_1$. As $X=\underline{X}_1\cupdot Y_1$, we conclude
	\begin{align*}
	\mathbb{E}_{x\sim X}[|f(x)|] & =\mathbb{E}_{x\sim
	\overline{X}_{1}}[|f_{1}(x)|]+\mathbb{E}_{x\sim
	Y_{1}}[|\tilde{f}(x)|]\le\mu\left(\underline{X}_{1}\right)\cdot\mu\left(X
	\right)^{\frac{1}{k}}+\mu\left(Y_{1}\right)\cdot\mu\left(X\right)^{
	\frac{1}{k}}=\mu\left(X\right)^{1+\frac{1}{k}}~.
	\end{align*}
\end{proof}
\Cref{lem:clanHopUltraMeasureAlt} now follows by combining
\Cref{lem:ClanDistortionAlt}, \Cref{lem:ClanPathDistortionAlt}, and
\Cref{lem:ClanMesureBoundAly}.

An observation that will be useful in some of our applications (specifically
for \Cref{thm:SubgraphAlt} and \Cref{cor:hBoundedSubgraphAlt}) is
the following:
\begin{observation}\label{obs:singleCopyAlt}
	Suppose that  \Cref{lem:clanHopUltraMeasureAlt} is applied on a measure
	$\mu$ where for some vertex $r\in V$, $\mu(r)>\frac12\mu(V)$. Then in the
	resulting clan embedding $(f,\chi)$ it will hold that $|f(r)|=1$.
\end{observation}
\begin{proof}
	Consider an execution of the
	\texttt{Clan-create-cluster-alt}$(G[Y],\mu,\red{h},k,2^i)$
	(\Cref{alg:clan-create-cluster-alt}) procedure where $r\in Y$.
	The algorithm picks a center $v$ that minimizes
	$\mu\left(B_{G[Y]}^{\red{(2k\cdot L\cdot h)}}(v,\frac14\Delta)\right)$. Note
	that every ball containing $r$ has measure strictly larger than
	$\frac12\mu(V)\ge \frac12\mu(Y)$, while every ball that does not contain $r$ has measure
	strictly less than $\frac12\mu(Y)$.
	Thus by \Cref{obs:halfTheMeasure}, either $\underline{A}=Y$, or
	$\mu(\overline{A})\le\frac12\mu(Y)$, which implies $r\notin \overline{A}$.
	It follows that in any call to the \texttt{clan-cover} procedure that
	returns $(\overline{X}_1,\dots,\overline{X}_s)$, $r$ will belong only to
	$\overline{X}_s$. The observation now follows by induction.
\end{proof}

\section{Subgraph preserving one-to-many embedding}

In our clan embeddings
(\Cref{thm:ClanHopUltrametric,thm:ClanHopUltrametricAlt}) we had the
hop-path-distortion property, which ensures that every path $P$ in $G$ has a
``corresponding'' copy $P_T$ in the tree $T$ of similar weight. However, as
$T$ contains many vertices that are not copies of $G$ vertices, it is not
immediately clear to which subgraph of $G$, $P_T$ corresponds. To clarify this
point, we introduce the notion of \emph{Path Tree One-To-Many Embedding}
(\Cref{def:OTM-PathTreeEmbedding}). Here all $T$ vertices are copies of $V$,
and every edge $e\in T$ corresponds to a path $P^T_e$ between its endpoints.
Now, there is a natural way to translate from subgraphs of $T$ back to
subgraphs of $G$. Morally, \Cref{cor:hBoundedSubgraph,cor:hBoundedSubgraphAlt}
below are just a translation of
\Cref{lem:clanHopUltraMeasure,lem:clanHopUltraMeasureAlt} (respectively) to
the language of path tree one-to-many embeddings.

A major drawback of \Cref{cor:hBoundedSubgraph,cor:hBoundedSubgraphAlt} is
that it only applies on $\red{h}$-respecting subgraphs.
Specifically, in order to apply the ``subgraph preservation'' guarantee, the
subgraph $H$ in question must be ``local'' (formally it is not allowed to
contain a $u-v$ path of weight less than $d_G\rhop{h}(u,v)$ ).
This is problematic, as in some of our applications the solution subgraph is
not $\red{h}$-respecting (e.g. in  h.c. group Steiner forest).
On the other hand, clearly one cannot simply drop the $\red{h}$-respecting
condition from  \Cref{cor:hBoundedSubgraph,cor:hBoundedSubgraphAlt} (or
\Cref{thm:ClanHopUltrametric,thm:ClanHopUltrametricAlt}).
Indeed, consider the $n$-path (a graph consisting only of the edges of the
length-$n$ path: $(v_0,\dots,v_{n-1})$) with hop parameter $\red{h}=1$, and
let $H=G$. The distance between every pair of copies $v'_0\in
f(v_0),v'_{n-1}\in f(v_{n-1})$ has to be $\infty$. Thus any ``image'' of the
entire graph has to be of weight $\infty$.

In \Cref{thm:Subgraph,thm:SubgraphAlt} we give a weaker guarantee which can be
applied to arbitrary subgraphs. Specifically, given a subgraph $H$ of $G$, we
construct a (not necessarily connected) subgraph $H'$ of $T$ where for every
$u,v\in H$ such that $\hop_H(u,v)\le\red{h}$, $H'$ contains a pair of copies
$u'\in f(u)$ and $v'\in f(v)$ in the same connected component. Further,
$w_T(H')\le O(k\cdot\log n)\cdot w(H)$.
This guarantee will be sufficient for our applications.

This section is devoted to proving \Cref{thm:Subgraph,thm:SubgraphAlt}, which
ensures that any arbitrary subgraph $H$ of $G$ will have a ``corresponding''
copy in the image.
We begin the section by formally defining Path Tree One-To-Many Embedding
(a very similar definition was given by \cite{HHZ21}). See \Cref{fig:PathTree}
for an illustration.
\begin{definition}[Path Tree One-To-Many
Embedding]\label{def:OTM-PathTreeEmbedding}
	A one-to-many path tree embedding of a weighted graph $G = (V, E, w_G)$
	consists of a dominating one-to-many embedding $f$ into a tree $T$ such
	that  $f(V)=V(T)$, and associated paths $\{P^T_{e}\}_{e\in E(T)}$, where
	for every edge $e=\{u',v'\}\in E(T)$ where $u'\in f(u)$ and $v'\in f(v)$,
	$ P^T_{e}$ is a path from $u$ to $v$ in $G$ of weight at most
	$w_T(\{u',v'\})$.

	For every two vertices $u',v'\in V(T)$, let $e_1,e_2,\dots,e_s$ be the
	unique simple path between them in $T$. We denote:
	$P^T_{u'v'}=P^T_{e_1}\circ P^T_{e_2}\circ\cdots\circ P^T_{e_s}$ to be the
	concatenation of all the associated paths (note that $P^T_{u'v'}$ is not
	necessarily simple).
	We say that $f$ has hop-bound $\red{\beta}$ if for every $u',v'\in V(T)$
	which are connected in $T$, $\hop(P^T_{u'v'})\le\red{\beta}$.
\end{definition}
\begin{figure}[t]
	\centering{\includegraphics[scale=.65]{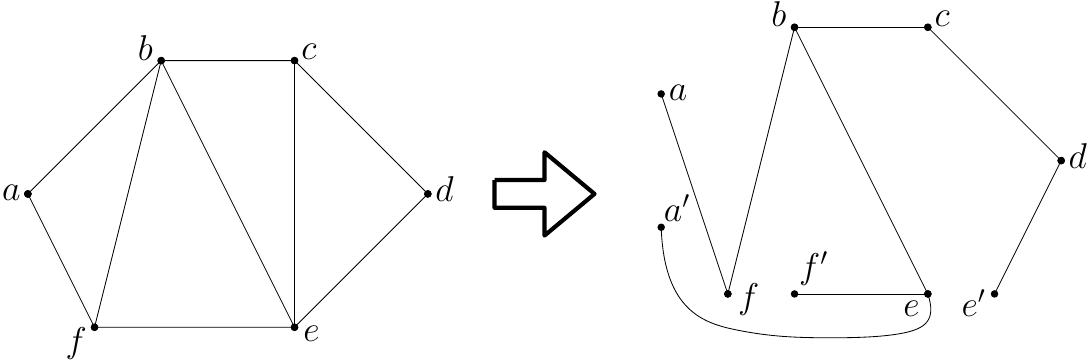}}
	\caption{\label{fig:PathTree}\small The graph $G$ is illustrated on the
	left, while the path tree one-to-many embedding into the tree $T$ is
	illustrated on the right.
	Each edge $e\in T$ has an associated path $P^T_e$, for example
	$P^T_{e',d}=(e,d)$, and $P^T_{a',e}=(a,f,e)$. Each pair of vertices
	$u,v\in V(T)$, has a (not necessarily simple) induced path $P^T_{u,v}$,
	for example $P^T_{f,e'}=(f,b,c,d,e)$, and $P^T_{a',e'}=(a,f,e,b,c,d,e)$.
	The hop-bound $\red{\beta}$ is the maximum number of hops in an induced
	path. Here $\red{\beta}=6$ (realized by $P_T^{a',e'}$).}
\end{figure}

We now translate \Cref{lem:clanHopUltraMeasure} to the language of
path-tree embeddings.
For simplicity, we state only the version needed for our applications:
using roughly a uniform measure (rather than a general measure $\mu$), we
obtain a unique copy of the root and
a bound on the total number of copies.
\begin{lemma}\label{lem:PathClanEmbedding}
	Given an $n$-point weighted graph $G = (V, E, w_G)$ with polynomial aspect
	ratio, a vertex $r\in V$,
	and integers $\red{h},k\ge1$, there is a dominating clan embedding
	$(f,\chi)$ into a
	tree $T$ such that:
	\begin{enumerate}
		\item $|f(V)|\le(2n)^{1+
			\frac1k}$ and $|f(r)|=1$.
		\item $(f,\chi)$ has hop-distortion $\left(O(k),\red{O(k\cdot\log
		n)},\red{h}\right)$, and hop-path-distortion $\left(\redno{O(k\cdot
		\log n)},\red{h}\right)$.
		\item $f$ is a path tree one-to-many embedding with associated paths
		$\{P^T_{e}\}_{e\in E(T)}$ and  hop-bound $\red{O(k\cdot\log^2
		n)\cdot h}$.
	\end{enumerate}
\end{lemma}
\begin{proof}
	First apply \Cref{lem:clanHopUltraMeasure} on the graph $G=(V,E,w)$
	with parameter $\red{h}$, $k$, and $(\ge1)$-measure $\mu$ where
	$\mu(r)=n$, and for every $v\ne r$, $\mu(v)=1$.
	As a result, we will obtain a dominating clan embedding $(f_U,\chi_U)$ into
	an ultrametric $U$, where
	(1) $|f_U(V)|\le\mu(V)^{1+\frac1k}<(2n)^{1+
		\frac1k}$,
	(2) $\forall u,v\in V$, $d_G^{\red{(O(k\cdot\log n)\cdot
	h)}}(u,v)\le\min_{v'\in f_U(v),u'\in f_U(u)}d_U(v',u')$,
	(3) $\forall u,v\in V$, $\min_{v'\in f_U(v)}d_U(v',\chi(u))\le O(k)\cdot
	d_G\rhop{h}(u,v)$,
	(4) $f$ has hop-path distortion $\left(O(k\cdot\log n),\red{h}\right)$,
	and finally, (5) as $\mu(r)>\frac12\mu(V)$, \Cref{obs:singleCopy} implies
	that $|f_U(r)|=1$.

	Next, we transform the ultrametric into a tree.
	Gupta \cite{Gupta01}, \footnote{See \cite{FKT19} for an alternative proof.
	We refer to \cite{Fil19SPR,Fil24Scattering} for further details on the
	Steiner point removal problem.}
	showed that given a tree $\mathcal{T}$ and a subset of terminals
	$K\subseteq V(\mathcal{T})$, one can (efficiently) embed (in a classic
	one-to-one fashion) the set of terminals $K$ into a tree $\mathcal{T}_K$
	such that $V(\mathcal{T}_K)=K$, and for every $u,v\in K$,
	$d_{\mathcal{T}}(u,v)\le d_{\mathcal{T}_K}(u,v)\le 8\cdot
	d_{\mathcal{T}}(u,v)$.
	That is, Gupta removed all the non-terminal (Steiner) points, and
	distorted all the terminal pairwise distances by a multiplicative factor
	of at most $8$.
	Applying \cite{Gupta01} on the ultrametric $U$ with $f_U(V)$ as the set of
	terminals $K$, and concatenating the two embeddings, we obtain a clan
	embedding $(f,\chi)$ into a tree $\mathcal{T}_K$ over $f_U(V)$ such that
	$\forall u,v\in V$,
	\begin{align*}
	d_{G}^{\red{(O(k\cdot\log n)\cdot h)}}(u,v) & \le\min_{v'\in f(v),u'\in
	f(u)}d_{\mathcal{T}_K}(v',u')\\
	\min_{v'\in f(v)}d_{\mathcal{T}_K}(v',\chi(u)) & \le O(k)\cdot
	d_{G}\rhop{h}(u,v)~.
	\end{align*}
	For every sequence of vertices $v'_0,\dots,v'_s$ in $f_U(V)$, it holds
	that $\sum_{i=0}^{s-1}d_{\mathcal{T}_K}(v'_i,v'_{i+1})\le 8\cdot
	\sum_{i=0}^{s-1}d_U(v'_i,v'_{i+1})$.
	As $f_U$ has hop-path distortion $\left(O(k\cdot\log n),\red{h}\right)$,
	it follows that $f$ has hop-path distortion $\left(O(k\cdot\log
	n),\red{h}\right)$ as well.

	Finally, for every edge $e=\{u',v'\}\in E(T)$ where $u'\in f(u)$ and
	$v'\in f(v)$, let $P^{\mathcal{T}_K}_{e}$ be a path of minimum weight
	among all the $u$-$v$ paths in $G$ with at most $\red{O(k\cdot\log n)\cdot
	h}$ hops. Note that $w_G(P_e^{\mathcal{T}_K})\le d_{G}^{\red{(O(k\cdot\log
	n)\cdot h)}}(u,v)\le d_{\mathcal{T}_K}(u',v')$.
	Following \Cref{alg:hier-Clan}, as $G$ has polynomial aspect ratio, the
	ultrametric $U$ has depth $O(\log n)$. In particular, every path in $U$
	consists of at most $\red{O(\log n)}$ hops.
	The tree $\mathcal{T}_K$ constructed by \cite{Gupta01} (see also
	\cite{FKT19}) is a minor of $U$ constructed by contracting edges only.
	Therefore, the hop-diameter of $\mathcal{T}_K$ is bounded by $\red{O(\log
	n)}$ as well.
	Consider $u',v'\in f(V)$, then $P^{\mathcal{T}_K}_{u'v'}$ consists of
	concatenating at most $O(\log n)$ paths, each with  $\red{O(k\cdot\log
	n)\cdot h}$ hops. It follows that $\hop(P^{\mathcal{T}_K}_{u'v'})\le
	\red{O(k\cdot\log^2n)\cdot h}$ as required.
\end{proof}

Recall that a subgraph $H$ of $G$ is $\red{h}$-respecting if for every $u,v\in
H$, it holds that $d\rhop{h}_G(u,v)\le d_H(u,v)$.
The following is a corollary of \Cref{lem:PathClanEmbedding}.
\begin{corollary}\label{cor:hBoundedSubgraph}
Consider an $n$-point graph $G=(V,E,w)$ with polynomial aspect ratio,
parameters $\red{h},k\in\N$ and vertex $r\in V$.
Then there is a path-tree one-to-many embedding $f$ of $G$ into a tree $T$
with hop bound $\red{O(k\cdot\log^2n)\cdot h}$, $|f(V)|=(2n)^{1+\frac1k}$,
$|f(r)|=1$, and such that for every connected $\red{h}$-respecting sub-graph
$H$, there is a connected subgraph $H'$ of $T$ where for every $u\in H$, $H'$
contains some vertex from $f(u)$. Further, $w_T(H')\le O(k\cdot\log
n)\cdot w(H)$.
\end{corollary}
\begin{proof}
	We apply \Cref{lem:PathClanEmbedding} on $G$. All the properties other
	than the last one hold immediately. Consider an $\red{h}$-respecting
	subgraph $H$ of $G$.
	We double each edge in $H$, to obtain a multi-graph, where all the
	vertices have even degrees. In particular, there is an Eulerian tour
	$P=(x_0,\dots,x_s)$ traversing all the edges in $H$. Note that for every
	$a,b\in[0,s]$, as $H$ is $\red{h}$-respecting subgraph of $G$, it holds
	that $$d_{P}(x_a,x_b)\ge d_{H}(x_a,x_b)\ge d\rhop{h}_{G}(x_a,x_b)~.$$
	Thus $P$ is an $\red{h}$-respecting path.
	It follows that there are vertices $x'_0,\dots,x'_s$ where $x'_i\in
	f(x_i)$ such that $\sum_{i=0}^{s-1}d_T(x'_i,x'_{i+1})\le O(k\cdot\log
	n)\cdot \sum_{i=0}^{s-1}w(x_i,x_{i+1})=O(k\cdot\log n)\cdot 2\cdot w(H)$.
	The corollary follows.
\end{proof}

Finally, we will state the result in this section in an alternative form,
where instead of \Cref{lem:clanHopUltraMeasure}, we will use
\Cref{lem:clanHopUltraMeasureAlt}.
In the counter-part of \Cref{lem:PathClanEmbedding} we will obtain the same
result with the only changes that
$(f,\chi)$ will have hop-distortion $\left(O(k\cdot\log\log
n),\red{O(k\cdot\log\log n)},\red{h}\right)$,   hop-path-distortion
$\left(\redno{O(k\cdot\log n\cdot\log\log n)},\red{h}\right)$, and the hop
bound of the path tree one-to-many embedding will be $\red{O(k\cdot\log
n\cdot\log \log n)\cdot h}$.
The following corollary follows the same lines:
\begin{corollary}\label{cor:hBoundedSubgraphAlt}
	Consider an $n$-point graph $G=(V,E,w)$ with polynomial aspect ratio,
	parameters $\red{h},k\in\N$ and vertex $r\in V$.
	Then there is a path-tree one-to-many embedding $f$ of $G$ into a tree $T$
	with hop bound $\red{O(k\cdot\log n\cdot\log \log n)\cdot h}$,
	$|f(V)|=(2n)^{1+\frac1k}$, $|f(r)|=1$, and such that for every connected
	$\red{h}$-respecting sub-graph $H$, there is a connected subgraph $H'$ of
	$T$ where for every $u\in H$, $H'$ contains some vertex from $f(u)$.
	Further, $w_T(H')\le O(k\cdot\log n\cdot\log \log n)\cdot w(H)$.
\end{corollary}

\subsection{General subgraph: Proofs of \Cref{thm:Subgraph,thm:SubgraphAlt}}
Both \Cref{cor:hBoundedSubgraph} and \Cref{cor:hBoundedSubgraphAlt} have a
major drawback: the guarantee is only w.r.t. $\red{h}$-respecting sub-graphs.
Unfortunately, the $\red{h}$-respecting condition is indeed necessary (as the
image tree $T$ should not allow us to connect vertices at large hop-distance).
Here we ensure the next best thing: given an arbitrary subgraph $H$ of $G$, we
will construct a (not necessarily connected) subgraph $H'$ of $T$ where for
every $u,v\in H$ such that $\hop_H(u,v)\le\red{h}$, $H'$ contains a pair of
copies $u'\in f(u)$ and $v'\in f(v)$ in the same connected component. Further,
$w_T(H')\le O(k\cdot\log n)\cdot w(H)$.
As a result, the hop bound will increase by a $\log$ factor.

\SubgraphPreservingEmbedding*
\begin{proof}
	Set $\red{h'}=\red{h\cdot 4\log n}$. We execute
	\Cref{cor:hBoundedSubgraph} with hop parameter $\red{h'}$. As a result we
	obtain a path-tree one-to-many embedding $f$ into a tree $T$ with hop
	bound $\red{O(k\cdot\log^2n)\cdot h'}=\red{O(k\cdot\log^3n)\cdot h}$,
	$|f(V)|=O((2n)^{1+1/k})$, and  $|f(r)|=1$. It only remains to prove the
	subgraph preservation property, which will be based on the
	following lemma.
	\begin{lemma}[Sparse Cover]\label{lem:SparseCover}
		Consider an $n$-vertex weighted graph $G=(V,E,w)$, with cost function
		$\mu:E\rightarrow\R_{\ge0}$, and parameter $\Delta>0$. Then there is a
		collection of clusters $C_{1},C_{2},\dots,C_{s}$ such that:
		\begin{enumerate}
			\item Each $C_{i}$ has radius at most $\Delta\cdot \log (2n)$.
			That is, there is a vertex $v_i\in C_i$ such that $\max_{u\in
			C_i}d_{G[C_i]}(u,v_i)\le \Delta\cdot  \log (2n)$.
			\label{property:diameter}
			\item Every pair of vertices $u,v$ such that $d_G(u,v)\le \Delta$
			belongs to some $C_{i}$. \label{property:cover}
			\item $\sum_{i=1}^{s} \mu(C_{i})\le4\cdot \mu(G)$, where
			$\mu(C)=\sum_{e\subseteq C}\mu(e)$. \label{property:sparsity}
		\end{enumerate}
	\end{lemma}
	We will now proceed with the proof of \Cref{thm:Subgraph}. The proof (and
	some historical context) of \Cref{lem:SparseCover} is deferred to
	\Cref{sec:SparseCover}.
	Consider a subgraph $H=(V_H,E_H,w)$ of $G$. Let $H_1=(V_H,E_H)$ be the
	graph $H$ where all the edge weights are changed to $1$. Set
	$\Delta=\red{h}$, and a cost function $\mu:E_H\rightarrow\R_{\ge 0}$,
	where $\mu(e)=w(e)$.
	Using \Cref{lem:SparseCover}, we obtain a collection of clusters
	$C_1,\dots,C_s\subseteq V_H$. Let $F_i$ be the shortest path tree rooted
	in $v_i$ spanning all the vertices in $C_i$ w.r.t. distances in $H_1$.
	Note that as $H_1$ is unweighted, by \propref{property:diameter} of
	\Cref{lem:SparseCover},
	$F_i$ is a tree of depth at most $\Delta\cdot \log 2n$. In particular,
	every path in $F_i$ contains at most $2\cdot \Delta\cdot \log
	2n\le\red{h}\cdot 4\log n=\red{h'}$ hops. It follows that $F_i$ (when
	considered as a subgraph of $G$ with the original weight function $w$) is
	$\red{h'}$-respecting subgraph of $G$ (this is as every $u$-$v$ path in
	$F_i$ consists of at most $\red{h'}$ hops, and thus has weight at least
	$d_G\rhop{h'}(u,v)$ ).

	By \Cref{cor:hBoundedSubgraph}, there is a connected subgraph $T_i$ of $T$
	where for every $u\in F_i$, $T_i$ contains some vertex from $f(u)$.
	Further, $w_T(T_i)\le O(k\cdot\log n)\cdot w(F_i)$. Set
	$H'=\cup_{i=1}^{s}T_i$.
	Note that for every $u,v\in H$ such that $\hop_H(u,v)\le \red{h}$, it
	holds that $d_{H_1}(u,v)\le \red{h}$, and hence by \Cref{lem:SparseCover}
	there is some cluster $C_i$ containing both $u,v$. According to
	\Cref{cor:hBoundedSubgraph}, $T_i\subseteq H'$ is a connected subgraph
	containing vertices from both $f(u)$ and $f(v)$, as required. Finally, we
	bound the weight of $H'$:
	\begin{align*}
		w_{T}(H') & \le\sum_{i=1}^{s}w_{T}(T_{i})=O(k\cdot\log
		n)\cdot\sum_{i=1}^{s}w(F_{i})\le O(k\cdot\log
		n)\cdot\sum_{i=1}^{s}w(C_{i})\\
		& =O(k\cdot\log n)\cdot\sum_{i=1}^{s}\mu(C_{i})=O(k\cdot\log
		n)\cdot\mu(H)=O(k\cdot\log n)\cdot w(H)~.
	\end{align*}
\end{proof}

By using \Cref{cor:hBoundedSubgraphAlt} instead of \Cref{cor:hBoundedSubgraph}
in the proof above, we obtain the following:
\begin{theorem}\label{thm:SubgraphAlt}
	Consider an $n$-point graph $G=(V,E,w)$ with polynomial aspect ratio,
	parameters $\red{h},k\in\N$ and vertex $r\in V$.
	Then there is a path-tree one-to-many embedding $f$ of $G$ into a tree $T$
	with hop bound $\red{O(k\cdot\log^2 n\cdot\log\log n)\cdot h}$,
	$|f(V)|=(2n)^{1+\frac1k}$, $|f(r)|=1$, and such that for every sub-graph
	$H$, there is a subgraph $H'$ of $T$ where for every $u,v\in H$ with
	$\hop_H(u,v)\le\red{h}$, $H'$ contains a pair of copies $u'\in f(u)$ and
	$v'\in f(v)$ in the same connected component. Further, $w_T(H')\le
	O(k\cdot\log n\cdot\log\log n)\cdot w(H)$.
\end{theorem}

\subsection{Proof of \Cref{lem:SparseCover}: sparse
cover}\label{sec:SparseCover}
Sparse covers were first introduced by Awerbuch and Peleg \cite{AP90}, and
were studied in various graph families (see e.g.
\cite{KPR93,Fil19padded,Fil24Scattering,Fil25}). We refer to
\cite{Fil19padded,Fil24Scattering} for further details.
For an integer parameter $k\ge 1$, and $\Delta>0$, Awerbuch and Peleg
partitioned a general $n$-vertex graph into clusters of radius at most
$(2k-1)\cdot\Delta$ such that every ball of radius $\Delta$ is fully contained
in some cluster, and every vertex belongs to at most $2k\cdot n^{\frac1k}$
different clusters. In particular, by fixing $k=\log n$, one will get all the
properties in \Cref{lem:SparseCover}, but the bound on $\sum_i\mu(C_i)$ will
only be $4\log n\cdot\mu(G)$ (we aim for $4\cdot\mu(G)$).
While the construction in \cite{AP90} is deterministic, here we will present a
random ball growing clustering algorithm (ala \cite{Bar96}). As a result,
while in the worst case a vertex might belong to $\log n$ different clusters,
in expectation, each vertex will belong only to a constant number. A bound on
the total weight will follow.
\vspace{7pt}

	\begin{wrapfigure}{r}{0.33\textwidth}
	\begin{center}
		\vspace{-40pt}
		\includegraphics[width=0.83\textwidth]{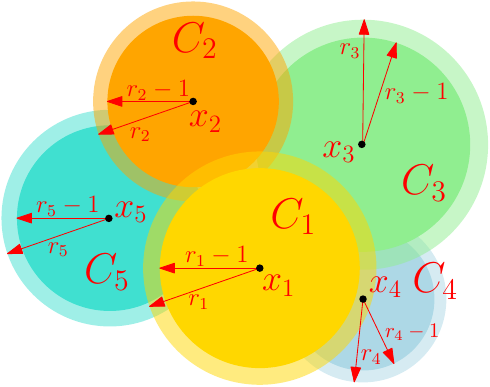}
	\end{center}
	\vspace{-25pt}
\end{wrapfigure}

\emph{\hspace{-18pt}Proof of \Cref{lem:SparseCover}.}
	By scaling, we can assume that $\Delta=1$. Let $\Geo(\frac12)$ be
	a geometric distribution with parameter $\frac12$ (that is, taking value
	$i\ge1$ with probability $2^{-i}$).
	Each created cluster $C$ will have an interior part denoted $\inter(C)$,
	that will be contained in $C$.
	We create a collection of clusters $\mathcal{C}$ as follows. Initially
	$Y_1=V$ is all the unclustered vertices. At step $i$, after we created the
	clusters $C_1,C_2,\dots,C_{i-1}$, the unclustered vertices are
	$Y_i=V\setminus(\cup_{j<i}\inter(C_j))$. Pick an arbitrary vertex $x_i\in
	Y_i$, and a radius $r_i\sim\Geo(\frac12)$.  Set the new cluster to be
	$C_{i}=B_{G[Y_i]}(x_i,r_i)$, the ball of radius $r_i$ around $x_i$ in the
	graph induced by the unclustered vertices, and its interior to be
	$\inter(C_i)=B_{G[Y_i]}(x_i,r_i-1)$. The process halts once
	$Y_i=\emptyset$.
	See the illustration on the right.

	Consider a pair $u,v$ such that $d_G(u,v)\le 1$, and let $P$ be the
	shortest path between them. Let $i$ be the first index such that
	$\inter(C_i)\cap P\ne\emptyset$, and let $z\in \inter(C_i)\cap P$. By the
	minimality of $i$, $P\subseteq Y_i$.
	As $z\in\inter(C_i)$ it follows that $d_{G[Y_i]}(x_i,z)\le r_i-1$. By
	triangle inequality $d_{G[Y_i]}(x_i,u)\le
	d_{G[Y_i]}(x_i,z)+d_{G[Y_i]}(z,u)\le r_i$, implying that $u\in C_i$.
	Similarly, $v\in C_i$. \Propref{property:cover} of
	\Cref{lem:SparseCover} follows.

	The radius of the cluster $C_i$ is bounded by $r_i$. Denoted by $\Psi$ the
	event that $r_i\le \log(2 n)$ for every $i$.
	By a union bound, as in the entire process we create at most $n$ balls,
	\[
	\Pr\left[\overline{\Psi}\right]=\Pr\left[\exists i\text{ s.t.
	}r_{i}>\log(2 n)\right]\le n\cdot2^{-\log(2 n)}=\frac{1}{2}~.
	\]

	Denote by $X_v$ the number of clusters containing a vertex $v$. Note that
	given that $v$ belongs to a cluster $C_i$, by the memoryless property of
	the geometric distribution, the probability that $v$ belongs to the
	interior of the cluster is $\frac12$:
	\[
	\Pr\left[v\in\inter(C_{i})\mid v\in
	C_{i}\right]=\Pr\left[d_{G}(x_{i},v)\le r_{i}-1\mid d_{G}(x_{i},v)\le
	r_{i}\right]=\frac{1}{2}~.
	\]
	Once $v$ belongs to the interior of a cluster $\inter(C_i)$, it will not
	belong to any other clusters. Hence $X_v$ is dominated by a geometric
	distribution with parameter $\frac12$. In particular,
	$\mathbb{E}[X_v]\le2$.
	For an edge $e=\{u,v\}$ denote by $X_e$ the number of clusters containing
	$e$. Then clearly $\mathbb{E}[X_e]\le\mathbb{E}[X_v]\le 2$. We conclude:
	\[
	\mathbb{E}\left[\sum_{i\ge1}\mu(C_{i})\right]=\sum_{e\in
	E}\mu(e)\cdot\mathbb{E}\left[X_{e}\right]\le 2\cdot\mu(G)~.
	\]
	We condition the execution of the algorithm on the event $\Psi$. Then
	clearly \propref{property:diameter} holds. Using the law of total
	probability, the expected cost is bounded by
	\[
	\mathbb{E}\left[\sum_{i\ge1}\mu(C_{i})\mid\Psi\right]=\frac{1}{\Pr\left[
	\Psi\right]}\cdot\left(\mathbb{E}\left[\sum_{i\ge1}\mu(C_{i})\right]-\Pr
	\left[\overline{\Psi}\right]\cdot\mathbb{E}\left[\sum_{i\ge1}\mu(C_{i})
	\mid\overline{\Psi}\right]\right)\le\nicefrac{2}{\frac{1}{2}}=4~.
	\]
	\Propref{property:sparsity} now follows as well.
	\QED

\section{Approximation algorithms}\label{sec:approxAlgs}
\subsection{Application of the Ramsey type embedding}
Our \Cref{thm:UltrametricRamsey,thm:UltrametricRamseyAlt} is similar in
spirit to a theorem in Haeupler \etal \cite{HHZ21}:
\begin{theorem}[\cite{HHZ21}]\label{thm:HHZ21}
	Consider an $n$-vertex graph $G=(V,E,w)$, and parameters $\red{h}\in [n]$,
	$\eps\in(0,\frac13)$. There is a distribution $\mathcal{D}$ over pairs
	$(M,U)$ where $M\subseteq V$, and $U$ is an ultrametric over
	$M$ such that:
	\begin{enumerate}
		\item For every $v\in V$, $\Pr[v\in M]\ge 1-\eps$.
		\item For every $v,u\in M$, $d_G\rhop{O(\log^2 n)\cdot h}(u,v)\le
		d_U(u,v)\le O(\frac{\log^2n}{\eps})\cdot d_G\rhop{h}(u,v)$.
	\end{enumerate}
\end{theorem}
Note that our \Cref{thm:UltrametricRamsey} implies \cite{HHZ21}
\Cref{thm:HHZ21} with an improved distortion parameter of $O(\frac{\log
n}{\eps})$, while our \Cref{thm:UltrametricRamseyAlt} implies \cite{HHZ21}
\Cref{thm:HHZ21} with an improved hop-stretch guarantee of $\red{O(\frac{\log
n\cdot\log\log n}{\eps})}$ and improved distortion parameter of $O(\frac{\log
n\cdot\log\log n}{\eps})$.
\cite{HHZ21} applied \Cref{thm:HHZ21} to obtain many different
approximation algorithms.
While their application is not straightforward,
one can use \Cref{thm:UltrametricRamsey} or \Cref{thm:UltrametricRamseyAlt}
instead of \Cref{thm:HHZ21}, to obtain improved approximation factors in all
the approximation algorithms in \cite{HHZ21}. We refer to \cite{HHZ21} for
problem definitions, background, approximation algorithms and their analysis,
and only summarize the results in the following corollary. See
\Cref{tab:ApproxAlgs} for a comparison between our results and previous results.
\begin{corollary}\label{cor:HHZ21Application}
	The hop-constrained version of each of the following problems has
	a poly-time approximation algorithm as follows:
	\begin{enumerate}
		\item The oblivious Steiner forest with $O(\log^2n)$-cost
		approximation and $O(\log^3n)$ hop approximation.
		Alternatively, one can obtain $\tilde{O}(\log^2n)$-cost approximation
		and $\tilde{O}(\log^2n)$ hop approximation.
		\item The relaxed k-Steiner tree with $O(\log n)$-cost approximation
		and $O(\log^3n)$ hop approximation.
		Alternatively, one can obtain $\tilde{O}(\log n)$-cost approximation
		and $\tilde{O}(\log^2n)$ hop approximation.
		\item The k-Steiner tree with $O(\log n\cdot\log k)$-cost
		approximation and $O(\log^3n)$ hop approximation.
		Alternatively, one can obtain $\tilde{O}(\log n\cdot\log k)$-cost
		approximation and $\tilde{O}(\log^2n)$ hop approximation.
		\item The oblivious network design with $O(\log^3 n)$-cost
		approximation and $O(\log^3n)$ hop approximation.
		Alternatively, one can obtain $\tilde{O}(\log^3n)$-cost approximation
		and $\tilde{O}(\log^2n)$ hop approximation.
		\item The group Steiner tree with $O(\log^3 n\cdot\log k)$-cost
		approximation and $O(\log^3n)$ hop approximation.
		Alternatively, one can obtain $\tilde{O}(\log^3n\cdot\log k)$-cost
		approximation and $\tilde{O}(\log^2n)$ hop approximation.
		\item The online group Steiner tree with $O(\log^4 n\cdot\log
		\kappa)$-cost approximation and $O(\log^3n)$ hop approximation.
		Alternatively, one can obtain $\tilde{O}(\log^4n\cdot\log
		\kappa)$-cost approximation and $\tilde{O}(\log^2n)$ hop
		approximation.
		\item The group Steiner forest with $O(\log^5 n\cdot\log k)$-cost
		approximation and $O(\log^3n)$ hop approximation.
		Alternatively, one can obtain $\tilde{O}(\log^5n\cdot\log k)$-cost
		approximation and $\tilde{O}(\log^2n)$ hop approximation.
		\item The online group Steiner forest with $O(\log^6 n\cdot\log
		\kappa)$-cost approximation and $O(\log^3n)$ hop approximation.
		Alternatively, one can obtain $\tilde{O}(\log^6n\cdot\log
		\kappa)$-cost approximation and $\tilde{O}(\log^2n)$ hop
		approximation.
	\end{enumerate}
\end{corollary}

Most notably, the state of the art bicriteria approximation for the h.c.
$k$-Steiner tree problem$^{\ref{foot:kSteinerTree}}$ was by Khani and
Salavatipour \cite{KS16}, who obtain $O(\log^2n)$ cost approximation and
$O(\log n)$ hop-stretch (improving over \cite{HKS09}).
Note that our $O(\log n\cdot\log k)$ cost approximation is strictly superior,
while our hop-stretch is inferior.

\subsection{Application of the clan embedding}
For the h.c. relaxed k-Steiner tree, h.c. k-Steiner tree, h.c. oblivious
Steiner forest, and h.c. oblivious network design problems
\Cref{cor:HHZ21Application} is the best approximation we obtain (by simply
replacing \Cref{thm:HHZ21} with
\Cref{thm:UltrametricRamsey,thm:UltrametricRamseyAlt}).
For the h.c. group Steiner tree/forest problems, and their online versions, we
provide better approximation factors by going through clan embeddings, and in
particular the subgraph preservation properties of
\Cref{cor:hBoundedSubgraph,cor:hBoundedSubgraphAlt} and
\Cref{thm:Subgraph,thm:SubgraphAlt}.

\cite{HHZ21} had an analogous theorem to \Cref{thm:Subgraph}, while
the cost was
$O(\log^3n)$ instead of our $O(\log n)$ (see Theorem 8 in their
\href{https://arxiv.org/abs/2011.06112}{arXiv version}). Using our
\Cref{thm:Subgraph,thm:SubgraphAlt} provides a $O(\log^2n)$ multiplicative
improvement for the h.c. group Steiner forest problem and its online version.
To illustrate the strength of \Cref{thm:Subgraph,thm:SubgraphAlt}, we provide
full details in \Cref{subsec:GroupForest,subsec:OnlineGroupForest}.
Further, for the h.c. group Steiner tree problem and its online version the
optimal solution is an $\red{h}$-respecting subgraph, and hence we can use
\Cref{cor:hBoundedSubgraph,cor:hBoundedSubgraphAlt} instead of
\Cref{thm:Subgraph,thm:SubgraphAlt}, thus shaving an $O(\log n)$ factor
from the hop stretch, see \Cref{subsec:GroupTree,subsec:OnlineGroupTree}
for details.
See \Cref{tab:ApproxAlgs} for a comparison between our results and previous results.

\subsubsection{Hop-constrained group Steiner tree}\label{subsec:GroupTree}
In the group Steiner tree problem we are given $k$ sets
$g_1,\dots,g_k\subseteq V$, and a root vertex $r\in V$. The goal is to
construct a minimum-weight tree spanning the vertex $r$ and at least one
vertex from each group (set) $g_i$.
In general, the state of the art is an $O(\log^2n\cdot \log k)$-approximation
\cite{GKR00}.
In the h.c. version of the problem, we are given in addition a hop-bound
$\red{h}$, and the requirement is to find a minimum weight subgraph $H$, such
that for every $g_i$, $\hop_H(r,g_i)\le \red{h}$.
Denote by $\opt$ the weight of the optimal solution.
\cite{HHZ21} provided a bicriteria approximation by constructing a subgraph
$H$ of weight $O(\log^4n\cdot \log k)\cdot\opt$, such that for every $i$,
$\hop_H(r,g_i)\le\red{O(\log^3n)\cdot h}$.
We restate our improvement:

\GroupSteinerTree*
Note that ignoring the hop constraints, our approximation factor matches
the state of the art \cite{GKR00}.
\begin{proof}[Proof of \Cref{thm:GroupSteinerTree}]
We apply \Cref{cor:hBoundedSubgraph} with hop parameter $\red{h'}=\red{2h}$,
stretch parameter $2$, and $r$ as the special vertex, to obtain a path-tree
one-to-many embedding $f$ into a tree $T$, with associated paths
$\{P^T_e\}_{e\in E(T)}$. For every $i$, let $\tilde{g}_i=f(g_i)$, and set
$\tilde{r}$ to be the only vertex in $f(r)$.
Garg, Konjevod, and Ravi \cite{GKR00} constructed an $O(\log n\cdot\log k)$
approximation algorithm for the group Steiner tree problem in trees. We apply
their algorithm on the tree $T$ with center $\tilde{r}$ and the groups
$\tilde{g}_1,\dots,\tilde{g}_k$. Suppose that \cite{GKR00} returns the solution
$\tilde{H}\subseteq T$. Our solution to the original problem will be
$H=\bigcup_{e\in \tilde{H}}P^T_e$. \footnote{If it is desirable that the
obtained solution will be a tree, one can simply take a shortest path tree
rooted at $r$ w.r.t. unweighted $H$. This will be a valid solution of at most
the same weight.}

For every $i$, $\tilde{H}$  contains a path $P_i=e_1,\dots,e_s$ from
$\tilde{r}$ to some vertex $v'\in \tilde{g}_i=f(g_i)$, in particular $v'\in
f(v)$ for some $v\in g_i$. It follows that the solution $H$ we return
will contain a path $P^{T}_{\tilde{r},v'}$ from $r$ to $v$. As the
path-tree one-to-many embedding $f$ has hop-bound $\red{O(\log^2n)\cdot h'}$,
it follows that $\hop_{H}(r,g_i)\le \red{O(\log^2n)\cdot h}$ as required.
We conclude that $H$ is a valid solution to the problem.

Next we bound the weight of $H$.
Let $F$ be an optimal solution in $G$ of weight $\opt$. Note that we can
assume that $F$ is a tree, as otherwise we can take the shortest path tree
w.r.t. the unweighted version of the graph $F$ rooted at $r$. The result will
still be a valid solution of weight at most $\opt$.
Note that for every $u\in F$, $\hop_F(r,u)\le \red{h}$, as otherwise we can
remove this vertex from $F$. As $F$ is a tree, every pair of vertices $u,v\in
F$ have a unique path between them with at most $\red{2h}$ hops, in
particular $d_G\rhop{2h}(u,v)\le d_F(u,v)$.  We conclude that $F$ is
$\red{h'}$-respecting.
Hence by \Cref{cor:hBoundedSubgraph}, $T$ contains a connected graph
$\tilde{F}$, such that for every $u\in F$, $\tilde{F}$ contains some vertex in
$f(u)$, and $w_T(\tilde{F})=O(\log n)\cdot w_G(F)=O(\log n)\cdot\opt$.
Note that $\tilde{F}$ is a valid solution to the group Steiner tree problem in
$T$. Hence by \cite{GKR00},
\[
w_{G}(H)\le w_{T}(\tilde{H})\le O(\log n\cdot\log k)\cdot
w_{T}(\tilde{F})=O(\log^{2}n\cdot\log k)\cdot\opt~.
\]
as required.
The alternative tradeoff is obtained by using \Cref{cor:hBoundedSubgraphAlt}
instead of \Cref{cor:hBoundedSubgraph}.
\end{proof}

\subsubsection{Online hop-constrained group Steiner
tree}\label{subsec:OnlineGroupTree}
The online group Steiner tree problem is an online version of the group
Steiner tree problem. Here we are given a graph $G=(V,E,w)$ with a root $r\in
V$. The groups $g_1,g_2,\dots$ arrive one at a time. The goal is to maintain
trees $T_1,T_2,\dots$ such that for every $i$,  $T_{i-1}\subseteq T_i$, and
such that in $T_i$, $r$ is connected to at least one vertex in $g_i$.
An algorithm has \emph{competitive ratio} $t$ if for every sequence of sets
$g_1,\dots,g_i$, the tree $T_i$ has weight at most $t$ times larger than the
optimal solution for the offline group Steiner tree problem with input groups
$g_1,\dots,g_i$ (and root $r$).
For randomized algorithms, we say that the algorithm has competitive ratio
$\alpha$, if for every input sequence (fixed in advance, unknown by the
algorithm) the expected ratio between the obtained solution to the optimal one
is at most $\alpha$.

As was shown by Alon \etal \cite{AAABN06}, in full generality, no
sub-polynomial approximation is possible.
Therefore, we will assume that there is a set of subsets $\mathcal{K}\subseteq
2^{V}$ of size $\kappa=|\cK|$, and all the sets $g_i$ are chosen from
$\mathcal{K}$.
The parameter in the competitive ratio of \cite{AAABN06} (and here) is
$\kappa$, the number of possible subsets, rather than the number of
requests that actually appear in a particular input sequence.
For this version, \cite{AAABN06} obtain a randomized $O(\log^2 n\cdot\log
\kappa)$ competitive ratio for the case where the host graph $G$ is a tree,
implying an $O(\log^3 n\cdot\log \kappa)$ competitive ratio for
general graphs.

In the online h.c. group Steiner tree problem, there is an additional
requirement that the number of hops in $T_i$ between $r$ and a member of $g_i$ is
at most $\red{h}$, i.e. $\hop_{T_i}(r,g_i)\le\red{h}$.
Using their $\red{h}$-Hop Repetition Tree Embedding, \cite{HHZ21} obtained
hop-stretch $\red{O(\log^3 n)}$, and competitive ratio $O(\log^5
n\cdot\log \kappa)$.
That is, constructing a sequence of subgraphs $T_1\subseteq
T_2\subseteq\dots$, such that for every $i$,
$\hop_{T_i}(r,g_i)\le\red{O(\log^3 n)\cdot h}$, and the (expected) weight of
$T_i$ is at most $O(\log^5 n\cdot\log \kappa)\cdot \opt_i$, where $\opt_i$ is
the cost of the optimal solution for the offline $\red{h}$ h.c. group Steiner
tree problem for root $r$ and groups $g_1,\dots,g_i$.

\begin{theorem}\label{thm:OnlineGroupSteinerTree}
	There is a poly-time randomized algorithm for the online
	$\red{h}$-hop-constrained group Steiner tree problem on graphs with
	polynomial aspect ratio, that maintains a solution $\{T_i\}_{i\ge
	1}$ which is
	$O(\log^3n\cdot\log \kappa)$-cost competitive (i.e. $\mathbb{E}[w(T_i)]\le
	O(\log^3n\cdot\log \kappa)\cdot\opt_i$), and such that $\hop_{T_i}(r,g_i)\le
	\red{O(\log^2n)\cdot h}$ for every $i$.
	Alternatively, one can obtain a solution which is
	$\tilde{O}(\log^3n)\cdot\log \kappa$-cost competitive, and such that
	$\hop_{T_i}(r,g_i)\le \red{\tilde{O}(\log n)\cdot h}$ for every $i$.
\end{theorem}
Note that ignoring the hop constraints, our approximation factor matches
the state of the art \cite{AAABN06}. It is also interesting to note that the
only randomized component in our algorithm is the \cite{AAABN06} black box
solution for trees.
\begin{proof}[Proof of \Cref{thm:OnlineGroupSteinerTree}]
	Consider an instance of the online $\red{h}$-h.c. group Steiner
	tree problem. Initially, we are given a graph $G=(V,E,w_G)$, hop parameter
	$\red{h}$, root $r\in V$, and a set $\mathcal{K}\subseteq2^V$ of possible
	groups, where $\kappa=|\mathcal{K}|$.
	We apply \Cref{cor:hBoundedSubgraph} with hop parameter
	$\red{h'}=\red{2h}$, stretch parameter $2$, and $r$ as the special vertex.
	As a result, we obtain a path-tree one-to-many embedding $f$ into a tree
	$\mathcal{T}$, with associated paths $\{P^{\mathcal{T}}_e\}_{e\in
	E(\mathcal{T})}$.

	Next, we apply Alon {\em et al.}'s \cite{AAABN06} algorithm for the online
	group Steiner tree problem on trees with competitive ratio $O(\log^2
	n\cdot\log \kappa)$.
	Here the input graph is $\mathcal{T}$, the root vertex is the only vertex
	$\tilde{r}$ in $f(r)$, and the set of possible groups is
	$\tilde{\mathcal{K}}=\left\{f(K)\mid K\in\mathcal{K}\right\}$.
	We will iteratively receive as an input the sets $g_1,g_2,\dots\in
	\mathcal{K}$, and feed Alon {\em et al.}'s \cite{AAABN06} algorithm the
	sets $f(g_1),f(g_2),\dots\in \tilde{\mathcal{K}}$. As a result, we will
	obtain sub-trees $\tilde{T}_1,\tilde{T}_2,\dots$ of $\mathcal{T}$.
	For every $i$, we will return the solution
	$T_i=\cup_{e\in\tilde{T}_i}P^{\mathcal{T}}_e$.

	We first argue that our solution is valid. As
	$\tilde{T}_1\subseteq\tilde{T}_2\subseteq\dots$, it holds that
	$T_1\subseteq T_2\subseteq\dots$.
	Next, $\tilde{T}_i$ contains a path $P_i=e_1,\dots,e_s$ from $\tilde{r}$
	to some vertex $v'\in f(g_i)$, in particular $v'\in f(v)$ for some $v\in
	g_i$. It follows that $T_i$ will contain a path
	$P^{\mathcal{T}}_{\tilde{r},v'}$ from $r$ to $v$. As the path-tree
	one-to-many embedding $f$ has hop-bound $\red{O(\log^2n)\cdot h'}$, it
	follows that $\hop_{T_i}(r,g_i)\le \red{O(\log^2n)\cdot h}$ as required.

	Next, we bound the expected weight of the returned solution.
	Let $H_i$ be the optimal solution for the (offline) $\red{h}$-h.c. group
	Steiner tree problem in the graph $G$ with root $r$ and groups
	$g_1,\dots,g_i$, of cost $\opt_i=w_G(H_i)$.
	Similarly to the proof of \Cref{thm:OnlineGroupSteinerTree},
	$H_i$ is a tree.
	Note that for every $u\in H_i$, $\hop_{H_i}(r,u)\le \red{h}$, as otherwise
	we can remove this vertex from $H_i$. As $H_i$ is a tree, every pair of
	vertices $u,v\in H_i$ have a unique path between them with at most
	$\red{2h}$ hops, in particular $d_G\rhop{2h}(u,v)\le d_{H_i}(u,v)$.  We
	conclude that $H_i$ is $\red{h'}$-respecting.
	Hence by \Cref{cor:hBoundedSubgraph}, $T$ contains a connected graph
	$\tilde{H}_i$, such that for every $u\in H_i$, $\tilde{H}_i$ contains some
	vertex in $f(u)$, and $w_T(\tilde{H}_i)=O(\log n)\cdot w_G(H_i)=O(\log
	n)\cdot\opt_i$.
	Hence $\tilde{H}_i$ is a valid solution to the group Steiner tree problem
	in $T$, and thus $\tilde{\opt}_i=O(\log n)\cdot\opt_i$, where
	$\tilde{\opt}_i$ denotes the cost of the optimal solution in
	$\tilde{H}_i$.
	By \cite{AAABN06}, it holds that
	\[
	\mathbb{E}\left[w_{G}(T_{i})\right]\le\mathbb{E}\left[w_{T}(\tilde{T}_{i})
	\right]\le O(\log^{2}n\cdot\log \kappa)\cdot\tilde{\opt}_{i}=O(\log^{3}n
	\cdot\log \kappa)\cdot\opt_{i}~.
	\]

	The alternative tradeoff is obtained by using
	\Cref{cor:hBoundedSubgraphAlt} instead of \Cref{cor:hBoundedSubgraph}.
\end{proof}

\subsubsection{Hop-constrained group Steiner forest}\label{subsec:GroupForest}
In the group Steiner forest problem we are given $k$ subset pairs
$(S_1,R_1),\dots,(S_k,R_k)\subseteq V$, where for every $i\in[k]$,
$S_i,R_i\subseteq V$. The goal is to construct a minimum-weight forest $F$,
such that for every $i$, there are vertices $s_i\in S_i$ and $r_i\in R_i$,
belonging to the same connected component of $F$.
The group Steiner forest problem is sometimes called the ``generalized
connectivity problem'' as it generalizes both group Steiner tree, and Steiner
forest problems.

The group Steiner forest problem was introduced by Alon \etal \cite{AAABN06}.
The state of the art is by Chekuri \etal \cite{CEGS11} who obtained an
$O(\log^2n\log^2k)$ approximation factor.

In the hop-constrained version of the problem (introduced by \cite{HHZ21}), we
are given in addition a hop-bound $\red{h}$, and the requirement is to find a
minimum weight subgraph $H$, such that for every $i$,
$\hop_H(S_i,R_i)\le \red{h}$.
Denote by $\opt$ the weight of the optimal solution.
\cite{HHZ21} obtained a bicriteria approximation by constructing a subgraph
$H$ of weight $O(\log^6n\cdot \log k)\cdot\opt$, such that for every $i$,
$\hop_H(S_i,R_i)\le\red{O(\log^3n)\cdot h}$.
We obtain the following:
\GroupSteinerForest*
\begin{proof}
	We apply \Cref{thm:Subgraph} (with hop parameter $\red{h}$, and stretch
	parameter $2$) to obtain a path-tree one-to-many embedding $f$ into a tree
	$T$, with associated paths $\{P^T_e\}_{e\in E(T)}$. For every $i$, let
	$\tilde{S}_i=f(S_i)$, and $\tilde{R}_i=f(R_i)$.
	Naor, Panigrahi, and Singh \cite{NPS11} constructed an $O(d\cdot\log^2
	n\cdot\log k)$ approximation algorithm for the group Steiner forest
	problem in trees of depth $d$. \footnote{The paper \cite{NPS11} deals only
	with the online version of the group Steiner forest problem. From private
	communication with Hershkowitz \cite{Hershkowitz22}, the $O(d\cdot\log^2
	n\cdot\log k)$ approximation factor follows from \cite{NPS11} implicitly.}
	Following the discussion in the proof of \Cref{lem:PathClanEmbedding} the
	tree returned by \Cref{cor:hBoundedSubgraph} (and therefore also
	\Cref{thm:Subgraph}) has depth $O(\log n)$.
	We apply \cite{NPS11} algorithm on the tree $T$ with the group pairs
	$(\tilde{S}_1,\tilde{R}_1),\dots,(\tilde{S}_k,\tilde{R}_k)$. Suppose that
	\cite{NPS11} returns the solution $\tilde{H}\subseteq T$. Our solution to
	the original problem will be $H=\bigcup_{e\in \tilde{H}}P^T_e$.

	For every $i$, $\tilde{H}$ contains a path $P_i=e_1,\dots,e_s$ from a
	vertex $s'\in\tilde{S}_i=f(S_i)$ to a vertex $r'\in \tilde{R}_i=f(R_i)$.
	It follows that the solution $H$ we return will contain a path
	$P^{T}_{s',r'}$ from $f^{-1}(s')\in S_i$ to $f^{-1}(r')\in R_i$. As the
	path-tree one-to-many embedding $f$ has hop-bound $\red{O(\log^3n)\cdot
	h}$, it follows that $\hop_{H}(S_i,R_i)\le \red{O(\log^3n)\cdot h}$
	as required.
	We conclude that $H$ is a valid solution to the problem.

	Next we bound the weight of $H$. Let $F$ be an optimal solution in $G$ of
	weight $\opt$.
	By \Cref{thm:Subgraph}, $T$ contains a subgraph $\tilde{F}$, such that for
	every $u,v\in F$ with $d_F(u,v)\le\red{h}$, $\tilde{F}$ contains two
	vertices from $f(u)$ and $f(v)$ in the same connected component, and
	$w_T(\tilde{F})=O(\log n)\cdot w_G(F)=O(\log n)\cdot\opt$.
	Note that $\tilde{F}$ is a valid solution to the h.c. group Steiner forest
	problem in $T$. Hence by \cite{NPS11},
	\[
	w_{G}(H)\le w_{T}(\tilde{H})\le O(\log^3 n\cdot\log k)\cdot
	w_{T}(\tilde{F})=O(\log^{4}n\cdot\log k)\cdot\opt~.
	\]
	as required.
	The alternative tradeoff is obtained by using \Cref{thm:SubgraphAlt}
	instead of \Cref{thm:Subgraph}.
\end{proof}

\subsubsection{Online hop-constrained group Steiner
forest}\label{subsec:OnlineGroupForest}
The online group Steiner forest is an online version of the group Steiner
forest problem. Here we are given a graph $G=(V,E,w)$, and subset pairs
$(S_1,R_1),\dots,(S_i,R_i),\dots$ (where $S_i,R_i\subseteq V$) arrive one at a
time. The goal is to maintain forests $F_1,F_2,\dots$ such that for every $i$,
there are vertices $s_i\in S_i$ and $r_i\in R_i$, belonging to the same
connected component of $F_i$.
An algorithm has \emph{competitive ratio} $t$ if for every sequence of subset
pairs $(S_1,R_1),\dots,(S_i,R_i)$, the forest $F_i$ has weight at most $t$
times larger than the optimal solution for the offline group Steiner forest
problem with input subset pairs $(S_1,R_1),\dots,(S_i,R_i)$.
For randomized algorithms, the competitive ratio is the worst-case expected
ratio between the weights of the constructed forest and that of the
optimal solution.
It follows from Alon \etal \cite{AAABN06} that in full generality no
sub-polynomial approximation is possible.
Therefore, we will assume that there is a set of subset pairs
$\mathcal{K}\subseteq 2^{V}\times2^{V}$ of size $\kappa=|\mathcal{K}|$, and
all the subset pairs $(S_i,R_i)$ are chosen from $\mathcal{K}$.
For this version, Naor \etal \cite{NPS11} obtain an $O(\log^5n\cdot\log
\kappa)$ competitive ratio for the online group Steiner forest (answering an
open question of \cite{CEGS11}).

In the online h.c. group Steiner tree problem (introduced by \cite{HHZ21}),
there is an additional requirement that the number of hops in $F_i$ between a
member of $S_i$ and a member of $R_i$ is at most $\red{h}$, i.e.
$\hop_{F_i}(S_i,R_i)\le\red{h}$.
Using their $\red{h}$-Hop Repetition Tree Embedding, \cite{HHZ21} obtained
hop-stretch $\red{O(\log^3 n)}$, and competitive ratio $O(\log^7
n\cdot\log \kappa)$.
That is, constructing a sequence of subgraphs $F_1\subseteq
F_2\subseteq\dots$, such that for every $i$,
$\hop_{F_i}(S_i,R_i)\le\red{O(\log^3 n)\cdot h}$, and the (expected) weight of
$F_i$ is at most $O(\log^7 n\cdot\log \kappa)\cdot \opt_i$, where $\opt_i$ is
the cost of the optimal solution for the offline $\red{h}$ h.c. group Steiner
forest problem for subset pairs $(S_1,R_1),\dots,(S_i,R_i)$.
Using our \Cref{thm:Subgraph,thm:SubgraphAlt} we obtain a significant
improvement:

\begin{theorem}\label{thm:OnlineGroupSteinerForest}
	There is a poly-time randomized algorithm for the online
	$\red{h}$-hop-constrained group Steiner forest problem on graphs with
	polynomial aspect ratio, that maintains a solution $\{F_i\}_{i\ge
	1}$ which is
	$O(\log^5n\cdot\log \kappa)$-cost competitive (i.e. $\mathbb{E}[w(F_i)]\le
	O(\log^5n\cdot\log \kappa)\cdot\opt_i$), and such that $\hop_{F_i}(S_i,R_i)\le
	\red{O(\log^3n)\cdot h}$ for every $i$.
	Alternatively, one can obtain a solution which is
	$\tilde{O}(\log^5n)\cdot\log \kappa$-cost competitive, and such that
	$\hop_{F_i}(S_i,R_i)\le \red{\tilde{O}(\log^2 n)\cdot h}$ for every $i$.
\end{theorem}
Note that ignoring the hop constraints, our competitive ratio matches the
state of the art \cite{NPS11}.
\begin{proof}[Proof of \Cref{thm:OnlineGroupSteinerForest}]
	Consider an instance of the online $\red{h}$-h.c. group Steiner forest
	problem. Initially, we are given a graph $G=(V,E,w_G)$, hop parameter
	$\red{h}$, and a set of possible subset pairs
	$\mathcal{K}\subseteq2^V\times2^V$, where $\kappa=|\mathcal{K}|$.
	We apply \Cref{thm:Subgraph} to obtain a path-tree one-to-many embedding $f$
	into a tree $T$, with associated paths $\{P^{T}_e\}_{e\in E(T)}$.
	Naor, Panigrahi, and Singh \cite{NPS11} constructed an online algorithm
	for the group Steiner forest problems on trees of depth $d$ with
	competitive ratio $O(d\cdot\log^3 n\cdot\log \kappa)$.
	Following the discussion in the proof of \Cref{lem:PathClanEmbedding} the
	tree returned by \Cref{cor:hBoundedSubgraph} (and therefore also
	\Cref{thm:Subgraph}) has depth $O(\log n)$.

	We apply \cite{NPS11} online algorithm on the tree $T$ with the set of
	possible subset pairs
	$\tilde{\mathcal{K}}=\{(f(S_i),f(R_i))\mid(S_i,R_i)\in\mathcal{K}\}$.
	We will iteratively receive as an input subset pairs
	$(S_1,R_1),(S_2,R_2),\dots\in \mathcal{K}$, and feed Naor \etal
	\cite{NPS11} algorithm the subset pairs
	$(f(S_1),f(R_1)),(f(S_2),f(R_2)),\dots\in \tilde{\mathcal{K}}$. As a
	result, we will obtain forests $\tilde{F}_1,\tilde{F}_2,\dots$ of $T$.
	For every $i$, we will return the solution
	$\mathcal{F}_i=\cup_{e\in\tilde{F}_i}P^{T}_e$.

	For every $i$, $\tilde{F}_i$ contains a path $P_i=e_1,\dots,e_s$ from a
	vertex $s'\in f(S_i)$ to a vertex $r'\in f(R_i)$.
	It follows that the solution $\mathcal{F}_i$ we return will contain
	a path $P^{T}_{s',r'}$ from $f^{-1}(s')\in S_i$ to $f^{-1}(r')\in R_i$.
	As the path-tree one-to-many embedding $f$ has hop-bound
	$\red{O(\log^3n)\cdot h}$, it follows that
	$\hop_{\mathcal{F}_i}(S_i,R_i)\le \red{O(\log^3n)\cdot h}$ as required.
	We conclude that $\mathcal{F}_i$ is a valid solution to the problem.

	Next we bound the weight of $\mathcal{F}_i$. Let $H_i$ be an optimal
	solution in $G$ of weight $\opt_i$ on input $(S_1,R_1),\dots,(S_i,R_i)$.
	By \Cref{thm:Subgraph}, $T$ contains a subgraph $\tilde{H}_i$, such that
	for every $u,v\in H_i$ with $d_{H_i}(u,v)\le\red{h}$, $\tilde{H}_i$
	contains two vertices from $f(u)$ and $f(v)$ in the same connected
	component, and $w_T(\tilde{H}_i)=O(\log n)\cdot w_G(H_i)=O(\log
	n)\cdot\opt_i$.
	Note that $\tilde{H}_i$ is a valid solution to the group Steiner forest
	problem in $T$. Hence by \cite{NPS11},
	\[
	\mathbb{E}\left[w_{G}(\mathcal{F}_{i})\right]\le\mathbb{E}\left[w_{T}(
	\tilde{F}_{i})\right]\le O(\log^{4}n\cdot\log \kappa)\cdot w_{T}(
	\tilde{H}_{i})=O(\log^{5}n\cdot\log \kappa)\cdot\opt_i~.
	\]
	as required.
	The alternative tradeoff is obtained by using \Cref{thm:SubgraphAlt}
	instead of \Cref{thm:Subgraph}.
\end{proof}

\section{Hop-constrained metric data structures}
In this section we construct hop-constrained metric data structures.
Mendel and Naor \cite{MN07} (and later Abraham \etal \cite{ACEFN20}) used
Ramsey trees in order to construct metric data structures. This is done by
simply constructing Ramsey trees until every vertex belongs to the set $M$ in
some tree. Here we begin by following the exact same approach using
\Cref{lem:UltraMeasureAlt} to obtain
\Cref{lem:DistanceLabelingLargeStretch,lem:DistanceOracleLargeStretch} below.

The distortion factor in
\Cref{lem:DistanceLabelingLargeStretch,lem:DistanceOracleLargeStretch} is
large, and we can improve it significantly. In our actual metric data
structures we will use
\Cref{lem:DistanceLabelingLargeStretch,lem:DistanceOracleLargeStretch} as a
subroutine to obtain an initial coarse estimate of the hop-constrained
distance. Given such an estimate $d\rhop{h}_G(u,v)\approx2^i$, we can construct
an auxiliary graph $G_i$ where the shortest path distance (not hop-constrained)
is an accurate estimate of $d\rhop{h}_G(u,v)$. We then use state of the art
metric data structures on top of $G_i$ in a black box manner.

We begin with a construction of an asymmetric version of a distance labeling,
where each vertex has short and long labels.

\begin{lemma}\label{lem:DistanceLabelingLargeStretch}
	Given a weighted graph $G=(V,E,w)$ on $n$ vertices with polynomial aspect
	ratio, and parameters $k\in\N$, $\eps\in(0,1)$, $\red{h}\in\N$, there is
	an efficient construction of a distance labeling that assigns each node
	$v$ a label $\ell(v)$ consisting of two parts
	$\ell^S(v),\ell^L(v)$ such that:
	$\ell^S(v)$ has size $O(1)$, $\ell^L(v)$ has size $O(n^{\frac1k}\cdot k)$,
	and there is an algorithm $\mathcal{A}$ that for every pair of vertices
	$u,v$, given $\ell^L(v),\ell^S(u)$  returns a value
	$\A(\ell^L(v),\ell^S(u))$ such that $d^{\red{(O(k\cdot\log\log n)\cdot
	h)}}_{G}(u,v)\le \A(\ell^L(v),\ell^S(u))\le O(k\cdot\log\log n)\cdot
	d\rhop{h}_{G}(u,v)$.
\end{lemma}
\begin{proof}
	Set $\mathcal{M}_0=V$.
	For $i\ge 1$, apply \Cref{lem:UltraMeasureAlt}
	\footnote{\label{foot:DiffTradeLabeling}One might wonder why we use here
	\Cref{lem:UltraMeasureAlt}, and whether other interesting tradeoffs can be
	obtained if we instead use \Cref{lem:UltraMeasure}. The answer is
	that \Cref{lem:DistanceLabelingLargeStretch} here is not our end goal, but
	will only be used by
	\Cref{thm:DistanceOracle,thm:DistanceLabeling,thm:RoutingScheme} to get an
	initial rough estimate of the h.c. distance. In particular, there is no
	significant difference between distortion $O(k)$ and $O(k\cdot \log\log
	n)$. However, the difference between hop-stretch $O(k\cdot \log\log n)$
	and $O(k\cdot\log n)$ (which we will get by using \Cref{lem:UltraMeasure})
	is very significant. We are therefore using here only
	\Cref{lem:UltraMeasureAlt}.} with parameters
	$k,\red{h},\mathcal{M}_{i-1}$, and uniform measure $\mu$ to construct a
	Ramsey type embedding into ultrametric $U_i$ with Ramsey hop distortion
	$(O(k\cdot \log\log n),M_i,\red{O(k\cdot\log\log n)},\red{h})$. Set
	$\mathcal{M}_i=\mathcal{M}_{i-1}\backslash M_i$.
	This process halts once $\mathcal{M}_i=\emptyset$.

	To analyze the number of steps until the process halts, denote by $i_1$
	the first time $|\mathcal{M}_{i_1}|\le\frac12n$.
	As long as $i<i_1$, $|M_i|=\mu(M_i)\ge
	\mu(\mathcal{M}_{i-1})^{1-\frac1k}\ge (\frac n2)^{1-\frac1k}$.
	Thus for each such step we remove from $\cM_i$ at least $(\frac
	n2)^{1-\frac1k}$, until we reach $i_1$.
	It follows that
	$i_{1}\le\frac{n/2}{(\frac{n}{2})^{1-\frac{1}{k}}}=(\frac{n}{2})^{
	\frac{1}{k}}$.
	Similarly, let $i_2$ be the first index such that
	$|\mathcal{M}_{i_1+i_2}|\le\frac n4$. The same calculation will provide us
	with $i_2\le (\frac n4)^{\frac1k}$.
	Following this line of thought, the process must halt after at most
	\[
	\sum_{i=1}^{\log n}(\frac{n}{2^{i}})^{\frac{1}{k}}\le
	n^{\frac{1}{k}}\cdot\sum_{i=1}^{\infty}2^{-\frac{i}{k}}=n^{\frac{1}{k}}
	\cdot\frac{2^{-\frac{1}{k}}}{1-2^{-\frac{1}{k}}}=O(n^{\frac{1}{k}}\cdot k)
	\]
	rounds.
	For every vertex $v\in V$, denote by $\home(v)$ some index $i$ such
	that $v\in M_i$.

	A basic ingredient for our data structure will be distance
	labeling for trees.
	Exact distance labeling on an $n$-vertex tree requires  $\Theta(\log n)$
	words \cite{AGHP16}. In order to shave a $\log n$ factor, we will use
	an approximate labeling scheme.
	Freedman \etal \cite{FGNW17} (improving upon \cite{AGHP16,GKKPP01}) showed
	that for any $n$-vertex unweighted tree, and $\delta\in(0,1)$, one can
	construct an $(1+\delta)$-labeling scheme with labels of size
	$O(\log\frac1\delta)$ words.
	\begin{theorem}[\cite{FGNW17}]\label{thm:tree-label}
		For any $n$-vertex unweighted tree $T=(V,E)$ and parameter
		$\delta\in(0,1)$, there is a distance labeling scheme with stretch
		$1+\delta$, label size $O(\log\frac1\delta)$, and $O(1)$ query time.
	\end{theorem}
	The result of \cite{FGNW17} cannot be improved even when restricted to the
	case of (polynomially weighted) ultrametrics\footnote{The lower bound is
	based on \cite{AGHP16} $(h,M)$-trees, and uses only the labels on leaf
	vertices. One can observe that the metric induced by leaf vertices in
	$(h,M)$-trees is in fact an ultrametric.}.
	We will use \Cref{thm:tree-label} with $\delta=\frac12$. For this case the
	theorem can be extended to weighted trees with polynomial aspect ratio (by
	subdividing edges).

	For every ultrametric $U_i$, we will construct a distance labeling using
	\Cref{thm:tree-label} with parameter $\delta=\frac12$, where the label of
	$v$ is $\ell_i(v)$.
	For every vertex $v\in V$, $\ell^S(v)$ will consist of
	$\ell_{\home(v)}(v)$ and the index $\home(v)$, while $\ell_L(v)$ will be
	the concatenation of $\ell_i(v)$ for all the constructed ultrametrics. The
	size guarantees clearly hold, while given $\ell^L(v),\ell^S(u)$, we can in
	$O(1)$ time report the approximate distance between $u$ and $v$ in
	$U_{\home(u)}$, which satisfies:
	\begin{align*}
	\A(\ell_{\home(u)}(v),\ell_{\home(u)}(u)) & \le(1+\delta)\cdot
	d_{U_{\home(u)}}(u,v)=O(k\cdot\log\log n)\cdot d_{G}\rhop{h}(u,v)~.\\
	\A(\ell_{\home(u)}(v),\ell_{\home(u)}(u)) & \ge d_{U_{\home(u)}}(u,v)\ge
	d_{G}^{\red{(O(k\cdot\log\log n)\cdot h)}}(u,v)~.
	\end{align*}
\end{proof}

Next we construct a distance oracle.
\begin{lemma}\label{lem:DistanceOracleLargeStretch}
	Given a weighted graph $G=(V,E,w)$ on $n$ vertices with polynomial aspect
	ratio, and parameters $k\in\N$, $\eps\in(0,1)$, $\red{h}\in\N$, there is
	an efficient construction of a distance oracle $\DO$ of size
	$O(n^{1+\frac1k})$, which given a query $(u,v)$, in $O(1)$ times returns a
	value such that $d^{\red{(O(k\cdot\log\log n)\cdot
	h)}}_{G}(u,v)\le\DO(u,v)\le O(k\cdot\log\log n)\cdot  d\rhop{h}_{G}(u,v)$.
\end{lemma}
\begin{proof}
	We will slightly modify the construction in
	\Cref{lem:DistanceLabelingLargeStretch}.
	For $i\ge 1$, apply \Cref{thm:UltrametricRamseyAlt} \footnote{Similarly to
	\cref{foot:DiffTradeLabeling},  no real advantage can be obtained here by
	replacing \Cref{thm:UltrametricRamseyAlt} with
	\Cref{thm:UltrametricRamsey}.} to construct an ultrametric $U_i$ with set
	$M_i$, such that for every $v\in V$, $\Pr[v\in M_i]\ge
	\Omega(n^{-\frac1k})$, and such that for every $v\in M_i$, $u\in V$ it
	holds that $d_G^{\red{(O(k\cdot\log\log n)\cdot h)}}(u,v)\le
	d_{U_i}(u,v)\le O(k\cdot\log\log n)\cdot d_G\rhop{h}(u,v)$.

	For each vertex $v\in V$, let $\home(v)$ be the minimal index $i$ such
	that $v\in M_i$. Note that $\home(v)$ is distributed according to
	geometric distribution with parameter $\Omega(n^{-\frac1k})$.
	In particular
	$\mathbb{E}[\home(v)]=O(n^{\frac1k})$.

	For each index $i$, we will apply \Cref{thm:tree-label} (with
	$\delta=\frac12$) on the ultrametric $U_i$, to obtain distance labeling
	$\{\ell_i(v)\}_{v\in V}$.
	In our distance oracle, for every vertex $v$ we will store the index
	$\home(v)$ and the labels of $v$ in the first $\home(v)$ ultrametrics:
	$\ell_1(v),\dots,\ell_{\home(v)}(v)$. The expected size of the
	distance oracle is
	\[
	\mathbb{E}[|\DO|]=\sum_{v\in V}\mathbb{E}[\home(v)]\cdot
	O(1)=O(n^{1+\frac{1}{k}})~.
	\]
	After several attempts (using Markov's inequality), we can construct
	a distance oracle of $O(n^{1+\frac1k})$ size.

	Given a query $u,v$, let $i=\min\{\home(v),\home(u)\}$, we will use
	\cite{FGNW17} algorithm on $\ell_{i}(v),\ell_{i}(u)$.
	As a result, we will obtain an estimate
	\[
	d^{\red{(O(k\cdot\log\log n)\cdot h)}}_{G}(u,v)\le d_{U_i}(u,v)\le
	\A(u,v)\le (1+\delta)\cdot d_{U_i}(u,v)=O(k\cdot\log\log n)\cdot
	d\rhop{h}_{G}(u,v)~.
	\]
	Finally note that both finding $\home(v)$, and computing the estimate
	using \cite{FGNW17} takes $O(1)$ time.
\end{proof}

While \Cref{lem:DistanceLabelingLargeStretch,lem:DistanceOracleLargeStretch}
already provide h.c. labeling schemes and distance oracles, we can
significantly improve the stretch guarantee (even up to $1+\eps$).
Next we will construct auxiliary graphs $\{G_i\}_{i\ge0}$. Originally these
graphs were constructed by \cite{HHZ21}, and used to construct padded
decompositions. Later, we will construct our metric data structures upon
these graphs.
Fix $\beta=O(k\cdot\log\log n)$ such
that the distance estimation $\DO(u,v)$ in
\Cref{lem:DistanceOracleLargeStretch} (or $\A(\ell(v),\ell(u))$ in
\Cref{lem:DistanceLabelingLargeStretch}) satisfies  $d_{G}^{\red{(\beta\cdot
h)}}(u,v)\le\DO(u,v)\le \beta\cdot d_{G}\rhop{h}(u,v)$.
Let $G_{i}=(V,E,w_{i})$ be the graph where the weight of every edge was
increased by an additive factor of $\omega_{i}=\frac{\epsilon}{\beta\cdot
h}\cdot2^{i}$.
That is, for every edge $e$, $w_{i}(e)=w(e)+\omega_{i}$. Suppose
that $2^{i}\le\DO(u,v)<2^{i+1}$. We have that

\begin{align*}
d_{G}^{\red{(\beta\cdot h)}}(u,v) & \le\DO(u,v)\le2^{i+1}\\
2^i & \le\DO(u,v)\le \beta\cdot d_{G}\rhop{h}(u,v)~.
\end{align*}
Consider $G_{i}$. By looking at the shortest $\red{h}$-hop $u$-$v$ path in $G$
we have that
\begin{equation}
d_{G_{i}}(u,v)\le d_{G}\rhop{h}(u,v)+h\cdot\omega_{i}\le
d_{G}\rhop{h}(u,v)+\frac{\epsilon}{\beta}\cdot2^{i}=(1+\epsilon)\cdot
d_{G}\rhop{h}(u,v)~.\label{eq:Gi_UB}
\end{equation}
On the other hand
\begin{align}
d_{G_{i}}(u,v) & \ge\min\left\{ d_{G}\rhop{\frac{2\beta\cdot
h}{\epsilon}}(u,v),\frac{2\beta\cdot h}{\epsilon}\cdot\omega_{i}\right\}
=\min\left\{ d_{G}\rhop{\frac{2\beta\cdot h}{\epsilon}}(u,v),2^{i+1}\right\}
\overset{(*)}{=}
d_{G}^{\red{(O(\frac{k\cdot\log\log n}{\epsilon})\cdot
h)}}(u,v)~.\label{eq:Gi_LB}
\end{align}
Here the first inequality holds as follows: if the shortest $u$-$v$ path in $G_i$ has
less than $\red{\frac{2\beta\cdot h}{\epsilon}}$ hops, then $d_{G_{i}}(u,v)\ge
d_{G}\rhop{\frac{2\beta\cdot h}{\epsilon}}(u,v)$, while otherwise
$d_{G_{i}}(u,v)\ge\frac{2\beta\cdot h}{\epsilon}\cdot\omega_{i}$.
The equality $(*)$ follows as $d_{G}\rhop{\frac{2\beta\cdot
h}{\epsilon}}(u,v)\le d_{G}^{\red{(\beta\cdot h)}}(u,v)\le2^{i+1}$.
Furthermore, consider a path $P$ in $G_{i}$ of weight at most $O(k)\cdot
d_{G_{i}}(u,v)\le O(k)\cdot d_{G}\rhop{h}(u,v)\le O(k)\cdot2^{i}$.
Then
\begin{equation}
\hop(P)\le\frac{w_{i}(P)}{\omega_{i}}=\red{O(\frac{k\cdot\beta}{\epsilon})
\cdot h}=\red{O(\frac{k^2\cdot\log\log n}{\epsilon})\cdot
h}~.\label{eq:Gi_hops}
\end{equation}

\subsection{Distance oracles.}\label{subsec:distanceOracle}
This section is devoted to proving \Cref{thm:DistanceOracle}; we restate it
for convenience:
\DistanceOracle*
We will use Chechik's distance oracle:
\begin{theorem}[\cite{C15}]\label{thm:chechkikDistanceOracle}
	Given a weighted graph $G=(V,E,w)$ on $n$ vertices and parameter $k\in\N$,
	there exists a distance oracle $\DO$ of size $O(n^{1+\frac1k})$, which
	given a query $(u,v)$, in $O(1)$ time returns a value such that
	$d_{G}(u,v)\le\DO(u,v)\le (2k-1)\cdot  d_{G}(u,v)$.
\end{theorem}
Let $D$ be the maximum distance of a pair with at most $\red{h}$ hops, among
all pairs with finite $\red{h}$-hop distance, that is $D=\max_{u,v\in
V}\{d_G\rhop{h}(u,v)\mid d_G\rhop{h}(u,v)<\infty\}$. Suppose w.l.o.g. that the
minimum distance is $1$.
For every $i\in\left[0,\log \left(O(D\cdot k\cdot\log\log n)\right)\right]$,
let $G_{i}=(V,E,w_{i})$ be the graph defined above (w.r.t. $\omega_i$).
We use \Cref{thm:chechkikDistanceOracle} to construct a distance oracle
$\DO_i$ for $G_i$.
Our distance oracle $\DO$ will consist of the distance oracle
$\widetilde{\DO}$ constructed in \Cref{lem:DistanceOracleLargeStretch}, and
$\DO_i$ for $i\in\left[0,\log \left(O(D\cdot k\cdot\log\log n)\right)\right]$.
Clearly, the size of $\DO$ is bounded by $O(n^{1+\frac1k}\cdot\log n)$.

Given a query $(u,v)$, we first query $\widetilde{\DO}$, to obtain
an estimate such that
$1\le d^{\red{(O(k\cdot\log\log n)\cdot
h)}}_{G}(u,v)\le\widetilde{\DO}(u,v)\le O(k\cdot\log\log n)\cdot
d\rhop{h}_{G}(u,v)$.
If $\widetilde{\DO}(u,v)> \Omega(k\cdot\log\log n)\cdot D$ (with the same
constant as in the upper bound in \Cref{lem:DistanceOracleLargeStretch}), then
$d\rhop{h}_{G}(u,v)>D$ and we will just return $\infty$. Otherwise,
$\widetilde{\DO}(u,v)\le \Omega(k\cdot\log\log n)\cdot D$.
Let $i$ be the unique index such that
$2^i\le \widetilde{\DO}(u,v)<2^{i+1}$. We will return $\DO_i(u,v)$.
As we invoked two distance oracles, each with $O(1)$ query time, our query
time is $O(1)$ as well.

Finally, for the approximation factor, by equations (\ref{eq:Gi_UB}) and
(\ref{eq:Gi_LB}) it holds that
\[
d_{G}^{\red{(O(\frac{k\cdot\log\log n}{\epsilon})\cdot h)}}(u,v)\le
d_{G_{i}}(u,v)\le (1+\epsilon)\cdot d_{G}\rhop{h}(u,v)~.
\]
Thus by \Cref{thm:chechkikDistanceOracle}, it will hold that
\[
d_{G}^{\red{(O(\frac{k\cdot\log\log n}{\epsilon})\cdot h)}}(u,v)\le
\DO_{i}(u,v)\le (2k-1)(1+\epsilon)\cdot d_{G}\rhop{h}(u,v)~.
\]

\subsection{Distance labeling.}\label{subsec:DistanceLabeling}
This section is devoted to proving \Cref{thm:DistanceLabeling}; we restate it
for convenience:
\DistanceLabeling*
We will use Thorup-Zwick distance labeling:
\begin{theorem}[\cite{TZ05}]\label{thm:TZDistanceLabeling}
	Given a weighted graph $G=(V,E,w)$ on $n$ vertices, and parameter
	$k\in\N$, there is an efficient construction of a distance labeling that
	assigns each node $v$ a label $\ell(v)$ of size $O(n^{\frac1k}\cdot\log
	n)$, and such that there is an algorithm $\A$ that on input
	$\ell(u),\ell(v)$, in $O(k)$ time returns a value such that $d_{G}(u,v)\le
	\A(\ell(v),\ell(u))\le (2k-1)\cdot  d_{G}(u,v)$.
\end{theorem}

We will follow the lines of the construction in \Cref{subsec:distanceOracle}.
Recall that  $D=\max_{u,v\in V}\{d_G\rhop{h}(u,v)\mid
d_G\rhop{h}(u,v)<\infty\}$, and the minimum distance is $1$.
For every $i\in\left[0,\log \left(O(D\cdot k\cdot\log\log n)\right)\right]$,
let $G_{i}=(V,E,w_{i})$ be the graph defined above (w.r.t. $\omega_i$).
We use \Cref{thm:TZDistanceLabeling} to construct a distance labeling
$\ell_i$ for $G_i$.
The label $\ell(v)$ of a vertex $v$ will consist of the label
$\tilde{\ell}(v)$ constructed in \Cref{lem:DistanceLabelingLargeStretch} (with
associated algorithm $\tilde{\mathcal{A}}$), and $\ell_i(v)$ for
$i\in\left[0,\log \left(O(D\cdot k\cdot\log\log n)\right)\right]$ (with
associated algorithm $\mathcal{A}_i$). Clearly, the size of $\ell(v)$ is
bounded by $O(n^{1+\frac1k}\cdot\log^2 n)$.

Given $\ell(v),\ell(u)$, we first compute
$\tilde{\mathcal{A}}(\tilde{\ell}(v),\tilde{\ell}(u))$, to obtain
an estimate such that
$1\le d^{\red{(O(k\cdot\log\log n)\cdot
h)}}_{G}(u,v)\le\tilde{\mathcal{A}}(\tilde{\ell}(v),\tilde{\ell}(u))\le
O(k\cdot\log\log n)\cdot  d\rhop{h}_{G}(u,v)$.
If $\tilde{\mathcal{A}}(\tilde{\ell}(v),\tilde{\ell}(u))>\Omega(k\cdot\log\log
n)\cdot D$, then $d\rhop{h}_{G}(u,v)>D$ and we simply return $\infty$.
Otherwise, $d\rhop{h}_{G}(u,v)\le O(k\cdot\log\log n)\cdot D$.
Let $i$ be the unique index such that
$2^i\le \tilde{\mathcal{A}}(\tilde{\ell}(v),\tilde{\ell}(u))<2^{i+1}$. We will
return $\mathcal{A}_i\left(\ell_i(v),\ell_i(u)\right)$.
We invoked two labeling algorithms, the first with $O(1)$ query time, while
the second had query time $O(k)$. Hence our query time is $O(k)$.
Finally, for the approximation factor, by equations (\ref{eq:Gi_UB}),
(\ref{eq:Gi_LB}), and \Cref{thm:TZDistanceLabeling} it holds that
\[
d_{G}^{\red{(O(\frac{k\cdot\log\log n}{\epsilon})\cdot h)}}(u,v)\le
\A(\ell(v),\ell(u))\le (2k-1)(1+\epsilon)\cdot d_{G}\rhop{h}(u,v)~.
\]

\begin{remark}
	If instead of using \cite{TZ05} distance labeling, one would use a Ramsey
	tree based one (such as in \Cref{lem:DistanceLabelingLargeStretch}), it is
	possible to obtain constant query time $O(1)$, while the stretch will
	increase from $(2k-1)(1+\eps)$ to $2ek\cdot(1+\eps)$ (using the best known
	Ramsey trees \cite{NT12}).
\end{remark}

\subsection{Compact routing scheme.}
This section is devoted to proving \Cref{thm:RoutingScheme}; we restate it for
convenience:
\RoutingScheme*
We will use Chechik's compact routing scheme:
\begin{theorem}[\cite{C13}]\label{thm:ChechikRouting}
	For every weighted graph $G=(V,E,w)$ on $n$ vertices and parameter
	$k\in\N$, there is an efficient construction of a compact routing scheme
	that assigns each node a table of size $O(k\cdot n^{\frac{1}{k}})$, a label
	of size $O(k\cdot\log n)$, and such that routing a packet from $u$ to
	$v$ will be done using a path $P$ such that $w(P)
	\le3.68k\cdot d_G(u,v)$.
\end{theorem}

We will follow the lines of the construction in
\Cref{subsec:DistanceLabeling}.
Recall that  $D=\max_{u,v\in V}\{d_G\rhop{h}(u,v)\mid
d_G\rhop{h}(u,v)<\infty\}$, and the minimum distance is $1$.
For every $i\in\left[0,\log \left(O(D\cdot k\cdot\log\log n)\right)\right]$,
let $G_{i}=(V,E,w_{i})$ be the graph defined above (w.r.t. $\omega_i$).
We use \Cref{thm:ChechikRouting} to construct a compact routing scheme for each
$G_i$, where the label and table of $v$ are denoted by $L_i(v),T_i(v)$, and have
sizes $O(k\cdot\log n)$ and $O(k\cdot n^{\frac{1}{k}})$, respectively.
In addition, we construct an asymmetric distance labeling scheme using
\Cref{lem:DistanceLabelingLargeStretch}, where each vertex $v$ gets two labels
$\ell^S(v),\ell^L(v)$ of sizes $O(1),O(n^{\frac1k}\cdot k)$, respectively.

For the routing scheme, the label $L(v)$ of a vertex $v$ will consist of
$\ell^S(v)$ and $L_i(v)$ for every $i\in\left[0,\log \left(O(D\cdot
k\cdot\log\log n)\right)\right]$. Note that $L(v)$ has size
$O(k\cdot\log^2 n)$.
The table $T(v)$ of $v$ will consist of $\ell^L(v)$ and $T_i(v)$ for every
$i\in\left[0,\log \left(O(D\cdot k\cdot\log\log n)\right)\right]$. Thus $T(v)$
has size $O(n^{\frac1k}\cdot k\cdot\log n)$.

Suppose that we wish to send a packet from $v$ to $u$. First, using the
distance labeling $\ell^L(v)$ and $\ell^S(u)$, we compute an estimate
$\mathcal{A}(\ell^L(v),\ell^S(u))$ such that
$d^{\red{(O(k\cdot\log\log n)\cdot h)}}_{G}(u,v)\le \A(\ell^L(v),\ell^S(u))\le
O(k\cdot\log\log n)\cdot  d\rhop{h}_{G}(u,v)$.
If $\A(\ell^L(v),\ell^S(u))> \Omega(k\cdot\log\log n)\cdot D$, then
$d\rhop{h}_{G}(u,v)>D$ (implying $\hop_G(u,v)>\red{h}$) and we can do nothing
(as we are not required to deliver the package). Otherwise,
$\A(\ell^L(v),\ell^S(u))\le \Omega(k\cdot\log\log n)\cdot D$.
Let $i$ be the unique index such that
$2^i\le\mathcal{A}(\ell^L(v),\ell^S(u))<2^{i+1}$. We will send the packet
using the routing scheme constructed for $G_i$ (the index $i$ will be attached
to the header of each message).
Then the packet will reach $u$ using a path $P$ such that $w(P)< w_i(P)\le
3.68k\cdot d_{G_i}(u,v)\overset{(\ref{eq:Gi_UB})}{\le} 3.68k\cdot(1+\eps)\cdot
d\rhop{h}_{G}(u,v)$, while the number of hops in $P$ is bounded by
$\hop(P)\overset{(\ref{eq:Gi_hops})}{\le} O(\frac{k^2\cdot\log\log
n}{\eps})\cdot h$.

{\small

	\bibliographystyle{alphaurlinit}
	\bibliography{RamseyTreewidthBib,RPTALGbib,Journal}

\newcommand{\etalchar}[1]{$^{#1}$}
\begin{thebibliography}{BLMN05b}

\bibitem[AAA{\etalchar{+}}06]{AAABN06}
N.~Alon, B.~Awerbuch, Y.~Azar, N.~Buchbinder, and J.~Naor.
\newblock A general approach to online network optimization problems.
\newblock {\em {ACM} Trans. Algorithms}, 2(4):640--660, 2006.
\newblock
\newblock preliminary version published in SODA 2004, \href
  {http://dx.doi.org/10.1145/1198513.1198522}
  {\path{doi:10.1145/1198513.1198522}}.

\bibitem[ABC{\etalchar{+}}05]{CBCDGKNS05}
I.~Abraham, Y.~Bartal, T.~H. Chan, K.~Dhamdhere, A.~Gupta, J.~M. Kleinberg,
  O.~Neiman, and A.~Slivkins.
\newblock Metric embeddings with relaxed guarantees.
\newblock In {\em 46th Annual {IEEE} Symposium on Foundations of Computer
  Science {(FOCS} 2005), 23-25 October 2005, Pittsburgh, PA, USA, Proceedings},
  pages 83--100. {IEEE} Computer Society,
\newblock 2005, \href {http://dx.doi.org/10.1109/SFCS.2005.51}
  {\path{doi:10.1109/SFCS.2005.51}}.

\bibitem[ABLP90]{ABLP90}
B.~Awerbuch, A.~Bar{-}Noy, N.~Linial, and D.~Peleg.
\newblock Improved routing strategies with succinct tables.
\newblock {\em J. Algorithms}, 11(3):307--341,
\newblock 1990, \href {http://dx.doi.org/10.1016/0196-6774(90)90017-9}
  {\path{doi:10.1016/0196-6774(90)90017-9}}.

\bibitem[ABN11]{ABN11}
I.~Abraham, Y.~Bartal, and O.~Neiman.
\newblock Advances in metric embedding theory.
\newblock {\em Advances in Mathematics}, 228(6):3026 -- 3126,
\newblock 2011, \href {http://dx.doi.org/10.1016/j.aim.2011.08.003}
  {\path{doi:10.1016/j.aim.2011.08.003}}.

\bibitem[ACE{\etalchar{+}}20]{ACEFN20}
I.~Abraham, S.~Chechik, M.~Elkin, A.~Filtser, and O.~Neiman.
\newblock Ramsey spanning trees and their applications.
\newblock {\em {ACM} Trans. Algorithms}, 16(2):19:1--19:21, 2020.
\newblock
\newblock preliminary version published in SODA 2018, \href
  {http://dx.doi.org/10.1145/3371039} {\path{doi:10.1145/3371039}}.

\bibitem[ACKW15]{ACKW15}
I.~Abraham, S.~Chechik, R.~Krauthgamer, and U.~Wieder.
\newblock Approximate nearest neighbor search in metrics of planar graphs.
\newblock In N.~Garg, K.~Jansen, A.~Rao, and J.~D.~P. Rolim, editors, {\em
  Approximation, Randomization, and Combinatorial Optimization. Algorithms and
  Techniques, {APPROX/RANDOM} 2015, August 24-26, 2015, Princeton, NJ, {USA}},
  volume~40 of {\em LIPIcs}, pages 20--42. Schloss Dagstuhl - Leibniz-Zentrum
  f{\"{u}}r Informatik,
\newblock 2015, \href {http://dx.doi.org/10.4230/LIPICS.APPROX-RANDOM.2015.20}
  {\path{doi:10.4230/LIPICS.APPROX-RANDOM.2015.20}}.

\bibitem[AFH{\etalchar{+}}05]{AFHKRS05}
E.~Althaus, S.~Funke, S.~Har{-}Peled, J.~K{\"{o}}nemann, E.~A. Ramos, and
  M.~Skutella.
\newblock Approximating \emph{k}-hop minimum-spanning trees.
\newblock {\em Oper. Res. Lett.}, 33(2):115--120, 2005.
\newblock
\newblock preliminary version published in SODA 2003, \href
  {http://dx.doi.org/10.1016/j.orl.2004.05.005}
  {\path{doi:10.1016/j.orl.2004.05.005}}.

\bibitem[AGHP16]{AGHP16}
S.~Alstrup, I.~L. G{\o}rtz, E.~B. Halvorsen, and E.~Porat.
\newblock Distance labeling schemes for trees.
\newblock In {\em 43rd International Colloquium on Automata, Languages, and
  Programming, {ICALP} 2016, July 11-15, 2016, Rome, Italy}, pages
  132:1--132:16,
\newblock 2016, \href {http://dx.doi.org/10.4230/LIPIcs.ICALP.2016.132}
  {\path{doi:10.4230/LIPIcs.ICALP.2016.132}}.

\bibitem[AIR18]{AIR18}
A.~Andoni, P.~Indyk, and I.~P. Razenshteyn.
\newblock Approximate nearest neighbor search in high dimensions.
\newblock {\em CoRR}, abs/1806.09823,
\newblock 2018, \href {http://arxiv.org/abs/1806.09823}
  {\path{arXiv:1806.09823}}.

\bibitem[AKPW95]{AKPW95}
N.~Alon, R.~M. Karp, D.~Peleg, and D.~B. West.
\newblock A graph-theoretic game and its application to the k-server problem.
\newblock {\em {SIAM} J. Comput.}, 24(1):78--100, 1995.
\newblock
\newblock preliminary version published in On-Line Algorithms 1991, \href
  {http://dx.doi.org/10.1137/S0097539792224474}
  {\path{doi:10.1137/S0097539792224474}}.

\bibitem[AKR95]{AKR95}
A.~Agrawal, P.~N. Klein, and R.~Ravi.
\newblock When trees collide: An approximation algorithm for the generalized
  steiner problem on networks.
\newblock {\em {SIAM} J. Comput.}, 24(3):440--456, 1995.
\newblock
\newblock preliminary version published in STOC 1991, \href
  {http://dx.doi.org/10.1137/S0097539792236237}
  {\path{doi:10.1137/S0097539792236237}}.

\bibitem[AMS94]{AMS94}
S.~Arya, D.~M. Mount, and M.~H.~M. Smid.
\newblock Randomized and deterministic algorithms for geometric spanners of
  small diameter.
\newblock In {\em 35th Annual Symposium on Foundations of Computer Science,
  Santa Fe, New Mexico, USA, 20-22 November 1994}, pages 703--712,
\newblock 1994, \href {http://dx.doi.org/10.1109/SFCS.1994.365722}
  {\path{doi:10.1109/SFCS.1994.365722}}.

\bibitem[AMS99]{AMS99Sketch}
N.~Alon, Y.~Matias, and M.~Szegedy.
\newblock The space complexity of approximating the frequency moments.
\newblock {\em J. Comput. Syst. Sci.}, 58(1):137--147, 1999.
\newblock
\newblock preliminary version published in STOC 1996, \href
  {http://dx.doi.org/10.1006/jcss.1997.1545}
  {\path{doi:10.1006/jcss.1997.1545}}.

\bibitem[AN19]{AN19}
I.~Abraham and O.~Neiman.
\newblock Using petal-decompositions to build a low stretch spanning tree.
\newblock {\em {SIAM} J. Comput.}, 48(2):227--248, 2019.
\newblock
\newblock preliminary version published in STOC 2012, \href
  {http://dx.doi.org/10.1137/17M1115575} {\path{doi:10.1137/17M1115575}}.

\bibitem[AP90]{AP90}
B.~Awerbuch and D.~Peleg.
\newblock Sparse partitions.
\newblock In {\em Proceedings of the 31st IEEE Symposium on Foundations of
  Computer Science (FOCS)}, pages 503--513,
\newblock 1990, \href {http://dx.doi.org/10.1109/FSCS.1990.89571}
  {\path{doi:10.1109/FSCS.1990.89571}}.

\bibitem[AP92]{AP92}
B.~Awerbuch and D.~Peleg.
\newblock Routing with polynomial communication-space trade-off.
\newblock {\em {SIAM} J. Discret. Math.}, 5(2):151--162,
\newblock 1992, \href {http://dx.doi.org/10.1137/0405013}
  {\path{doi:10.1137/0405013}}.

\bibitem[AS03]{AS03}
V.~Athitsos and S.~Sclaroff.
\newblock Database indexing methods for 3d hand pose estimation.
\newblock In {\em Gesture-Based Communication in Human-Computer Interaction,
  5th International Gesture Workshop, {GW} 2003, Genova, Italy, April 15-17,
  2003, Selected Revised Papers}, pages 288--299,
\newblock 2003, \href {http://dx.doi.org/10.1007/978-3-540-24598-8\_27}
  {\path{doi:10.1007/978-3-540-24598-8\_27}}.

\bibitem[ASZ20]{ASZ20}
A.~Andoni, C.~Stein, and P.~Zhong.
\newblock Parallel approximate undirected shortest paths via low hop emulators.
\newblock In {\em Proccedings of the 52nd Annual {ACM} {SIGACT} Symposium on
  Theory of Computing, {STOC} 2020, Chicago, IL, USA, June 22-26, 2020}, pages
  322--335,
\newblock 2020, \href {http://dx.doi.org/10.1145/3357713.3384321}
  {\path{doi:10.1145/3357713.3384321}}.

\bibitem[AT11]{AT11}
I.~Akg{\"{u}}n and B.~{\c{C}}. Tansel.
\newblock New formulations of the hop-constrained minimum spanning tree problem
  via miller-tucker-zemlin constraints.
\newblock {\em Eur. J. Oper. Res.}, 212(2):263--276,
\newblock 2011, \href {http://dx.doi.org/10.1016/j.ejor.2011.01.051}
  {\path{doi:10.1016/j.ejor.2011.01.051}}.

\bibitem[BA92]{BA92}
A.~Balakrishnan and K.~Altinkemer.
\newblock Using a hop-constrained model to generate alternative communication
  network design.
\newblock {\em {INFORMS} J. Comput.}, 4(2):192--205,
\newblock 1992, \href {http://dx.doi.org/10.1287/ijoc.4.2.192}
  {\path{doi:10.1287/ijoc.4.2.192}}.

\bibitem[Bar96]{Bar96}
Y.~Bartal.
\newblock Probabilistic approximations of metric spaces and its algorithmic
  applications.
\newblock In {\em 37th Annual Symposium on Foundations of Computer Science,
  {FOCS} '96, Burlington, Vermont, USA, 14-16 October, 1996}, pages 184--193,
\newblock 1996, \href {http://dx.doi.org/10.1109/SFCS.1996.548477}
  {\path{doi:10.1109/SFCS.1996.548477}}.

\bibitem[Bar98]{Bartal98}
Y.~Bartal.
\newblock On approximating arbitrary metrices by tree metrics.
\newblock In {\em Proceedings of the Thirtieth Annual {ACM} Symposium on the
  Theory of Computing, Dallas, Texas, USA, May 23-26, 1998}, pages 161--168,
\newblock 1998, \href {http://dx.doi.org/10.1145/276698.276725}
  {\path{doi:10.1145/276698.276725}}.

\bibitem[Bar04]{Bartal04}
Y.~Bartal.
\newblock Graph decomposition lemmas and their role in metric embedding
  methods.
\newblock In {\em Algorithms - {ESA} 2004, 12th Annual European Symposium,
  Bergen, Norway, September 14-17, 2004, Proceedings}, pages 89--97,
\newblock 2004, \href {http://dx.doi.org/10.1007/978-3-540-30140-0\_10}
  {\path{doi:10.1007/978-3-540-30140-0\_10}}.

\bibitem[Bar21]{Bar21}
Y.~Bartal.
\newblock Advances in metric ramsey theory and its applications.
\newblock {\em CoRR}, abs/2104.03484,
\newblock 2021, \href {http://arxiv.org/abs/2104.03484}
  {\path{arXiv:2104.03484}}.

\bibitem[BBM06]{BBM06}
Y.~Bartal, B.~Bollob{\'{a}}s, and M.~Mendel.
\newblock Ramsey-type theorems for metric spaces with applications to online
  problems.
\newblock {\em J. Comput. Syst. Sci.}, 72(5):890--921, 2006.
\newblock
\newblock Special Issue on FOCS 2001, \href
  {http://dx.doi.org/10.1016/j.jcss.2005.05.008}
  {\path{doi:10.1016/j.jcss.2005.05.008}}.

\bibitem[BCL{\etalchar{+}}18]{BCLLM18}
S.~Bubeck, M.~B. Cohen, Y.~T. Lee, J.~R. Lee, and A.~Madry.
\newblock k-server via multiscale entropic regularization.
\newblock In {\em Proceedings of the 50th Annual {ACM} {SIGACT} Symposium on
  Theory of Computing, {STOC} 2018, Los Angeles, CA, USA, June 25-29, 2018},
  pages 3--16,
\newblock 2018, \href {http://dx.doi.org/10.1145/3188745.3188798}
  {\path{doi:10.1145/3188745.3188798}}.

\bibitem[BF18]{BF18}
J.~D. Boeck and B.~Fortz.
\newblock Extended formulation for hop constrained distribution network
  configuration problems.
\newblock {\em Eur. J. Oper. Res.}, 265(2):488--502,
\newblock 2018, \href {http://dx.doi.org/10.1016/j.ejor.2017.08.017}
  {\path{doi:10.1016/j.ejor.2017.08.017}}.

\bibitem[BFG15]{BFG15}
Q.~Botton, B.~Fortz, and L.~E.~N. Gouveia.
\newblock On the hop-constrained survivable network design problem with
  reliable edges.
\newblock {\em Comput. Oper. Res.}, 64:159--167,
\newblock 2015, \href {http://dx.doi.org/10.1016/j.cor.2015.05.009}
  {\path{doi:10.1016/j.cor.2015.05.009}}.

\bibitem[BFGP13]{BFGP13}
Q.~Botton, B.~Fortz, L.~E.~N. Gouveia, and M.~Poss.
\newblock Benders decomposition for the hop-constrained survivable network
  design problem.
\newblock {\em {INFORMS} J. Comput.}, 25(1):13--26,
\newblock 2013, \href {http://dx.doi.org/10.1287/ijoc.1110.0472}
  {\path{doi:10.1287/ijoc.1110.0472}}.

\bibitem[BFM86]{BFM86}
J.~Bourgain, T.~Figiel, and V.~Milman.
\newblock On {H}ilbertian subsets of finite metric spaces.
\newblock {\em Israel J. Math.}, 55(2):147--152,
\newblock 1986, \href {http://dx.doi.org/10.1007/BF02801990}
  {\path{doi:10.1007/BF02801990}}.

\bibitem[BFN19]{BFN19}
Y.~Bartal, A.~Filtser, and O.~Neiman.
\newblock On notions of distortion and an almost minimum spanning tree with
  constant average distortion.
\newblock {\em J. Comput. Syst. Sci.}, 105:116--129, 2019.
\newblock
\newblock preliminary version published in SODA 2016, \href
  {http://dx.doi.org/10.1016/j.jcss.2019.04.006}
  {\path{doi:10.1016/j.jcss.2019.04.006}}.

\bibitem[BFN22]{BFN22}
Y.~Bartal, O.~N. Fandina, and O.~Neiman.
\newblock Covering metric spaces by few trees.
\newblock {\em J. Comput. Syst. Sci.}, 130:26--42,
\newblock 2022, \href {http://dx.doi.org/10.1016/J.JCSS.2022.06.001}
  {\path{doi:10.1016/J.JCSS.2022.06.001}}.

\bibitem[BGS16]{BGS16}
G.~E. Blelloch, Y.~Gu, and Y.~Sun.
\newblock A new efficient construction on probabilistic tree embeddings.
\newblock {\em CoRR}, abs/1605.04651, 2016.
\newblock
\newblock \url{https://arxiv.org/abs/1605.04651}, \href
  {http://arxiv.org/abs/1605.04651} {\path{arXiv:1605.04651}}.

\bibitem[BHR13]{BHR13}
A.~Bley, S.~M. Hashemi, and M.~Rezapour.
\newblock {IP} modeling of the survivable hop constrained connected facility
  location problem.
\newblock {\em Electron. Notes Discret. Math.}, 41:463--470,
\newblock 2013, \href {http://dx.doi.org/10.1016/j.endm.2013.05.126}
  {\path{doi:10.1016/j.endm.2013.05.126}}.

\bibitem[BKP01]{BKP01}
J.~Bar{-}Ilan, G.~Kortsarz, and D.~Peleg.
\newblock Generalized submodular cover problems and applications.
\newblock {\em Theor. Comput. Sci.}, 250(1-2):179--200, 2001.
\newblock
\newblock preliminary version published in ISTCS 1996, \href
  {http://dx.doi.org/10.1016/S0304-3975(99)00130-9}
  {\path{doi:10.1016/S0304-3975(99)00130-9}}.

\bibitem[BLMN05a]{BLMN05b}
Y.~Bartal, N.~Linial, M.~Mendel, and A.~Naor.
\newblock On metric {R}amsey-type dichotomies.
\newblock {\em Journal of the London Mathematical Society}, 71(2):289--303,
\newblock 2005, \href {http://dx.doi.org/10.1112/S0024610704006155}
  {\path{doi:10.1112/S0024610704006155}}.

\bibitem[BLMN05b]{BLMN05}
Y.~Bartal, N.~Linial, M.~Mendel, and A.~Naor.
\newblock Some low distortion metric ramsey problems.
\newblock {\em Discret. Comput. Geom.}, 33(1):27--41,
\newblock 2005, \href {http://dx.doi.org/10.1007/s00454-004-1100-z}
  {\path{doi:10.1007/s00454-004-1100-z}}.

\bibitem[BM04]{BM04multi}
Y.~Bartal and M.~Mendel.
\newblock Multiembedding of metric spaces.
\newblock {\em {SIAM} J. Comput.}, 34(1):248--259, 2004.
\newblock
\newblock preliminary version published in SODA 2003, \href
  {http://dx.doi.org/10.1137/S0097539703433122}
  {\path{doi:10.1137/S0097539703433122}}.

\bibitem[CCL{\etalchar{+}}23]{CCLMST23}
H.~Chang, J.~Conroy, H.~Le, L.~Milenkovic, S.~Solomon, and C.~Than.
\newblock Covering planar metrics (and beyond): {O(1)} trees suffice.
\newblock In {\em 64th {IEEE} Annual Symposium on Foundations of Computer
  Science, {FOCS} 2023, Santa Cruz, CA, USA, November 6-9, 2023}, pages
  2231--2261. {IEEE},
\newblock 2023, \href {http://dx.doi.org/10.1109/FOCS57990.2023.00139}
  {\path{doi:10.1109/FOCS57990.2023.00139}}.

\bibitem[CDG06]{CDG06}
T.~H. Chan, M.~Dinitz, and A.~Gupta.
\newblock Spanners with slack.
\newblock In Y.~Azar and T.~Erlebach, editors, {\em Algorithms - {ESA} 2006,
  14th Annual European Symposium, Zurich, Switzerland, September 11-13, 2006,
  Proceedings}, volume 4168 of {\em Lecture Notes in Computer Science}, pages
  196--207. Springer,
\newblock 2006, \href {http://dx.doi.org/10.1007/11841036\_20}
  {\path{doi:10.1007/11841036\_20}}.

\bibitem[CEGS11]{CEGS11}
C.~Chekuri, G.~Even, A.~Gupta, and D.~Segev.
\newblock Set connectivity problems in undirected graphs and the directed
  steiner network problem.
\newblock {\em {ACM} Trans. Algorithms}, 7(2):18:1--18:17, 2011.
\newblock
\newblock preliminary version published in SODA 2008, \href
  {http://dx.doi.org/10.1145/1921659.1921664}
  {\path{doi:10.1145/1921659.1921664}}.

\bibitem[CFKL20]{CFKL20}
V.~Cohen{-}Addad, A.~Filtser, P.~N. Klein, and H.~Le.
\newblock On light spanners, low-treewidth embeddings and efficient traversing
  in minor-free graphs.
\newblock In {\em 61st {IEEE} Annual Symposium on Foundations of Computer
  Science, {FOCS} 2020, Durham, NC, USA, November 16-19, 2020}, pages 589--600,
\newblock 2020, \href {http://dx.doi.org/10.1109/FOCS46700.2020.00061}
  {\path{doi:10.1109/FOCS46700.2020.00061}}.

\bibitem[Che13]{C13}
S.~Chechik.
\newblock Compact routing schemes with improved stretch.
\newblock In {\em {ACM} Symposium on Principles of Distributed Computing,
  {PODC} '13, Montreal, QC, Canada, July 22-24, 2013}, pages 33--41,
\newblock 2013, \href {http://dx.doi.org/10.1145/2484239.2484268}
  {\path{doi:10.1145/2484239.2484268}}.

\bibitem[Che14]{C14}
S.~Chechik.
\newblock Approximate distance oracles with constant query time.
\newblock In {\em Proceedings of the 46th Annual ACM Symposium on Theory of
  Computing}, STOC '14, pages 654--663, New York, NY, USA, 2014.
\newblock ACM, \href {http://dx.doi.org/10.1145/2591796.2591801}
  {\path{doi:10.1145/2591796.2591801}}.

\bibitem[Che15]{C15}
S.~Chechik.
\newblock Approximate distance oracles with improved bounds.
\newblock In {\em Proceedings of the Forty-Seventh Annual {ACM} on Symposium on
  Theory of Computing, {STOC} 2015, Portland, OR, USA, June 14-17, 2015}, pages
  1--10,
\newblock 2015, \href {http://dx.doi.org/10.1145/2746539.2746562}
  {\path{doi:10.1145/2746539.2746562}}.

\bibitem[CLPP23]{CLPP23}
V.~Cohen{-}Addad, H.~Le, M.~Pilipczuk, and M.~Pilipczuk.
\newblock Planar and minor-free metrics embed into metrics of polylogarithmic
  treewidth with expected multiplicative distortion arbitrarily close to 1.
\newblock In {\em 64th {IEEE} Annual Symposium on Foundations of Computer
  Science, {FOCS} 2023, Santa Cruz, CA, USA, November 6-9, 2023}, pages
  2262--2277. {IEEE},
\newblock 2023, \href {http://dx.doi.org/10.1109/FOCS57990.2023.00140}
  {\path{doi:10.1109/FOCS57990.2023.00140}}.

\bibitem[CN24]{CN24}
Y.~Cherapanamjeri and J.~Nelson.
\newblock Terminal embeddings in sublinear time.
\newblock {\em TheoretiCS}, 3,
\newblock 2024, \href {http://dx.doi.org/10.46298/THEORETICS.24.6}
  {\path{doi:10.46298/THEORETICS.24.6}}.

\bibitem[Coh00]{Cohen00}
E.~Cohen.
\newblock Polylog-time and near-linear work approximation scheme for undirected
  shortest paths.
\newblock {\em J. {ACM}}, 47(1):132--166, 2000.
\newblock
\newblock preliminary version published in STOC 1994, \href
  {http://dx.doi.org/10.1145/331605.331610} {\path{doi:10.1145/331605.331610}}.

\bibitem[Cow01]{Cowen01}
L.~Cowen.
\newblock Compact routing with minimum stretch.
\newblock {\em J. Algorithms}, 38(1):170--183, 2001.
\newblock
\newblock preliminary version published in SODA 1999, \href
  {http://dx.doi.org/10.1006/jagm.2000.1134}
  {\path{doi:10.1006/jagm.2000.1134}}.

\bibitem[DGM{\etalchar{+}}16]{DGMG16}
I.~Diarrassouba, V.~Gabrel, A.~R. Mahjoub, L.~E.~N. Gouveia, and P.~Pesneau.
\newblock Integer programming formulations for the \emph{k}-edge-connected
  3-hop-constrained network design problem.
\newblock {\em Networks}, 67(2):148--169,
\newblock 2016, \href {http://dx.doi.org/10.1002/net.21667}
  {\path{doi:10.1002/net.21667}}.

\bibitem[DKR16]{DKR16}
M.~Dinitz, G.~Kortsarz, and R.~Raz.
\newblock Label cover instances with large girth and the hardness of
  approximating basic \emph{k}-spanner.
\newblock {\em {ACM} Trans. Algorithms}, 12(2):25:1--25:16, 2016.
\newblock
\newblock preliminary version published in ICALP 2012, \href
  {http://dx.doi.org/10.1145/2818375} {\path{doi:10.1145/2818375}}.

\bibitem[DMMY18]{DMMY18}
I.~Diarrassouba, M.~Mahjoub, A.~R. Mahjoub, and H.~Yaman.
\newblock k-node-disjoint hop-constrained survivable networks: polyhedral
  analysis and branch and cut.
\newblock {\em Ann. des T{\'{e}}l{\'{e}}communications}, 73(1-2):5--28,
\newblock 2018, \href {http://dx.doi.org/10.1007/s12243-017-0622-3}
  {\path{doi:10.1007/s12243-017-0622-3}}.

\bibitem[EEST08]{EEST08}
M.~Elkin, Y.~Emek, D.~A. Spielman, and S.~Teng.
\newblock Lower-stretch spanning trees.
\newblock {\em {SIAM} J. Comput.}, 38(2):608--628, 2008.
\newblock
\newblock preliminary version published in STOC 2005, \href
  {http://dx.doi.org/10.1137/050641661} {\path{doi:10.1137/050641661}}.

\bibitem[EFN17]{EFN17}
M.~Elkin, A.~Filtser, and O.~Neiman.
\newblock Terminal embeddings.
\newblock {\em Theoretical Computer Science}, 697:1 -- 36,
\newblock 2017, \href
  {http://dx.doi.org/https://doi.org/10.1016/j.tcs.2017.06.021}
  {\path{doi:https://doi.org/10.1016/j.tcs.2017.06.021}}.

\bibitem[EFN18]{EFN18}
M.~Elkin, A.~Filtser, and O.~Neiman.
\newblock Prioritized metric structures and embedding.
\newblock {\em {SIAM} J. Comput.}, 47(3):829--858, 2018.
\newblock
\newblock preliminary version published in STOC 2015, \href
  {http://dx.doi.org/10.1137/17M1118749} {\path{doi:10.1137/17M1118749}}.

\bibitem[EGP03]{EGP03}
T.~Eilam, C.~Gavoille, and D.~Peleg.
\newblock Compact routing schemes with low stretch factor.
\newblock {\em J. Algorithms}, 46(2):97--114, 2003.
\newblock
\newblock preliminary version published in PODC 1998, \href
  {http://dx.doi.org/10.1016/S0196-6774(03)00002-6}
  {\path{doi:10.1016/S0196-6774(03)00002-6}}.

\bibitem[EN19]{EN19hop}
M.~Elkin and O.~Neiman.
\newblock Hopsets with constant hopbound, and applications to approximate
  shortest paths.
\newblock {\em {SIAM} J. Comput.}, 48(4):1436--1480, 2019.
\newblock
\newblock preliminary version published in FOCS 2016, \href
  {http://dx.doi.org/10.1137/18M1166791} {\path{doi:10.1137/18M1166791}}.

\bibitem[EN20]{EN20}
M.~Elkin and O.~Neiman.
\newblock Near-additive spanners and near-exact hopsets, {A} unified view.
\newblock {\em Bull. {EATCS}}, 130, 2020.
\newblock
\newblock see
  \href{http://bulletin.eatcs.org/index.php/beatcs/article/view/608}{here}.

\bibitem[EN21]{EN21}
M.~Elkin and O.~Neiman.
\newblock Near isometric terminal embeddings for doubling metrics.
\newblock {\em Algorithmica}, 83(11):3319--3337,
\newblock 2021, \href {http://dx.doi.org/10.1007/S00453-021-00843-6}
  {\path{doi:10.1007/S00453-021-00843-6}}.

\bibitem[EN22]{EN22}
M.~Elkin and O.~Neiman.
\newblock Lossless prioritized embeddings.
\newblock {\em {SIAM} J. Discret. Math.}, 36(3):1529--1550,
\newblock 2022, \href {http://dx.doi.org/10.1137/21M1436221}
  {\path{doi:10.1137/21M1436221}}.

\bibitem[Erd64]{Erdos64}
P.~Erd\H{o}s.
\newblock Extremal problems in graph theory.
\newblock {\em Theory of Graphs and Its Applications (Proc. Sympos.
  Smolenice)}, pages 29--36, 1964.
\newblock
\newblock see
  \href{http://citeseerx.ist.psu.edu/viewdoc/summary?doi=10.1.1.210.7240}{here}.

\bibitem[ES15]{ES15ShallowLight}
M.~Elkin and S.~Solomon.
\newblock Steiner shallow-light trees are exponentially lighter than spanning
  ones.
\newblock {\em {SIAM} J. Comput.}, 44(4):996--1025,
\newblock 2015, \href {http://dx.doi.org/10.1137/13094791X}
  {\path{doi:10.1137/13094791X}}.

\bibitem[FGK24]{FGK24}
A.~Filtser, L.~Gottlieb, and R.~Krauthgamer.
\newblock Labelings vs. embeddings: On distributed and prioritized
  representations of distances.
\newblock {\em Discret. Comput. Geom.}, 71(3):849--871,
\newblock 2024, \href {http://dx.doi.org/10.1007/S00454-023-00565-2}
  {\path{doi:10.1007/S00454-023-00565-2}}.

\bibitem[FGN24]{FGN24SoCG}
A.~Filtser, Y.~Gitlitz, and O.~Neiman.
\newblock Light, reliable spanners.
\newblock In W.~Mulzer and J.~M. Phillips, editors, {\em 40th International
  Symposium on Computational Geometry, SoCG 2024, Athens, Greece, June 11-14,
  2024}, volume 293 of {\em LIPIcs}, pages 56:1--56:15. Schloss Dagstuhl -
  Leibniz-Zentrum f{\"{u}}r Informatik,
\newblock 2024, \href {http://dx.doi.org/10.4230/LIPICS.SOCG.2024.56}
  {\path{doi:10.4230/LIPICS.SOCG.2024.56}}.

\bibitem[FGNW17]{FGNW17}
O.~Freedman, P.~Gawrychowski, P.~K. Nicholson, and O.~Weimann.
\newblock Optimal distance labeling schemes for trees.
\newblock In {\em Proceedings of the {ACM} Symposium on Principles of
  Distributed Computing, {PODC} 2017, Washington, DC, USA, July 25-27, 2017},
  pages 185--194,
\newblock 2017, \href {http://dx.doi.org/10.1145/3087801.3087804}
  {\path{doi:10.1145/3087801.3087804}}.

\bibitem[Fil19a]{Fil19padded}
A.~Filtser.
\newblock On strong diameter padded decompositions.
\newblock In {\em Approximation, Randomization, and Combinatorial Optimization.
  Algorithms and Techniques, {APPROX/RANDOM} 2019, September 20-22, 2019,
  Massachusetts Institute of Technology, Cambridge, MA, {USA}}, pages
  6:1--6:21,
\newblock 2019, \href {http://dx.doi.org/10.4230/LIPIcs.APPROX-RANDOM.2019.6}
  {\path{doi:10.4230/LIPIcs.APPROX-RANDOM.2019.6}}.

\bibitem[Fil19b]{Fil19SPR}
A.~Filtser.
\newblock Steiner point removal with distortion ${O}(log k)$ using the
  relaxed-voronoi algorithm.
\newblock {\em {SIAM} J. Comput.}, 48(2):249--278, 2019.
\newblock
\newblock preliminary version published in SODA 2018, \href
  {http://dx.doi.org/10.1137/18M1184400} {\path{doi:10.1137/18M1184400}}.

\bibitem[Fil21]{Fil21}
A.~Filtser.
\newblock Hop-constrained metric embeddings and their applications.
\newblock In {\em 62nd {IEEE} Annual Symposium on Foundations of Computer
  Science, {FOCS} 2021, Denver, CO, USA, February 7-10, 2022}, pages 492--503.
  {IEEE},
\newblock 2021, \href {http://dx.doi.org/10.1109/FOCS52979.2021.00056}
  {\path{doi:10.1109/FOCS52979.2021.00056}}.

\bibitem[Fil23]{Fil23}
A.~Filtser.
\newblock Labeled nearest neighbor search and metric spanners via locality
  sensitive orderings.
\newblock In E.~W. Chambers and J.~Gudmundsson, editors, {\em 39th
  International Symposium on Computational Geometry, SoCG 2023, June 12-15,
  2023, Dallas, Texas, {USA}}, volume 258 of {\em LIPIcs}, pages 33:1--33:18.
  Schloss Dagstuhl - Leibniz-Zentrum f{\"{u}}r Informatik,
\newblock 2023, \href {http://dx.doi.org/10.4230/LIPICS.SOCG.2023.33}
  {\path{doi:10.4230/LIPICS.SOCG.2023.33}}.

\bibitem[Fil24]{Fil24Scattering}
A.~Filtser.
\newblock Scattering and sparse partitions, and their applications.
\newblock {\em {ACM} Trans. Algorithms}, 20(4):30:1--30:42,
\newblock 2024, \href {http://dx.doi.org/10.1145/3672562}
  {\path{doi:10.1145/3672562}}.

\bibitem[Fil25]{Fil25}
A.~Filtser.
\newblock On sparse covers of minor free graphs, low dimensional metric
  embeddings, and other applications.
\newblock In O.~Aichholzer and H.~Wang, editors, {\em 41st International
  Symposium on Computational Geometry, SoCG 2025, Kanazawa, Japan, June 23-27,
  2025}, volume 332 of {\em LIPIcs}, pages 49:1--49:16. Schloss Dagstuhl -
  Leibniz-Zentrum f{\"{u}}r Informatik,
\newblock 2025, \href {http://dx.doi.org/10.4230/LIPICS.SOCG.2025.49}
  {\path{doi:10.4230/LIPICS.SOCG.2025.49}}.

\bibitem[FKS19]{FKS19}
E.~Fox{-}Epstein, P.~N. Klein, and A.~Schild.
\newblock Embedding planar graphs into low-treewidth graphs with applications
  to efficient approximation schemes for metric problems.
\newblock In {\em Proceedings of the 30th Annual ACM-SIAM Symposium on Discrete
  Algorithms}, SODA `19, page 1069–1088,
\newblock 2019, \href {http://dx.doi.org/10.1137/1.9781611975482.66}
  {\path{doi:10.1137/1.9781611975482.66}}.

\bibitem[FKT19]{FKT19}
A.~Filtser, R.~Krauthgamer, and O.~Trabelsi.
\newblock Relaxed voronoi: {A} simple framework for terminal-clustering
  problems.
\newblock In {\em 2nd Symposium on Simplicity in Algorithms, SOSA@SODA 2019,
  January 8-9, 2019 - San Diego, CA, {USA}}, pages 10:1--10:14,
\newblock 2019, \href {http://dx.doi.org/10.4230/OASIcs.SOSA.2019.10}
  {\path{doi:10.4230/OASIcs.SOSA.2019.10}}.

\bibitem[FL21]{FL21}
A.~Filtser and H.~Le.
\newblock Clan embeddings into trees, and low treewidth graphs.
\newblock In {\em {STOC} '21: 53rd Annual {ACM} {SIGACT} Symposium on Theory of
  Computing, Virtual Event, Italy, June 21-25, 2021}, pages 342--355,
\newblock 2021, \href {http://dx.doi.org/10.1145/3406325.3451043}
  {\path{doi:10.1145/3406325.3451043}}.

\bibitem[FL22a]{FL22}
A.~Filtser and H.~Le.
\newblock Locality-sensitive orderings and applications to reliable spanners.
\newblock In S.~Leonardi and A.~Gupta, editors, {\em {STOC} '22: 54th Annual
  {ACM} {SIGACT} Symposium on Theory of Computing, Rome, Italy, June 20 - 24,
  2022}, pages 1066--1079. {ACM},
\newblock 2022, \href {http://dx.doi.org/10.1145/3519935.3520042}
  {\path{doi:10.1145/3519935.3520042}}.

\bibitem[FL22b]{FL22tw}
A.~Filtser and H.~Le.
\newblock Low treewidth embeddings of planar and minor-free metrics.
\newblock In {\em 63rd {IEEE} Annual Symposium on Foundations of Computer
  Science, {FOCS} 2022, Denver, CO, USA, October 31 - November 3, 2022}, pages
  1081--1092. {IEEE},
\newblock 2022, \href {http://dx.doi.org/10.1109/FOCS54457.2022.00105}
  {\path{doi:10.1109/FOCS54457.2022.00105}}.

\bibitem[FN22]{FN22}
A.~Filtser and O.~Neiman.
\newblock Light spanners for high dimensional norms via stochastic
  decompositions.
\newblock {\em Algorithmica}, 84(10):2987--3007,
\newblock 2022, \href {http://dx.doi.org/10.1007/S00453-022-00994-0}
  {\path{doi:10.1007/S00453-022-00994-0}}.

\bibitem[FRT04]{FRT04}
J.~Fakcharoenphol, S.~Rao, and K.~Talwar.
\newblock A tight bound on approximating arbitrary metrics by tree metrics.
\newblock {\em J. Comput. Syst. Sci.}, 69(3):485--497, November 2004.
\newblock
\newblock preliminary version published in STOC 2003, \href
  {http://dx.doi.org/10.1016/j.jcss.2004.04.011}
  {\path{doi:10.1016/j.jcss.2004.04.011}}.

\bibitem[GHZ23]{GHZ23}
M.~Ghaffari, B.~Haeupler, and G.~Zuzic.
\newblock Hop-constrained oblivious routing.
\newblock {\em SIAM Journal on Computing}, pages STOC21--1--STOC21--25,
\newblock 2023, \href {http://arxiv.org/abs/https://doi.org/10.1137/21M1443467}
  {\path{arXiv:https://doi.org/10.1137/21M1443467}}, \href
  {http://dx.doi.org/10.1137/21M1443467} {\path{doi:10.1137/21M1443467}}.

\bibitem[GKK{\etalchar{+}}01]{GKKPP01}
C.~Gavoille, M.~Katz, N.~A. Katz, C.~Paul, and D.~Peleg.
\newblock Approximate distance labeling schemes.
\newblock In {\em Algorithms - {ESA} 2001, 9th Annual European Symposium,
  Aarhus, Denmark, August 28-31, 2001, Proceedings}, pages 476--487,
\newblock 2001, \href {http://dx.doi.org/10.1007/3-540-44676-1\_40}
  {\path{doi:10.1007/3-540-44676-1\_40}}.

\bibitem[GKK17]{GKK17}
L.~Gottlieb, A.~Kontorovich, and R.~Krauthgamer.
\newblock Efficient regression in metric spaces via approximate lipschitz
  extension.
\newblock {\em {IEEE} Trans. Inf. Theory}, 63(8):4838--4849, 2017.
\newblock
\newblock preliminary version published in SIMBAD 2013, \href
  {http://dx.doi.org/10.1109/TIT.2017.2713820}
  {\path{doi:10.1109/TIT.2017.2713820}}.

\bibitem[GKR00]{GKR00}
N.~Garg, G.~Konjevod, and R.~Ravi.
\newblock A polylogarithmic approximation algorithm for the group steiner tree
  problem.
\newblock {\em J. Algorithms}, 37(1):66--84,
\newblock 2000, \href {http://dx.doi.org/10.1006/jagm.2000.1096}
  {\path{doi:10.1006/jagm.2000.1096}}.

\bibitem[GM03]{GM03}
L.~E.~N. Gouveia and T.~L. Magnanti.
\newblock Network flow models for designing diameter-constrained
  minimum-spanning and steiner trees.
\newblock {\em Networks}, 41(3):159--173,
\newblock 2003, \href {http://dx.doi.org/10.1002/net.10069}
  {\path{doi:10.1002/net.10069}}.

\bibitem[Gou95]{Gou95}
L.~E.~N. Gouveia.
\newblock Using the miller-tucker-zemlin constraints to formulate a minimal
  spanning tree problem with hop constraints.
\newblock {\em Comput. Oper. Res.}, 22(9):959--970,
\newblock 1995, \href {http://dx.doi.org/10.1016/0305-0548(94)00074-I}
  {\path{doi:10.1016/0305-0548(94)00074-I}}.

\bibitem[Gou96]{Gou96}
L.~Gouveia.
\newblock Multicommodity flow models for spanning trees with hop constraints.
\newblock {\em European Journal of Operational Research}, 95(1):178--190,
\newblock 1996, \href
  {http://dx.doi.org/https://doi.org/10.1016/0377-2217(95)00090-9}
  {\path{doi:https://doi.org/10.1016/0377-2217(95)00090-9}}.

\bibitem[GPdSV03]{GPSV03}
L.~E.~N. Gouveia, P.~Patr{\'{\i}}cio, A.~de~Sousa, and R.~Valadas.
\newblock {MPLS} over {WDM} network design with packet level qos constraints
  based on {ILP} models.
\newblock In {\em Proceedings {IEEE} {INFOCOM} 2003, The 22nd Annual Joint
  Conference of the {IEEE} Computer and Communications Societies, San Franciso,
  CA, USA, March 30 - April 3, 2003}, pages 576--586,
\newblock 2003, \href {http://dx.doi.org/10.1109/INFCOM.2003.1208708}
  {\path{doi:10.1109/INFCOM.2003.1208708}}.

\bibitem[GPPR04]{GPPR04}
C.~Gavoille, D.~Peleg, S.~P{\'{e}}rennes, and R.~Raz.
\newblock Distance labeling in graphs.
\newblock {\em J. Algorithms}, 53(1):85--112, 2004.
\newblock
\newblock preliminary version published in SODA 2001, \href
  {http://dx.doi.org/10.1016/j.jalgor.2004.05.002}
  {\path{doi:10.1016/j.jalgor.2004.05.002}}.

\bibitem[GR01]{GR01}
L.~E.~N. Gouveia and C.~Requejo.
\newblock A new lagrangean relaxation approach for the hop-constrained minimum
  spanning tree problem.
\newblock {\em Eur. J. Oper. Res.}, 132(3):539--552,
\newblock 2001, \href {http://dx.doi.org/10.1016/S0377-2217(00)00143-0}
  {\path{doi:10.1016/S0377-2217(00)00143-0}}.

\bibitem[GRTU17]{GRTU17}
A.~Gupta, R.~Ravi, K.~Talwar, and S.~W. Umboh.
\newblock {LAST} but not least: Online spanners for buy-at-bulk.
\newblock In P.~N. Klein, editor, {\em Proceedings of the Twenty-Eighth Annual
  {ACM-SIAM} Symposium on Discrete Algorithms, {SODA} 2017, Barcelona, Spain,
  Hotel Porta Fira, January 16-19}, pages 589--599. {SIAM},
\newblock 2017, \href {http://dx.doi.org/10.1137/1.9781611974782.38}
  {\path{doi:10.1137/1.9781611974782.38}}.

\bibitem[Gup01]{Gupta01}
A.~Gupta.
\newblock Steiner points in tree metrics don't (really) help.
\newblock In {\em Proceedings of the Twelfth Annual Symposium on Discrete
  Algorithms, January 7-9, 2001, Washington, DC, {USA}}, pages 220--227, 2001.
\newblock
\newblock see \href{https://dl.acm.org/doi/10.5555/365411.365448}{here}.

\bibitem[HBK{\etalchar{+}}03]{HBKKW03}
E.~Halperin, J.~Buhler, R.~M. Karp, R.~Krauthgamer, and B.~Westover.
\newblock {Detecting protein sequence conservation via metric embeddings}.
\newblock {\em Bioinformatics}, 19(suppl 1):i122--i129,
\newblock 07 2003, \href
  {http://arxiv.org/abs/https://academic.oup.com/bioinformatics/article-pdf/19/suppl\_1/i122/614436/btg1016.pdf}
  {\path{arXiv:https://academic.oup.com/bioinformatics/article-pdf/19/suppl\_1/i122/614436/btg1016.pdf}},
  \href {http://dx.doi.org/10.1093/bioinformatics/btg1016}
  {\path{doi:10.1093/bioinformatics/btg1016}}.

\bibitem[Her22]{Hershkowitz22}
D.~E. Hershkowitz.
\newblock private communication,
\newblock January 2022.

\bibitem[HHZ21a]{HHZ21NonHop}
B.~Haeupler, D.~E. Hershkowitz, and G.~Zuzic.
\newblock Deterministic tree embeddings with copies for algorithms against
  adaptive adversaries.
\newblock {\em CoRR}, abs/2102.05168,
\newblock 2021, \href {http://arxiv.org/abs/2102.05168}
  {\path{arXiv:2102.05168}}.

\bibitem[HHZ21b]{HHZ21}
B.~Haeupler, D.~E. Hershkowitz, and G.~Zuzic.
\newblock Tree embeddings for hop-constrained network design.
\newblock In {\em {STOC} '21: 53rd Annual {ACM} {SIGACT} Symposium on Theory of
  Computing, Virtual Event, Italy, June 21-25, 2021}, pages 356--369,
\newblock 2021, \href {http://dx.doi.org/10.1145/3406325.3451053}
  {\path{doi:10.1145/3406325.3451053}}.

\bibitem[HIS13]{HIS13}
S.~Har{-}Peled, P.~Indyk, and A.~Sidiropoulos.
\newblock Euclidean spanners in high dimensions.
\newblock In {\em Proceedings of the Twenty-Fourth Annual {ACM-SIAM} Symposium
  on Discrete Algorithms, {SODA} 2013, New Orleans, Louisiana, USA, January
  6-8, 2013}, pages 804--809,
\newblock 2013, \href {http://dx.doi.org/10.1137/1.9781611973105.57}
  {\path{doi:10.1137/1.9781611973105.57}}.

\bibitem[HKS09]{HKS09}
M.~T. Hajiaghayi, G.~Kortsarz, and M.~R. Salavatipour.
\newblock Approximating buy-at-bulk and shallow-light \emph{k}-steiner trees.
\newblock {\em Algorithmica}, 53(1):89--103, 2009.
\newblock
\newblock preliminary version published in APPROX-RANDOM 2006, \href
  {http://dx.doi.org/10.1007/s00453-007-9013-x}
  {\path{doi:10.1007/s00453-007-9013-x}}.

\bibitem[KKM{\etalchar{+}}12]{KKMPT12}
M.~Khan, F.~Kuhn, D.~Malkhi, G.~Pandurangan, and K.~Talwar.
\newblock Efficient distributed approximation algorithms via probabilistic tree
  embeddings.
\newblock {\em Distributed Comput.}, 25(3):189--205, 2012.
\newblock
\newblock preliminary version published in PODC 2008, \href
  {http://dx.doi.org/10.1007/s00446-012-0157-9}
  {\path{doi:10.1007/s00446-012-0157-9}}.

\bibitem[KLS05]{KLS05}
J.~K{\"{o}}nemann, A.~Levin, and A.~Sinha.
\newblock Approximating the degree-bounded minimum diameter spanning tree
  problem.
\newblock {\em Algorithmica}, 41(2):117--129, 2005.
\newblock
\newblock preliminary version published in RANDOM-APPROX 2003, \href
  {http://dx.doi.org/10.1007/s00453-004-1121-2}
  {\path{doi:10.1007/s00453-004-1121-2}}.

\bibitem[KP97]{KP97}
G.~Kortsarz and D.~Peleg.
\newblock Approximating shallow-light trees (extended abstract).
\newblock In {\em Proceedings of the Eighth Annual {ACM-SIAM} Symposium on
  Discrete Algorithms, 5-7 January 1997, New Orleans, Louisiana, {USA}}, pages
  103--110, 1997.
\newblock
\newblock see \href{https://dl.acm.org/doi/10.5555/314161.314191}{here}.

\bibitem[KP09]{KP09}
E.~Kantor and D.~Peleg.
\newblock Approximate hierarchical facility location and applications to the
  bounded depth steiner tree and range assignment problems.
\newblock {\em J. Discrete Algorithms}, 7(3):341--362,
\newblock 2009, \href {http://dx.doi.org/10.1016/j.jda.2008.11.006}
  {\path{doi:10.1016/j.jda.2008.11.006}}.

\bibitem[KPR93]{KPR93}
P.~N. Klein, S.~A. Plotkin, and S.~Rao.
\newblock Excluded minors, network decomposition, and multicommodity flow.
\newblock In {\em Proceedings of the 25th Annual ACM Symposium on Theory of
  Computing}, STOC `93, page 682–690,
\newblock 1993, \href {http://dx.doi.org/10.1145/167088.167261}
  {\path{doi:10.1145/167088.167261}}.

\bibitem[KRY95]{KRY95}
S.~Khuller, B.~Raghavachari, and N.~E. Young.
\newblock Balancing minimum spanning trees and shortest-path trees.
\newblock {\em Algorithmica}, 14(4):305--321,
\newblock 1995, \href {http://dx.doi.org/10.1007/BF01294129}
  {\path{doi:10.1007/BF01294129}}.

\bibitem[KS16]{KS16}
M.~R. Khani and M.~R. Salavatipour.
\newblock Improved approximations for buy-at-bulk and shallow-light k-steiner
  trees and (k, 2)-subgraph.
\newblock {\em J. Comb. Optim.}, 31(2):669--685, 2016.
\newblock
\newblock preliminary version published in ISAAC 2011, \href
  {http://dx.doi.org/10.1007/s10878-014-9774-5}
  {\path{doi:10.1007/s10878-014-9774-5}}.

\bibitem[LCM99]{BCM99}
L.~J. LeBlanc, J.~Chifflet, and P.~Mahey.
\newblock Packet routing in telecommunication networks with path and flow
  restrictions.
\newblock {\em {INFORMS} J. Comput.}, 11(2):188--197,
\newblock 1999, \href {http://dx.doi.org/10.1287/ijoc.11.2.188}
  {\path{doi:10.1287/ijoc.11.2.188}}.

\bibitem[Lei16]{Lei16}
M.~Leitner.
\newblock Layered graph models and exact algorithms for the generalized
  hop-constrained minimum spanning tree problem.
\newblock {\em Comput. Oper. Res.}, 65:1--18,
\newblock 2016, \href {http://dx.doi.org/10.1016/j.cor.2015.06.012}
  {\path{doi:10.1016/j.cor.2015.06.012}}.

\bibitem[LLR95]{LLR95}
N.~Linial, E.~London, and Y.~Rabinovich.
\newblock The geometry of graphs and some of its algorithmic applications.
\newblock {\em Comb.}, 15(2):215--245, 1995.
\newblock
\newblock preliminary version published in FOCS 1994, \href
  {http://dx.doi.org/10.1007/BF01200757} {\path{doi:10.1007/BF01200757}}.

\bibitem[LMS23]{LMS23}
H.~Le, L.~Milenkovic, and S.~Solomon.
\newblock Sparse euclidean spanners with optimal diameter: {A} general and
  robust lower bound via a concave inverse-ackermann function.
\newblock In E.~W. Chambers and J.~Gudmundsson, editors, {\em 39th
  International Symposium on Computational Geometry, SoCG 2023, June 12-15,
  2023, Dallas, Texas, {USA}}, volume 258 of {\em LIPIcs}, pages 47:1--47:17.
  Schloss Dagstuhl - Leibniz-Zentrum f{\"{u}}r Informatik,
\newblock 2023, \href {http://dx.doi.org/10.4230/LIPICS.SOCG.2023.47}
  {\path{doi:10.4230/LIPICS.SOCG.2023.47}}.

\bibitem[Mat96]{Mat96}
J.~Matou{\v{s}}ek.
\newblock On the distortion required for embedding finite metric spaces into
  normed spaces.
\newblock {\em Israel Journal of Mathematics}, 93(1):333--344,
\newblock 1996, \href {http://dx.doi.org/10.1007/BF02761110}
  {\path{doi:10.1007/BF02761110}}.

\bibitem[MMMR18]{MMMR18}
S.~Mahabadi, K.~Makarychev, Y.~Makarychev, and I.~P. Razenshteyn.
\newblock Nonlinear dimension reduction via outer bi-lipschitz extensions.
\newblock In I.~Diakonikolas, D.~Kempe, and M.~Henzinger, editors, {\em
  Proceedings of the 50th Annual {ACM} {SIGACT} Symposium on Theory of
  Computing, {STOC} 2018, Los Angeles, CA, USA, June 25-29, 2018}, pages
  1088--1101. {ACM},
\newblock 2018, \href {http://dx.doi.org/10.1145/3188745.3188828}
  {\path{doi:10.1145/3188745.3188828}}.

\bibitem[MMP08]{MMP08}
A.~Meyerson, K.~Munagala, and S.~A. Plotkin.
\newblock Cost-distance: Two metric network design.
\newblock {\em {SIAM} J. Comput.}, 38(4):1648--1659,
\newblock 2008, \href {http://dx.doi.org/10.1137/050629665}
  {\path{doi:10.1137/050629665}}.

\bibitem[MN07]{MN07}
M.~Mendel and A.~Naor.
\newblock Ramsey partitions and proximity data structures.
\newblock {\em Journal of the European Mathematical Society}, 9(2):253--275,
  2007.
\newblock
\newblock preliminary version published in FOCS 2006, \href
  {http://dx.doi.org/10.4171/JEMS/79} {\path{doi:10.4171/JEMS/79}}.

\bibitem[MRS{\etalchar{+}}98]{MRSRRH98}
M.~V. Marathe, R.~Ravi, R.~Sundaram, S.~S. Ravi, D.~J. Rosenkrantz, and
  H.~B.~H. III.
\newblock Bicriteria network design problems.
\newblock {\em J. Algorithms}, 28(1):142--171, 1998.
\newblock
\newblock preliminary version published in ICALP 1995, \href
  {http://dx.doi.org/10.1006/jagm.1998.0930}
  {\path{doi:10.1006/jagm.1998.0930}}.

\bibitem[NN19]{NN19}
S.~Narayanan and J.~Nelson.
\newblock Optimal terminal dimensionality reduction in euclidean space.
\newblock In M.~Charikar and E.~Cohen, editors, {\em Proceedings of the 51st
  Annual {ACM} {SIGACT} Symposium on Theory of Computing, {STOC} 2019, Phoenix,
  AZ, USA, June 23-26, 2019}, pages 1064--1069. {ACM},
\newblock 2019, \href {http://dx.doi.org/10.1145/3313276.3316307}
  {\path{doi:10.1145/3313276.3316307}}.

\bibitem[NPS11]{NPS11}
J.~Naor, D.~Panigrahi, and M.~Singh.
\newblock Online node-weighted steiner tree and related problems.
\newblock In {\em {IEEE} 52nd Annual Symposium on Foundations of Computer
  Science, {FOCS} 2011, Palm Springs, CA, USA, October 22-25, 2011}, pages
  210--219,
\newblock 2011, \href {http://dx.doi.org/10.1109/FOCS.2011.65}
  {\path{doi:10.1109/FOCS.2011.65}}.

\bibitem[NT12]{NT12}
A.~Naor and T.~Tao.
\newblock Scale-oblivious metric fragmentation and the nonlinear dvoretzky
  theorem.
\newblock {\em Israel Journal of Mathematics}, 192(1):489--504,
\newblock 2012, \href {http://dx.doi.org/10.1007/s11856-012-0039-7}
  {\path{doi:10.1007/s11856-012-0039-7}}.

\bibitem[OR94]{OR94}
M.~J. Osborne and A.~Rubinstein.
\newblock {\em A Course in Game Theory}.
\newblock MIT Press,
\newblock 1994.

\bibitem[Pel00]{Peleg00Labling}
D.~Peleg.
\newblock Proximity-preserving labeling schemes.
\newblock {\em J. Graph Theory}, 33(3):167--176, 2000.
\newblock
\newblock preliminary version published in WG 1999, \href
  {http://dx.doi.org/10.1002/(SICI)1097-0118(200003)33:3<167::AID-JGT7>3.0.CO;2-5}
  {\path{doi:10.1002/(SICI)1097-0118(200003)33:3<167::AID-JGT7>3.0.CO;2-5}}.

\bibitem[PS03]{PS03}
H.~Pirkul and S.~Soni.
\newblock New formulations and solution procedures for the hop constrained
  network design problem.
\newblock {\em Eur. J. Oper. Res.}, 148(1):126--140,
\newblock 2003, \href {http://dx.doi.org/10.1016/S0377-2217(02)00366-1}
  {\path{doi:10.1016/S0377-2217(02)00366-1}}.

\bibitem[PU89]{PU89}
D.~Peleg and E.~Upfal.
\newblock A trade-off between space and efficiency for routing tables.
\newblock {\em J. {ACM}}, 36(3):510--530,
\newblock 1989, \href {http://dx.doi.org/10.1145/65950.65953}
  {\path{doi:10.1145/65950.65953}}.

\bibitem[RAJ12]{RAJ12}
A.~Rossi, A.~Aubry, and M.~Jacomino.
\newblock Connectivity-and-hop-constrained design of electricity distribution
  networks.
\newblock {\em Eur. J. Oper. Res.}, 218(1):48--57,
\newblock 2012, \href {http://dx.doi.org/10.1016/j.ejor.2011.10.006}
  {\path{doi:10.1016/j.ejor.2011.10.006}}.

\bibitem[Rav94]{Rav94}
R.~Ravi.
\newblock Rapid rumor ramification: Approximating the minimum broadcast time
  (extended abstract).
\newblock In {\em 35th Annual Symposium on Foundations of Computer Science,
  Santa Fe, New Mexico, USA, 20-22 November 1994}, pages 202--213,
\newblock 1994, \href {http://dx.doi.org/10.1109/SFCS.1994.365693}
  {\path{doi:10.1109/SFCS.1994.365693}}.

\bibitem[RTZ05]{RTZ05}
L.~Roditty, M.~Thorup, and U.~Zwick.
\newblock Deterministic constructions of approximate distance oracles and
  spanners.
\newblock In {\em Proceedings of the 32Nd International Conference on Automata,
  Languages and Programming}, ICALP'05, pages 261--272, Berlin, Heidelberg,
  2005.
\newblock Springer-Verlag, \href {http://dx.doi.org/10.1007/11523468_22}
  {\path{doi:10.1007/11523468_22}}.

\bibitem[SCRS01]{SCRS01}
F.~S. Salman, J.~Cheriyan, R.~Ravi, and S.~Subramanian.
\newblock Approximating the single-sink link-installation problem in network
  design.
\newblock {\em {SIAM} Journal on Optimization}, 11(3):595--610,
\newblock 2001, \href {http://dx.doi.org/10.1137/S1052623497321432}
  {\path{doi:10.1137/S1052623497321432}}.

\bibitem[Sol13]{Sol13}
S.~Solomon.
\newblock Sparse euclidean spanners with tiny diameter.
\newblock {\em {ACM} Trans. Algorithms}, 9(3):28:1--28:33, 2013.
\newblock
\newblock preliminary version published in SODA 2011, \href
  {http://dx.doi.org/10.1145/2483699.2483708}
  {\path{doi:10.1145/2483699.2483708}}.

\bibitem[TCG15]{TCG15}
B.~Thiongane, J.~Cordeau, and B.~Gendron.
\newblock Formulations for the nonbifurcated hop-constrained multicommodity
  capacitated fixed-charge network design problem.
\newblock {\em Comput. Oper. Res.}, 53:1--8,
\newblock 2015, \href {http://dx.doi.org/10.1016/j.cor.2014.07.013}
  {\path{doi:10.1016/j.cor.2014.07.013}}.

\bibitem[TZ01]{TZ01b}
M.~Thorup and U.~Zwick.
\newblock Compact routing schemes.
\newblock In {\em Proceedings of the Thirteenth Annual {ACM} Symposium on
  Parallel Algorithms and Architectures, {SPAA} 2001, Heraklion, Crete Island,
  Greece, July 4-6, 2001}, pages 1--10,
\newblock 2001, \href {http://dx.doi.org/10.1145/378580.378581}
  {\path{doi:10.1145/378580.378581}}.

\bibitem[TZ05]{TZ05}
M.~Thorup and U.~Zwick.
\newblock Approximate distance oracles.
\newblock {\em J. {ACM}}, 52(1):1--24,
\newblock 2005, \href {http://dx.doi.org/10.1145/1044731.1044732}
  {\path{doi:10.1145/1044731.1044732}}.

\bibitem[Vo{\ss}99]{Voss99}
S.~Vo{\ss}.
\newblock The steiner tree problem with hop constraints.
\newblock {\em Ann. Oper. Res.}, 86:321--345,
\newblock 1999, \href {http://dx.doi.org/10.1023/A\%3A1018967121276}
  {\path{doi:10.1023/A\%3A1018967121276}}.

\bibitem[WA88]{WA88}
K.~A. Woolston and S.~L. Albin.
\newblock The design of centralized networks with reliability and availability
  constraints.
\newblock {\em Comput. Oper. Res.}, 15(3):207--217,
\newblock 1988, \href {http://dx.doi.org/10.1016/0305-0548(88)90033-0}
  {\path{doi:10.1016/0305-0548(88)90033-0}}.

\bibitem[Wul13]{W13}
C.~Wulff{-}Nilsen.
\newblock Approximate distance oracles with improved query time.
\newblock In S.~Khanna, editor, {\em Proceedings of the Twenty-Fourth Annual
  {ACM-SIAM} Symposium on Discrete Algorithms, {SODA} 2013, New Orleans,
  Louisiana, USA, January 6-8, 2013}, pages 539--549. {SIAM},
\newblock 2013, \href {http://dx.doi.org/10.1137/1.9781611973105.39}
  {\path{doi:10.1137/1.9781611973105.39}}.

\end{thebibliography}
}
\end{document}